\documentclass[11pt]{article}
\usepackage[usenames,dvipsnames,svgnames,table]{xcolor}
\usepackage{enumitem}
\usepackage{graphicx}
\usepackage{fancybox}
\usepackage{comment}
\usepackage{xcolor}
\usepackage{nameref}
\usepackage{bbm}
\usepackage[T1]{fontenc}

\definecolor{ForestGreen}{rgb}{0.1333,0.5451,0.1333}
\definecolor{DarkRed}{rgb}{0.8,0,0}
\definecolor{Red}{rgb}{1,0,0}
\usepackage[linktocpage=true,
pagebackref=true,colorlinks,
linkcolor=ForestGreen,citecolor=ForestGreen,
bookmarks,bookmarksopen,bookmarksnumbered]
{hyperref}
\usepackage[nottoc]{tocbibind}

\usepackage{amsmath}
\usepackage{amssymb}
\usepackage{amsthm}

\usepackage[ruled, noend, linesnumbered]{algorithm2e}

\usepackage{subfig}

\usepackage{thm-restate}

\usepackage{float}
\usepackage{cleveref}

\newtheorem{theorem}{Theorem}[section]
\newtheorem{informaltheorem}[theorem]{Informal Theorem}

\newtheorem{corollary}[theorem]{Corollary}
\newtheorem{lemma}[theorem]{Lemma}

\newtheorem{claim}[theorem]{Claim}

\newtheorem{fact}[theorem]{Fact}

\newtheorem{definition}[theorem]{Definition}
\newtheorem{informaldefinition}[theorem]{Informal Definition}
\newtheorem{remark}[theorem]{Remark}

\newtheorem*{theorem*}{Theorem}
\newtheorem*{corollary*}{Corollary}
\newtheorem*{conjecture*}{Conjecture}
\newtheorem*{lemma*}{Lemma}
\newtheorem*{thm*}{Theorem}
\newtheorem*{prop*}{Proposition}
\newtheorem*{obs*}{Observation}
\newtheorem*{definition*}{Definition}

\newtheorem*{remark*}{Remark}
\newtheorem*{rec*}{Recommendation}

\newenvironment{fminipage}%
  {\begin{Sbox}\begin{minipage}}%
  {\end{minipage}\end{Sbox}\fbox{\TheSbox}}

\def\defeq{\stackrel{\mathrm{def}}{=}}

\def\floor#1{\left\lfloor #1 \right\rfloor}

\def\norm#1{\left\| #1 \right\|}

\DeclareMathOperator{\dist}{dist}

\def\calA{\mathcal{A}}

\def\calD{\mathcal{D}}

\def\calG{\mathcal{G}}

\def\calT{\mathcal{T}}

\def\calP{\mathcal{P}}

\def\gammasnc{\gamma_{\text{SNC}}}

\newcommand\DDelta{\boldsymbol{\mathit{\Delta}}}

\newcommand\PPi{\boldsymbol{\Pi}}

\def\aa{\pmb{\mathit{a}}}
\newcommand\bb{\boldsymbol{\mathit{b}}}
\newcommand\cc{\boldsymbol{\mathit{c}}}
\newcommand\dd{\boldsymbol{\mathit{d}}}

\newcommand\ff{\boldsymbol{\mathit{f}}}
\renewcommand\gg{\boldsymbol{\mathit{g}}}

\renewcommand\ll{\boldsymbol{\mathit{l}}}
\newcommand\pp{\boldsymbol{\mathit{p}}}

\newcommand\uu{\boldsymbol{\mathit{u}}}
\newcommand\vv{\boldsymbol{\mathit{v}}}

\newcommand\veczero{\boldsymbol{0}}
\newcommand\vecone{\boldsymbol{1}}

\renewcommand{\deg}{\operatorname{deg}}

\renewcommand\AA{\boldsymbol{\mathit{A}}}
\newcommand\BB{\boldsymbol{\mathit{B}}}

\newcommand\II{\boldsymbol{\mathit{I}}}

\newcommand\MM{\boldsymbol{\mathit{M}}}
\newcommand\LL{\boldsymbol{\mathit{L}}}
\newcommand\PP{\boldsymbol{\mathit{P}}}

\newcommand\Otil{\widetilde{O}}
\renewcommand\O{\widetilde{O}}

\newcommand\R{\mathbb{R}}

\newcommand{\trp}{\top}

\newcommand{\proj}{\PPi}

\DeclareMathOperator*{\argmin}{arg\,min}

\DeclareMathOperator*{\diag}{diag}

\newcommand{\concat}{\oplus}
\newcommand{\econg}{\mathrm{econg}}
\newcommand{\vcong}{\mathrm{vcong}}

\newcommand{\diam}{\mathrm{diam}}

\newcommand\Z{\mathbb{Z}}

\newcommand{\eps}{\epsilon}
\renewcommand{\O}{\widetilde{O}}

\renewcommand{\l}{\langle}
\renewcommand{\r}{\rangle}
\newcommand{\assign}{\leftarrow}

\renewcommand{\forall}{\mathrm{\text{ for all }}}

\newcommand{\length}{\mathrm{length}}

\newcommand{\bg}{\boldsymbol{g}}

\newcommand{\bell}{\boldsymbol{\ell}}

\newcommand{\bDelta}{\boldsymbol{\Delta}}

\renewcommand{\hat}{\widehat}
\renewcommand{\tilde}{\widetilde}

\DeclareFontFamily{U}{mathb}{\hyphenchar\font45}
\DeclareFontShape{U}{mathb}{m}{n}{<5> <6> <7> <8> <9> <10> gen * mathb
<10.95> mathb10 <12> <14.4> <17.28> <20.74> <24.88> mathb12}{}
\DeclareSymbolFont{mathb}{U}{mathb}{m}{n}
\DeclareMathSymbol{\rcirclearrow}{\mathbin}{mathb}{'367}

\newcommand{\wh}{\widehat}

\renewcommand{\bar}{\overline}

\newif\ifrandom
\randomtrue

\newcommand{\rev}{\mathsf{rev}}

\newcommand{\last}{\mathsf{last}}

\renewcommand{\l}{\langle}
\renewcommand{\r}{\rangle}

\newcommand{\Abs}[1]{\left|#1\right|}

\newcommand{\todolater}[1]{}

\newcommand{\cA}{\mathcal{A}}
\newcommand{\cE}{\mathcal{E}}
\newcommand{\tcE}{\widetilde{\mathcal{E}}}
\newcommand{\cEout}{\mathcal{E}^{\mathrm{out}}}
\newcommand{\cEin}{\mathcal{E}^{\mathrm{in}}}

\newcommand{\cW}{\mathcal{W}}

\newcommand{\tcycle}{T_{\circlearrowleft}}
\DeclareUnicodeCharacter{2113}{$\ell$}

\newcommand{\gammadecrsnc}{\gamma_{decrSNC}}

\usepackage{fullpage}

\DeclareMathOperator{\master}{master}
\DeclareMathOperator{\eproj}{proj}

\def\gammahrg{\gamma_{\text{HRG}}}
\def\kappahrg{\kappa_{\text{HRG}}}
\def\gammadiam{\gamma_{\text{diam}}}

\def\TCyc{T_\circlearrowleft}

\title{
Almost-Linear Time Algorithms for Incremental Graphs:
\\
\mbox{Cycle Detection, SCCs, $s$-$t$ Shortest Path,
and Minimum-Cost Flow}}

\date{}
\newcommand*\samethanks[1][\value{footnote}]{\footnotemark[#1]}
\author{Li Chen\thanks{Li Chen was supported by NSF Grant CCF-2330255.} \\ Carnegie Mellon University \\ lichenntu@gmail.com
\and  Rasmus Kyng\thanks{The research leading to these results has received funding from grant no. 200021 204787 of the Swiss National Science Foundation.} \\ ETH Zurich \\ kyng@inf.ethz.ch \\  
\and Yang P. Liu\thanks{This material is based upon work supported by the National Science Foundation 
under Grant No. DMS-1926686.} \\ Institute for Advanced Study \\ yangpliu@ias.edu \\
\and Simon Meierhans\samethanks[2] \\ ETH Zurich \\ mesimon@inf.ethz.ch \and Maximilian Probst Gutenberg\samethanks[2] \\ ETH Zurich \\ maximilian.probst@inf.ethz.ch}

\begin{document}

\maketitle

\begin{abstract}
We give the first almost-linear time algorithms for several problems in incremental graphs
including cycle detection, strongly connected component maintenance, $s$-$t$ shortest path, maximum flow, and minimum-cost flow.
To solve these problems, we give a deterministic data structure that returns a $m^{o(1)}$-approximate minimum-ratio cycle in fully dynamic graphs in amortized $m^{o(1)}$ time per update. 
Combining this with the interior point method framework of Brand-Liu-Sidford (STOC 2023) gives the first almost-linear time algorithm for deciding the first update in an incremental graph after which the cost of the minimum-cost flow attains value at most some given threshold $F$.
By rather direct reductions to minimum-cost flow, we are then able to solve the problems in incremental graphs mentioned above.

Our new data structure also leads to a modular and deterministic almost-linear time algorithm for minimum-cost flow by removing the need for complicated modeling of a restricted adversary, in contrast to recent randomized and deterministic algorithms for minimum-cost flow in 
Chen-Kyng-Liu-Peng-Probst\hspace{0.2em}Gutenberg-Sachdeva (FOCS 2022) \&
Brand-Chen-Kyng-Liu-Peng-Probst\hspace{0.2em}Gutenberg-Sachdeva-Sidford (FOCS 2023).

At a high level, our algorithm dynamizes the $\ell_1$ oblivious routing of Rozho\v{n}-Grunau-Haeupler-Zuzic-Li (STOC 2022), and develops a method to extract an approximate minimum ratio cycle from the structure of the oblivious routing. To maintain the oblivious routing, we use tools from concurrent work of Kyng-Meierhans-Probst\hspace{0.2em}Gutenberg which designed vertex sparsifiers for shortest paths, in order to maintain a sparse neighborhood cover in fully dynamic graphs.

To find a cycle, we first show that an approximate minimum ratio cycle can be represented as a fundamental cycle on a small set of trees resulting from the oblivious routing. Then, we find a cycle whose quality is comparable to the best tree cycle.
This final cycle query step involves vertex and edge sparsification procedures reminiscent of the techniques 
introduced in Chen-Kyng-Liu-Peng-Probst\hspace{0.2em}Gutenberg-Sachdeva (FOCS 2022), but crucially requires a more powerful dynamic spanner, which can handle far more edge insertions than prior work.
We build such a spanner via a construction that hearkens back to the classic greedy spanner algorithm of
Alth{\"o}fer-Das-Dobkin-Joseph-Soares
(Discrete \& Computational Geometry 1993).
\end{abstract}

\pagenumbering{gobble}

\pagebreak

\tableofcontents

\pagebreak

\pagenumbering{arabic}

\section{Introduction}
\label{sec:intro}
The goal of \emph{dynamic graph algorithms}\footnote{We use the terms \emph{dynamic algorithm} and \emph{data structure} interchangeably.} is to solve problems on a changing input graph, while minimizing the time required to compute and update solutions as the graph changes.
In this paper, we focus on the setting of \emph{incremental graphs}, where the input is a directed graph that undergoes edge insertions over time.
We essentially settle the complexity of several fundamental and long-studied problems on incremental directed graphs, by giving algorithms that process up to $m$ edge insertions while using $m^{1+o(1)}$ total update time.
In this setting, we obtain almost-linear total update time algorithms for cycle detection
\cite{MNR96,KB06,LC07,AFM08,AF10,BFG09,HKMST12,BFGT16,CFKR13,BC18,BK20},
maintaining strongly-connected components
\cite{BFGT16,BDP21},
and approximate or thresholded versions of 
$s$-$t$ shortest paths
\cite{bernstein2013maintaining,PVW20,chechik2021incremental, KMP22, das2022near},
weighted bipartite matching \cite{Gupta14,BKS23,BK23},
and maximum and min-cost flow \cite{GH22,vdBrand23incr, brand2023incremental}.
All our results are essentially achieved by reduction to one central problem, known as \emph{incremental thresholded minimum-cost flow} \cite{vdBrand23incr}, defined as follows.

\begin{definition}[Incremental thresholded minimum-cost flow]
\label{def:incrementalmincostflow}
The thresholded min-cost flow problem is defined on a directed graph $G = (V, E)$ with capacities $\uu \in \R^E$ and costs $\cc \in \R^E$, undergoing edge insertions, along with vertex demands $\dd \in \R^V$ and a threshold $F$.
A dynamic algorithm solves the problem if, after every update, the algorithm outputs whether the graph can support a feasible flow $\ff \in \R^E$ routing demand $\dd$ with cost $\cc^\top \ff$ at most $F$, or answers that no such flow exists yet.
\end{definition}
Note that the optimal cost is non-increasing, as we can give $0$ flow to the newly inserted edge, so there is a unique first time when such a flow is feasible, and it stays feasible afterwards.
Our main result is an almost linear time algorithm for incremental thresholded min-cost flow.
\begin{informaltheorem}
    There exists a deterministic algorithm that solves the thresholded min-cost flow problem with $m$ insertions in $m^{1+o(1)}$ total time, provided capacities and costs are polynomially bounded in $m$.
\end{informaltheorem}
From this result, we derive our other results on incremental graphs, except in the case of strongly-connected component maintenance, which requires a slightly different framework, but uses the same techniques and data structures.
All our results are deterministic.

Our result builds on and strengthens many recent results on minimum-cost flow.
Almost-linear time algorithms for minimum-cost flow were, after at least nine decades of research, obtained in \cite{chen2022maximum}.
A key element of their algorithm was the introduction of the dynamic min-ratio cycle problem, which is defined as follows.
\begin{informaldefinition}[$\alpha$-Approximate dynamic min-ratio cycle problem]
\label{def:informalminratio}
    The dynamic min-ratio cycle problem is defined on a directed graph $G = (V, E, \ll, \gg)$ with (undirected) lengths $\ll \in \R_{\ge0}^E$ and gradient $\gg \in \R^E$.
    At each time step, the gradient and length of a single edge may be updated.
    
    A dynamic algorithm solves the problem if, after the $i$-th update, it can identify a cycle $\cc_i \in \R^E$, 
    such that $\BB^T \cc_i = \veczero$, and
    \begin{align*}
        \frac{\l\bg, \cc_i \r}{\norm{\LL \cc_i}_1} \le \frac{1}{\alpha} \cdot \min_{\BB^\top \bDelta = 0} \frac{\l\bg, \bDelta\r}{\norm{\LL \bDelta}_1}
        .
    \end{align*}
\end{informaldefinition}
\cite{chen2022maximum} introduced an $\ell_1$-interior point method (IPM), which they showed reduces min-cost flow on a graph with $m$ edges to $m^{1+o(1)}$ update steps of solving the dynamic min-ratio problem defined above with approximation quality $\alpha = m^{o(1)}$. 
However, \cite{chen2022maximum} did \emph{not} manage to solve this problem against a general adversary.
Instead, they developed a randomized data structure that solves the problem against an \emph{oblivious} adversary,\footnote{An oblivious adversary is an adversary that does not depend on random choices made by the data structure}
and then they showed the $\ell_1$-IPM creates a sequence of problems that behaves almost like an oblivious adversary, and managed to adapt their data structure to this particular sequence.
This left an important open question: can approximate dynamic min-ratio cycle be solved against any adversary?
Solving this question with a deterministic data structure would immediately yield a deterministic algorithm for min-cost flow.

\cite{vdBrand23incr} raised the stakes of this question dramatically: they showed, by slightly modifying the \cite{chen2022maximum} $\ell_1$-IPM, that solving approximate dynamic min-ratio cycle across $m^{1+o(1)}$ update steps directly implies solving the incremental thresholded min-cost flow problem, and hence a host of important incremental graph problems.
In this setting, techniques from \cite{chen2022maximum} for solving dynamic min-ratio cycle against a restricted adversary fail.
Instead, \cite{vdBrand23incr} developed a randomized algorithm for solving this problem against a general adversary with total update time $m^{1+o(1)} \sqrt{n}$.

\cite{detMaxFlow} gave a deterministic algorithm for min-cost flow with almost-linear running time, but, remarkably, they still did not solve the approximate dynamic min-ratio cycle problem. Instead, they tailored a deterministic data structure to the specific update sequence generated by the $\ell_1$-IPM, similar to the randomized approach in \cite{chen2022maximum}. 
Very recently, 
\cite{brand2023incremental} gave the first algorithm to solve $(1+\eps)$-approximate maximum flow in \emph{undirected graphs} in time $m^{1+o(1)}\eps^{-3}$.
However, instead of attacking the dynamic min-ratio cycle problem head on, they 
showed that for undirected, approximate maximum flow, it can be circumvented: Rather than solving the problem with an $\ell_1$-IPM, they could rely on a multiplicative weight update method \cite{adil2019iterative}, reminiscent of first-order methods for approximate undirected maximum flow \cite{S13,KLOS14}.
This method leads to a monotone version of dynamic min-ratio cycle (where, for long periods, lengths are only increasing and the gradient is fixed), 
with a very well-behaved update sequence, and in this more tractable setting, the data structure of \cite{chen2022maximum}
works, even for the incremental problem.

In this paper, we finally solve the approximate dynamic min-ratio cycle problem with $m^{1+o(1)}$ total update time deterministically, across $m^{1+o(1)}$ update steps.
Thus, by the $\ell_1$-IPM of \cite{chen2022maximum}, we immediately get a much more modular, and conceptually simpler, deterministic almost-linear time algorithm for min-cost flow than \cite{chen2022maximum, detMaxFlow}.
And, more importantly, we solve thresholded min-cost flow, and hence the classic incremental graph problems of cycle detection, strongly-connected components, $s$-$t$ shortest paths, and maximum and min-cost flow.

\subsection{Paper Organization}
\label{sec:organization}
The remainder of the paper is organized as follows. In \Cref{sec:applications}, we give applications of our deterministic min-ratio cycle data structure to incremental min-cost flow and its corollaries. We give an overview of our algorithm in \Cref{sec:overview}.

We give the preliminaries in \Cref{sec:prelims}. In \Cref{sec:min_ratio_ds}, we state our data structures which use an $\ell_1$-oblivious routing to reduce min-ratio cycle to min-ratio tree cycle, and a data structure to reduce min-ratio tree cycle to min-ratio cycle on a smaller graph. These are combined to give our overall min-ratio cycle data structure. Then we discuss how to use an interior point method to leverage this data structure to solve incremental min-cost flow.
In \Cref{sec:hrg_ds} we build the $\ell_1$ oblivious routing, and in \Cref{sec:DynMinRatioTreeCycle} we give the dynamic algorithm that reduces min-ratio tree cycle to min-ratio cycle on a smaller graph. In \Cref{sec:SNC} we present our fully dynamic sparse neighborhood cover used in \Cref{sec:hrg_ds}, and in \Cref{sec:newSpanner} we present our new spanner algorithm used in \Cref{sec:DynMinRatioTreeCycle}. Finally, in \Cref{sec:ipm} we recall the IPM framework and present the adaption necessary for incremental strongly connected components. 

\label{sec:roadmap}

\subsection{Applications}
\label{sec:applications}

\paragraph{Application \#1: More modular static and deterministic min-cost flow.} \cite{chen2022maximum} developed an $\ell_1$-interior point method that reduced minimum-cost flow to approximate dynamic min-ratio cycle (\Cref{def:informalminratio}) with $m^{1+o(1)}$ updates.
This reduction is deterministic.
But, \cite{chen2022maximum} gave a data structure for min-ratio cycle that in fact only succeeds against \emph{oblivious adversaries}. 
To adapt their data structure to the problem, they heavily used properties of the update sequence produced by the IPM.
One of these modifications included an (arguably unintuitive) rebuilding game where layers of the data structure were rebuilt when it failed. Later, \cite{detMaxFlow} gave a deterministic algorithm for min-cost flow. However, the data structure they designed was based on that of \cite{chen2022maximum}, and still critically used properties of the IPM update sequence and a rebuilding game. By using our deterministic min-ratio cycle data structure, we achieve a deterministic min-cost flow algorithm that is more modular, in that the data structure and IPM can be completely separated, and avoid studying a rebuilding game.

\paragraph{Application \#2: Incremental thresholded min-cost flow.} \cite{vdBrand23incr} showed that the IPM which reduces min-cost flow to dynamic min-ratio cycle naturally extends to the setting of \emph{incremental thresholded min-cost flow} (\Cref{def:incrementalmincostflow}) with a few modifications.
For this problem, the IPM still only 
needs to solve dynamic min-ratio cycle (\Cref{def:informalminratio}) across
$m^{1+o(1)}$ updates.
However, due to difficulties in reasoning about the rebuilding game and the update sequence, \cite{vdBrand23incr} was not able to show that the data structures of \cite{chen2022maximum,detMaxFlow} suffice for the incremental setting, and only achieved a $m^{1+o(1)}\sqrt{n}$ runtime. Our deterministic min-ratio cycle data structure avoids these issues. We prove the following.
\begin{theorem}[Incremental thresholded min-cost flow]
    \label{thm:thresh_mc}
    There is an algorithm $\textsc{MinCostFlow}(G,$ $F)$ that given an incremental directed graph $G = (V, E, \uu, \cc)$ with capacities $\uu$ in $[1, U]$, costs $\cc$ in $[-C, C]$ such that $U, C \leq m^{O(1)}$ where $m$ is an upper bound on the total number of edges in the graph, a demand $\dd \in \R^V$, and a parameter $F \in \R_{\geq 0}$
    reports a flow of cost at most $F$ the moment such a flow becomes feasible. The algorithm is deterministic and runs in time $m \cdot e^{O(\log^{167/168} m \log \log m)}$. 
\end{theorem}

\begin{remark}
    The theorem can be extended to work with fixed point arithmetic with polylogarithmic bit precision (see \cite{chen2022maximum}).
    Furthermore, the algorithm does not need to know the final edge count $m$, as we can restart after the edge count has doubled.
\end{remark}

A simple corollary of \Cref{thm:thresh_mc} is an algorithm for approximate incremental min-cost flow.

\begin{theorem}[Approximate incremental min-cost flow]
\label{thm:approx_mincost} There is an algorithm that given an incremental directed graph $G = (V, E, \uu, \cc)$ with capacities $\uu$ in $[1, U]$, costs $\cc$ in $[1, C]$ such that $U, C \leq m^{O(1)}$ where $m$ is an upper bound on the number of edges in $G$, and demand $\dd \in \R^V$ maintains a flow of cost at most $(1 + \epsilon)\mathrm{OPT}$ throughout where $\mathrm{OPT}$ denotes the cost of the current min-cost flow. The algorithm is deterministic and runs in time $m \eps^{-1} \cdot e^{O(\log^{167/168} m \log \log m)}$. 
\end{theorem}
\begin{proof}
    We run a thresholded min-cost flow algorithm (\Cref{thm:thresh_mc}) for each threshold $(1+\epsilon)^i$, and notice that $\mathrm{OPT}$ is monotonically decreasing because the graph is incremental. Since the initial cost is polynomially upper bounded by $mCU$, the number of thresholded min-cost flow algorithms we run is $O(\eps^{-1}\log(mCU))$. The result follows directly from \Cref{thm:thresh_mc}.
\end{proof}
This also implies approximate incremental maximum flow (as a special case) and approximate incremental weighted bipartite matching (by standard reductions).
We note that the dependence on $\eps$ is optimal under the online matrix-vector (OMv) conjecture \cite{HenzingerKNS15}. Indeed, even the dynamic incremental bipartite matching problem requires $\Omega(mn^{1-o(1)})$ time under OMv. This should be compared with several previous algorithms for dynamic incremental matching which incur worse $\eps$ dependencies \cite{Gupta14,BKS23,BK23} (at least $\eps^{-2}$).
Another notable prior work is 
\cite{GH22} which gave the first sublinear amortized time algorithm for $(1-\eps)$-approximate maximum flow in incremental directed unit capacity graphs and achieved $m^{1.5+o(1)}/\eps^{1/2}$ time.

\paragraph{Application \#3: Incremental cycle detection and strongly connected components.} The incremental cycle detection problem asks us to find, in a directed graph undergoing edge insertions, the first update during which the graph has a directed cycle. This problem has been extensively studied over the last two decades \cite{MNR96,KB06,LC07,AFM08,AF10,BFG09,HKMST12,BFGT16,CFKR13,BC18, BK20}, with the current best runtimes being the minimum of $\widetilde{O}(m^{4/3})$ \cite{BK20} and $O(n^2 \log n)$ \cite{BFGT16}. Incremental cycle detection can trivially be cast as a minimum cost cycle problem by giving every edge capacity $1$, cost $-1$, and setting the threshold $F = -1$. Thus, we achieve a deterministic almost linear time algorithm for incremental cycle detection.
\begin{theorem}[Incremental cycle detection]
There is an algorithm that given a directed graph $G$ undergoing edge insertions, reports the first update during which $G$ has a directed cycle. The algorithm is deterministic and runs in time $m \cdot e^{O(\log^{167/168} m \log \log m)}$.
\end{theorem}

A generalization of the incremental cycle detection problem is to maintain the strongly connected components (SCCs) in an incremental graph.
For this problem, the best known runtimes are the same as those for incremental cycle detection, due to work of \cite{BDP21} achieving runtime $\tilde{O}(m^{4/3})$ and \cite{BFGT16}, which achieves $O(n^2\log n)$ time.
Somewhat surprisingly, the dynamic IPM framework gives a way to maintain SCCs. Intuitively, this is because the IPM maintains some circulation on the graph, such that when some edge in the circulation has nontrivial amounts of flow, it is guaranteed to be in an SCC, and thus can be contracted. Contracting an edge preserves that the remaining flow is still a circulation, so we can continue running the IPM on the contracted graph.
\begin{theorem}[Incremental SCC]
\label{thm:scc}
There is an algorithm that given a directed graph $G$ undergoing edge insertions, explicitly maintains the strongly connected components of $G$. The algorithm is deterministic and runs in time $m \cdot e^{O(\log^{167/168} m \log \log m)}$.
\end{theorem}
We should mention that nearly all the previous works discussed above also give an algorithm for maintaining a topological sort in an incremental graph, where the best runtimes are again a combination of $\widetilde{O}(m^{4/3})$ \cite{BK20} and $O(n^2 \log n)$ \cite{BFGT16}. We do not know how to do this with our current techniques.

\paragraph{Application \#4: Incremental $s$-$t$ shortest path.} For fixed vertices $s$ and $t$, the $s$-$t$ shortest path, even in a graph with negative edge lengths, can be cast as a min-cost flow problem. Thus, we obtain a deterministic almost linear time algorithm for a thresholded version of $s$-$t$ shortest path even in graphs with negative edge lengths. For graphs with positive edge lengths, we obtain a $(1+\eps)$-approximation algorithm with total runtime $m\eps^{-1} e^{O(\log^{167/168} m \log \log m)}$ by \Cref{thm:approx_mincost}.

Previously, the best runtimes to maintain $(1+\eps)$-approximate incremental $s$-$t$ shortest path were $\tilde{O}(m^{4/3}/\eps^2)$ \cite{KMP22} and $\tilde{O}(n^2/\eps^{2.5})$ \cite{PVW20}. However, all previous algorithms \cite{PVW20,chechik2021incremental, KMP22} consider the more general problem of maintaining $(1+\eps)$-approximate single-source shortest paths (SSSP) in directed graphs. \cite{bernstein2013maintaining} further gives a randomized $\tilde{O}(mn/\eps)$ time algorithm for the incremental $(1+\eps)$-approximate all-pairs shortest paths (APSP) problem. This is the best possible runtime given the approximation factor of $(1+\eps)$ \cite{dor2000all, williams2010subcubic}. The currently best deterministic such algorithm obtains total runtime $\tilde{O}(mn^{4/3}/\eps)$ \cite{karczmarz2019reliable}. In planar incremental graphs, the APSP problem can be solved with $\tilde{O}(\sqrt{n})$ worst-case time per update/query \cite{das2022near}.

It is thus an interesting question whether the techniques in this work can be used to give a (deterministic) algorithm for $(1+\eps)$-approximate incremental single-source shortest path with almost linear total update time as this would result in near-optimal algorithms for both the $(1+\eps)$-approximate incremental SSSP and APSP problems. Finally, we point out that in undirected graphs, both of these questions are settled \cite{henzinger2016dynamic, gutenberg2020deterministic, bernstein2022deterministic, KMP23}.

\subsection{Overview}
\label{sec:overview}

\paragraph{(Incremental) min-cost flow via dynamic min-ratio cycle problems.} Building on the almost linear time min-cost flow algorithm of \cite{chen2022maximum}, \cite{vdBrand23incr} used an interior point method (IPM) to show that incremental thresholded min-cost flow can be solved in $m^{1+o(1)}$ calls to a data structure which finds a min-ratio cycle in a graph undergoing at most $m^{1+o(1)}$ dynamic updates, and augments along this cycle. Our main technical contribution is an algorithm that essentially maintains an $\ell_1$-oblivious routing for a dynamic graph that allows us to extract approximate min-ratio cycles. 

Thus, we focus on designing such a data structure in the remainder of the overview. We give a description of the IPM in \Cref{sec:ipm} both for completeness and because our SCC algorithm (\Cref{thm:scc}) requires white-boxing the IPM. We present the SCC algorithm in \Cref{subsec:scc}.

\paragraph{Finding approximate min-ratio cycles via an $\ell_1$-oblivious routing.} Before we delve into a more technical discussion, let us first define $\ell_1$-oblivious routings and show how to use these objects statically to extract an approximate min-ratio cycle. 

Formally, an $\ell_1$-oblivious routing for a graph $G$ is a linear mapping $\AA \in \mathbb{R}^{E \times V}$ that takes any demand vector $\dd \in \mathbb{R}^V, \dd \bot \vecone$ and maps it to a flow $\ff = \AA \dd$ that routes the demand $\dd$. The quality of an oblivious routing matrix $\AA$ is the worst-case $\ell_1$ cost of the flow $\ff$ compared to the cheapest flow routing demand $\dd$. It is not hard to establish that the quality of $\AA$ in a unit-length graph is given by the maximum average flow path length to route a flow between the endpoints of an edge in $E$ which is given by the quantity $\|\AA\|_{1 \to 1}$. The proof extends to lengths where the quality of $\AA$ is then equal to $\|\LL \AA \LL^{-1}\|_{1 \to 1}$. 

A very simple oblivious routing scheme (though of very poor quality) is to fix an orientation of the edge set $E$, take a spanning tree $T$ in $G$, root it at an arbitrary vertex $r \in V(T)$, and let $\AA$ be the matrix defined in the entry of edge $e = (x,y)$ and vertex $v$ by 
\[
\AA_{e,v}= \begin{cases}
1 & \text{if } e \in T[v,r] \text{ and $x$ appears on $T[v,r]$ before $y$}\\
-1 & \text{if } e \in T[v,r] \text{ and $x$ appears on $T[v,r]$ after $y$}\\
0 & \text{otherwise}
\end{cases}
\]
Consider now the demand $\dd = \vecone_t - \vecone_s$ that sends one unit of flow from source $s$ to sink $t$. Let $x$ be the lowest common ancestor (LCA) of $s$ and $t$ in $T$. Then note that $\AA \dd = \AA (\vecone_t - \vecone_s)$ adds flow on the edges $T[s, x]$ and $T[x,t]$ but cancels the flow on $T[x,r]$. While using a single tree $T$ to obtain an $\ell_1$-oblivious routing cannot guarantee good quality, a distribution over roughly $m$ carefully chosen trees (called low-stretch spanning trees) yields an $\ell_1$-oblivious routing of quality $O(\log m)$  \cite{racke2008optimal}. 

Let us now consider the min-ratio cycle problem for $G$ where $\ll$ are the lengths and $\gg$ denotes the gradients. Let us fix an optimal circulation $\DDelta^*$, e.g. $\BB^{\trp} \DDelta^* = 0$, and let us assume for convenience that the orientation of all edges coincides with the direction of the flow on $\DDelta^*$. Let us define for an edge $e = (u,v)$ the vector $\bb_e = \vecone_v - \vecone_u$, i.e. the unit demand sending one unit of flow from the tail of $e$ to its head. Then, note that for any oblivious routing $\AA$, we have $\sum_{e \in E} \AA \;\DDelta^*_e \bb_e = \veczero$ because since $\DDelta^*$ is a circulation, the total demand at every vertex from vectors $\bb_e$ weighted by $\DDelta^*_e$ is $0$.

But this implies that there is an edge $e \in E$ such that the circulation obtained from routing one unit of flow along an edge $e$ and the routing the demand $- \bb_e$ via the $\ell_1$-oblivious routing is competitive to the min-cycle ratio cycle. This follows from the calculations below, we have
\begin{align*}
 \min_{e \in E} \frac{\gg^{\trp} \DDelta^*_e 
 (\vecone_e - \AA \bb_e)}{\| \LL\; |\DDelta^*_e 
 (\vecone_e - \AA \bb_e)|\|_1} &\le \frac{\sum_{e\in E} \gg^{\trp} \DDelta^*_e 
 (\vecone_e - \AA \bb_e)}{\sum_{e \in E} \| \LL\; |\DDelta^*_e 
 (\vecone_e - \AA \bb_e)|\|_1} = \frac{\gg^{\trp} \DDelta^*}{\sum_{e \in E} \| \LL\; |\DDelta^*_e 
 (\vecone_e - \AA \bb_e)|\|_1} \\
 &\le \frac{\gg^{\trp} \DDelta^*}{\sum_{e \in E} \ll_e|\DDelta_e^*| + \|\LL\AA\LL^{-1}\|_{1\to1}\ll_e|\DDelta_e^*|} \\ &= (1+\|\LL\AA\LL\|_{1\to1})^{-1} \frac{\gg^{\trp} \DDelta^*}{\|\LL\DDelta^*\|_1},
\end{align*}
where the first inequality is an averaging argument, and the second uses the triangle inequality and the quality of the oblivious routing $\AA$ with respect to $\LL$. Thus, some circulation $\vecone_e - \AA \bb_e$ has quality within a $O(\|\LL\AA\LL^{-1}\|_{1\to1})$ factor of the best ratio.

Note that if $\AA$ was again obtained from a single tree, then $\vecone_e - \AA \bb_e$ would form a simple cycle. However, in general, $\vecone_e - \AA \bb_e$ is simply a circulation. In our approach, we show that the particular structure of $\AA$ then allows us to further decompose this circulation into few cycles such that one of these preserves an approximate min-ratio cycle.

\paragraph{Dynamic min-ratio cycles in previous work.} Before we describe our new approach, let us briefly reflect on previous work with the perspective from above. \cite{chen2022maximum} used probabilistic low-stretch spanning trees (LSSTs) to find the min-ratio cycle. Interpreting their framework with our perspective from above, they essentially sample $m^{o(1)}$ LSSTs from the distribution over trees that forms the $\ell_1$-oblivious routing given in \cite{racke2008optimal} and argue that one of them has a fundamental tree cycle that is approximately a min-ratio cycle (the argument follows our discussion above as each LSSTs has the same expected guarantees as its underlying $\ell_1$-oblivious routing). But since the IPM requires us to solve a dynamic min-ratio cycle problem, \cite{chen2022maximum} then has to argue that updates to the graph do not interfere much with the randomness used when picking the LSSTs as they do not have time to sample a new tree after each update.
Instead, they resample some edges if the data structure fails to produce a good cycle.
While \cite{detMaxFlow} was able to derandomize the above construction, this was done using the structure of the update sequence produced by the IPM and does not solve the min-ratio cycle problem against an adaptive adversary.

\paragraph{High-level strategy.} In this work, we give an algorithm that essentially maintains an $\ell_1$-oblivious routing $\AA$ of graph $G$ with quality $m^{o(1)}$ and $m^{o(1)}$ update time. Moreover, we show that the particular structure of the $\ell_1$-oblivious routing $\AA$ that we maintain allows us to decompose each circulation $\vecone_e - \AA \bb_e$ into only $m^{o(1)}$ cycles, one of which being approximately of quality comparable to the quality of the circulation $\vecone_e - \AA \bb_e$. This then allows us to efficiently extract from the $\ell_1$-oblivious routing $\AA$ an approximate min-ratio circulation at any time. This is in stark contrast with previous approaches that could only maintain a small subsample of the $\ell_1$-oblivious routing and our stronger invariant presents major challenges both in terms of maintaining the $\ell_1$-oblivious routing itself \emph{and} for extracting a min-ratio circulation, as the object we maintain  essentially covers all cycles approximately (instead of only maintaining each fixed cycle with high probability).

We emphasize another important point in our approach: while \cite{chen2022maximum, detMaxFlow} argues about preservation of lengths and gradients by the LSSTs simultaneously, our framework separates these concerns: we maintain the $\ell_1$-oblivious routing with an algorithm that is oblivious to the gradients. Only thereafter, when the routing is given, do we introduce gradients again and show how to extract a min-ratio cycle. Thus, a crucial point is that we first deal with maintaining distances, where we offload some work to the toolbox designed in \cite{KMP23}, and then develop and use another set of tools that allow us to handle gradients and flow routing.
We note that there are major obstacles to directly extending the the distance preservation techniques of \cite{KMP23} to simultaneously handle gradients.

\paragraph{Our algorithmic framework.} To obtain our algorithm, we build on the framework in \cite{rozhon2022undirected} to maintain an $\ell_1$-oblivious routing. \cite{rozhon2022undirected} takes as an initial point a rather simple building block, a so-called sparse neighborhood cover (SNC). Given graph $G$ and a distance parameter $D$, an SNC is a collection $\mathcal{C}$ of clusters $C \in \mathcal{C}$ such that 
\begin{enumerate}
    \item for every vertex $v \in V$, the ball of radius $D$ around $v$, which we denote as $B_G(v, D)$, is contained in some cluster $C$, and
    \item each vertex $v \in V$ appears in at most $\tilde{O}(1)$ clusters, and
    \item each cluster $C$ has diameter at most $m^{o(1)} \cdot D$. 
\end{enumerate} 
\cite{rozhon2022undirected} maintains SNCs $\mathcal{C}_i$ for increasing distance parameters $D_i$. We choose $D_i = \gamma^i$ for $\gamma = m^{o(1)}$. Then, each cluster $C \in \mathcal{C}_i$ is assigned an arbitrary cluster center $r_C$ and selects flow paths to go from centers of clusters $C \in \mathcal{C}_i$ to centers of clusters $C' \in \mathcal{C}_{i+1}$ with $C \subseteq C'$. Thus, note that flow out of $C$ only goes to $\O(1)$ different clusters. Finally, \cite{rozhon2022undirected} carefully chooses an extremely clever but simple weighting of these flow paths (we however will not maintain such weights).

In order to dynamize the algorithm from \cite{rozhon2022undirected}, we first design a new algorithm to maintain SNCs in a fully dynamic graph. While this in itself already presents a considerable challenge, we also need to maintain low-diameter trees spanning each cluster $C \in \mathcal{C}_i$ to keep track of the paths between centers of clusters $C \in \mathcal{C}_i$ and $C' \in \mathcal{C}_{i+1}$. In fact, later, when we route flow along min-ratio cycles, we have to use these low-diameter trees to efficiently maintain the flow routed. To maintain the low-diameter trees, we build on concurrent work of \cite{KMP23}.

Given the dynamic SNC algorithm, we then maintain a relaxed version of the oblivious routing of \cite{rozhon2022undirected}, which we call an \emph{abstracted hierarchical routing graph} (abstracted HRG) $\tilde{H}$. This graph $\tilde{H}$ consists of $\kappa := O(\log_{\gamma} m)$ layers where each layer $0 \leq i \leq \kappa$ consists of a copy $V_i$ of the vertex set $V(G)$. We draw an edge out of $v_i \in V_i$ as follows: let $v_i \in C_i$, and $C_i \subseteq C_{i+1}$. Then draw an edge from $v_i$ to $r_{C_{i+1}}$. At a high level, this graph captures all possible paths that flow can go on in the oblivious routing. We direct edges in $\tilde{H}$ from lower layers to higher layers, and thus $\tilde{H}$ is a directed acyclic graph (DAG). Again, the maximum out-degree of any vertex in $\tilde{H}$ is $\Otil(1)$. We let the HRG be the graph $H$ obtained from $\tilde{H}$ by replacing each edge in $\tilde{H}$ by a path between the same endpoints from the low-diameter tree on the corresponding cluster. Thus, paths in the HRG exactly are paths in $G$ that the oblivious routing pushes flow onto. See \Cref{fig:HRG}, and \Cref{def:HRG} for a more formal definition.

We then show that picking random out-edges in $\tilde{H}$ gives with probability $m^{-o(1)}$ a tree $T$ that yields a fundamental cycle that approximately preserves the min-ratio cycle. We show that this approach can be derandomized and maintain $m^{o(1)}$ trees $T_1, T_2, \ldots, T_{\lambda}$ of $\tilde{H}$ such that one of them has a fundamental cycle that approximately preserves the min-ratio cycle. 

Finally, we are left with extracting such a fundamental tree cycle. At a high level, our approach is similar to how \cite{chen2022maximum} extracts such cycles from its LSSTs, we use vertex and edge-reduction. To this end, we partitioning the trees $T_1, T_2, \ldots, T_{\lambda}$ further into forests $F_1, F_2, \ldots, F_{\lambda}$ each with $\Omega(m/k)$ connected components by deleting edges that ensure that each component is incident to $\tilde{O}(k)$ volume of the underlying graph $G$. We then show for each $1 \leq i \leq \lambda$, how to extract an approximate min-ratio cycle: either if it was preserved on $F_i$, we use that each component of $F_i$ is small to extract it, or if the approximate min-ratio cycle preserved was preserved by $T_i$ but not by $F_i$, we extract it by recursing on the graphs obtained from contracting in $G$ vertices in the same component of a forest $F_i$ obtained from partitioning $T_i$.
For each forest $F_i$, this contraction yields a graph $P_i$ with only $O(m/k)$ vertices.
We then use a spanner (i.e. a distance-preserving edge sparsifier) on this graph $P_i$ to reduce the edge count down to $m^{1+o(1)}/k$ and recurse on the resulting smaller graph $S_i$. 
Applying edge sparsification to $P_i$ unfortunately means we cannot guarantee that there exists a good tree cycle w.r.t. the contraction of tree $T_i$ in $S_i$.
Hence, our recursion on $H_i$ needs to solve the min-ratio cycle problem, not the simpler min-ratio tree cycle problem.
Our edge sparsification procedure needs significantly stronger properties than the spanner in \cite{chen2022maximum}. 
We build a more powerful spanner to solve this issue, as we describe later in the overview. 

In the next sections, we describe the components of the algorithm in more detail.

\paragraph{From hierarchical routing graphs (HRGs) to min-ratio cycle-preserving trees.}
We now describe how to use the HRG to maintain $m^{o(1)}$ trees, such that in some tree, the fundamental cycle arising
from the tree combined with an off-tree edge is an approximate min-ratio circulation.
To start, as we showed earlier, there is an edge $e \in E(G)$ such that obliviously routing $e$ in the HRG gives an approximate min-ratio circulation. Algebraically, this circulation is $\AA\BB^\top \vecone_e - \vecone_e$, where $\vecone_e$ is the unit flow on edge $e$ (note that $\bb_e = \BB^{\top} \vecone_e$).
However, the flow $\AA\BB^\top \vecone_e$ may consist of multiple paths within the HRG, because vertices are in multiple clusters.
To reduce to considering circulations which are a simple cycle, we will build a collection of trees that ``covers'' all flow paths in the HRG. To this end, we show that for $e = (u,v)$, we can decompose the flow $\AA\BB^\top \vecone_e$ into paths each consisting of a path aligning with the HRG from $u$ to one of the nodes in a higher layer, and the reverse path from such a higher layer node down to vertex $v$. Then, to cover these flow paths, consider picking a random edge out of each vertex $v_i \in V_i$ for all $i$, and repeating $m^{o(1)}$ times, to find trees $T_1, T_2, \ldots, T_\lambda$. Indeed, note that all flow paths use at most $O(\kappa)$ hops in the abstracted HRG (this is the number of layers in the HRG where we defined $\kappa := O(\log_{\gamma} m)$ which is sublogarithmic), and each vertex has out-degree $\O(1)$. So the probability that a flow path is contained in a random tree $T_j$ is $\O(1)^{-O(\kappa)} \geq m^{-o(1)}$ w.h.p. 

A subtle issue arises: we claimed above that $\AA\BB^\top \vecone_e$ can be decomposed into flow paths that consist of a path segment starting in $u$ that is only going up and a path segment ending in $v$ that is only going down in the HRG. We call these monotone cycles. 

Unfortunately, this montone cycle decomposition must sometimes be lossy as shown by simple examples. However, we are still able to show that every circulation on the HRG can be decomposed into monotone cycles such that the sum of the cycle lengths is only a little larger than the length of the circulation which suffices to preserve a cycle of good quality.

Finally, we derandomize this construction by reducing the amount of randomness via coloring techniques and enumeration. 

\paragraph{Completing our recursion: Querying min-ratio tree cycles with portal routing.} Having constructed trees whose fundamental cycles yield approximate min-ratio cycles, we now need to solve the problem of finding a min-ratio cycle that consists of an off-tree edge and a simple path on a single tree in this collection derived from the HRG. We follow the high-level approach from \cite{chen2022maximum}: we apply vertex and edge reduction tools to reduce the problem again to min-ratio cycle, but on a smaller graph with $O(m/k)$ edges, for some size reduction factor $k$. Then we recursively build an HRG on this graph to reduce to tree cycles again, all the way down to graphs of size $O(k)$. Combining these pieces gives our overall data structure.

Recall that, as described earlier, to solve the min-ratio tree cycle problem on a tree $T$, we first partition the tree into a forest $F$ with $m/k$ components of size $\Otil(k)$. We then check tree cycles that are internal to components, before constructing a smaller graph $P$ with components of $F$ contracted, so that $P$ contains $O(m/k)$ vertices.
$P$ still has roughly $m$ edges, and is hence rather dense.
To construct $P$, we use a contraction scheme (called ``portal routing'' \cite{ST04,KPSW19}), that differs from \cite{chen2022maximum} as we have to ensure that no fundamental cycle is lengthened by contraction.
This scheme is less stable than the scheme from \cite{chen2022maximum} (called ``core graphs''): Updates to $G$ cause vertex splits in $P$ where the edges incident to a split vertex may all decrease drastically in length (we model these decreases by re-inserting the edges with new lengths).
This makes it challenging to apply edge sparsification to $P$, because a single split in $P$ might cause up to $\tilde{O}(k)$ edge insertions (the maximum degree of a vertex in $P$ is $\tilde{O}(k)$ by our decomposition of a tree $T_i$ into a forest $F_i$). Thus, $m/k$ updates to $G$ may cause $P$ to undergo roughly $m$ edge insertions. But, we show that nonetheless, we can maintain a spanner $S$ of $P$ which only undergoes roughly $m/k$ updates.

\paragraph{Fully dynamic sparse neighborhood cover.} 
Dynamically maintaining the HRG and the hierarchical routing trees can be reduced to giving a deterministic data structure that maintains a sparse neighborhood cover of a fully dynamic graph, and a representation of a low-diameter tree over each cluster $C \in \mathcal{C}_i$. We discuss our fully dynamic SNC algorithm in \Cref{sec:SNC}. Although there are previous works on dynamic SNC \cite{Chu21, bernstein2022deterministic,CZ23} it is unclear whether these algorithms can maintain a low-diameter tree representation of the SNC that allows for efficient routing.\footnote{It is not hard to see that these constructions internally maintain low-diameter trees that embed with very high edge-congestion into $G$ which does not suffice for our purposes.}

We maintain a fully dynamic SNC by leveraging several vertex sparsification tools from the work \cite{KMP23}, which developed a set of tools based on fully dynamic vertex sparsifiers that approximately preserve distances and have subpolynomial amortized update time. In particular, we require data structures which 1) maintain a SNC under edge \emph{deletions} only (\Cref{thm:dec_snc}), 2) maintain a vertex sparsifier onto a dynamic terminal set $A$ (\Cref{thm:mainTheoremVSExt}), 3) maintain a low-diameter tree over a graph (\Cref{thm:mainTheoremLowDiamTree}).

We briefly describe how to use these pieces. Let the dynamic graph be $G$, and let $\hat{G}$ be the dynamic graph if we ignore all insertions. We start by maintaining a decremental SNC on $G$ with diameter parameter $D$. At the same time, we grow a terminal set $A$ consisting of endpoints of all inserted edges, on which we maintain a vertex sparsifier called $H$. We recurse on $H$ (with slightly larger diameter parameter), and rebuild $G$ whenever $|H| > |G|/\gamma$, for a size reduction factor $\gamma = m^{o(1)}$.\footnote{Our approach is somewhat similar to the technique developed in \cite{forster2023bootstrapping} for maintaining approximate APSP in a fully dynamic graph by reducing to maintaining approximate APSP on decremental graphs.}

Let us describe why this construction works. We show that for every vertex $v$ there is a maintained cluster that contains the ball of radius $D/\gammasnc$ around it. There are two cases. If a low-diameter cluster $\hat{C}$ in $\hat{G}$ has no touching edge insertions, then any ball $B_{\hat{G}}(v, D/\gammasnc) \subseteq \hat{C}$ also satisfies $B_G(v, D/\gammasnc) \subseteq \hat{C}$. Thus, $\hat{C}$ still covers many balls around vertices $v$ in $G$.
Otherwise, there is an insertion touching $\hat{C}$, say at vertex $c$. We know $c \in A$ (the terminal set), because we add all endpoints of insertions to $A$. Now, recall that we recurse on $H$, which is a vertex sparsifier for terminal set $A$. Thus, we recursively maintain some sparse neighborhood cover on $H$, in particular a low-diameter set containing $c \in C$ (let's call it $C_H$). Thus, we can consider the low-diameter set $C \cup C_H$, which we show contains balls around $v$.

To get a low-diameter forest/tree representation out of this construction, we maintain a low-diameter tree on each cluster in $\hat{G}$ (and recursively on vertex sparsifiers), and glue these together at terminals (when forming the sets $C \cup C_H$) to form the overall forest. Our formal construction is more complicated than this, partially because we must duplicate vertices and forests multiple times due to vertex congestion in the vertex sparsifiers.

\paragraph{Low-recourse dynamic spanner under many insertions.} When solving the problem of min-ratio tree cycle, we apply vertex and edge reduction techniques. We require a stronger spanner than \cite{chen2022maximum} to implement our edge reduction step, and we develop this in \Cref{sec:newSpanner}.

A spanner $H$ of a graph $G$ with stretch $\gamma$ is a subgraph of $G$ with the same vertex set such that for every edge $(u,v)$ in $G$ there is a path between the $u$ and $v$ that is at most a factor $\gamma$ longer than the edge $e$. By standard reductions, it suffices to consider the unit weight case.

In \cite{chen2022maximum, detMaxFlow} it was sufficient for the spanner to handle roughly $n$ edge insertions and deletions, and vertex splits. In this case, the inserted edges can all be added to $H$ without increasing its density by more than an additive $n$. Therefore, a dynamic spanner handling edge deletions and vertex splits suffices. We present a spanner that can handle a large amount of insertions. 

\begin{informaltheorem}\label{thm:informal_spanner}
    Given a dynamic graph $G$ undergoing a sequence of up to $\Delta n$ edge insertions and up to $n$ edge deletions and vertex splits such that the maximum degree in $G$ never exceeds $\Delta$ even when playing the update sequence without vertex splits. Then, there is an algorithm that maintains a spanner containing at most at most $\gamma n$ edges with stretch $\gamma$ and total recourse $\gamma n$ for some $\gamma = m^{o(1)}$ independent of $\Delta$. The total runtime is $\gamma n \Delta$.
\end{informaltheorem}

We first describe a simple low recourse version of this result, and then sketch how we obtain the efficient version. 
Recall the classic spanner construction by \cite{althofer1993sparse}. Given $G = (V, E)$, initialize the spanner $H = (V, \emptyset)$. Then, consider all edges $(u,v) \in E$ in some arbitrary order. Whenever there is no path of length $2 \log n$ between $u$ and $v$ in $H$, add $(u, v)$ to $H$. This construction ensures that $H$ has girth $> 2 \log n$, which implies that $H$ only contains $O(n)$ edges. Notice that this already describes an incremental algorithm. 

After deletions and vertex splits, we  observe that these operations only increase the girth of $H$. Therefore, after each edge deletion/vertex split we apply $H$ we go through all edges $(u,v)$ currently in $G$ and check if there is a path of length $2 \log n$ between $u$ and $v$ in $H$. If not, we add $(u, v)$ to $H$. Since the girth of $H$ remains $>2 \log n$, the number of edges in $H$ remains bounded by $O(n)$. But an edge only leaves $H$ when it is deleted since no edges leave $H$ when a vertex is split. This bounds the amortized recourse caused by insertions and deletions. A simple potential function argument can be used to show that the recourse caused by simulating vertex splits is also at most $\tilde{O}(n)$. 

To achieve an efficient implementation, we first observe that, using fully dynamic APSP data structures \cite{CZ23, forster2023bootstrapping,KMP23}, we can essentially implement the classic greedy incremental algorithm in almost-linear time.
Furthermore, we can maintain explicit, short embedding paths for each edge $e \in E \setminus E_H$ into $E_H$.
To manage edge deletions and vertex splits,  
we combine the incremental spanner idea with the decremental batching scheme of \cite{chen2022maximum}. To make this idea work, we transform the incremental spanner to produce an embedding from $G$ to $H$ with low vertex congestion. It is worth mentioning that the work \cite{BSS22} realized that the greedy algorithm can maintain a spanner with optimal $O(1)$ amortized recourse in the decremental setting, albeit with large polynomial running time.
However, they did not observe that in fact the spanner can be made fully dynamic with extremely low recourse under edge insertions, and instead appealed to standard recursions to achieve $O(1)$ recourse per insertion (which is insufficient for our purposes).

\paragraph{Tree representations.}
A point we have glossed over is how we precisely route the min-ratio cycle that our algorithm returns. In \cite{chen2022maximum}, the algorithm dynamically maintain a spanning tree of the graph, and the approximate min-ratio cycle was always represented by $m^{o(1)}$ off-tree edges and $m^{o(1)}$ paths on this tree. Due to the complexity of our data structures, we actually maintain a generalization of a spanning tree, which we call a forest $F$ that has a \emph{flat embedding} into $G$. Formally, this means that there is a map $\Pi: V(F) \to V(G)$ such that for every edge $e = (x, y) \in V(F)$, $(\Pi(x), \Pi(y))$ is also an edge in $G$. This way, paths and cycles in $F$ exactly correspond to paths and cycles in $G$. Overall, we will represent our min-ratio cycles as $m^{o(1)}$ off-tree edges and paths on $F$. The tree-representation of the SNC, and the tree in the portal routing construction both have flat embeddings into $G$. Crucially, this embedding is of low edge-congestion $m^{o(1)}$. This allows us to maintain the absolute change of flow for a single edge in $G$ up to a $m^{o(1)}$ approximation which is crucial for the IPM to maintain (approximate) gradients and lengths efficiently.
\section{Preliminaries}
\label{sec:prelims}
\paragraph{General notation.}
We denote vectors as bold lowercase letter $\vv$, and matrices as bold uppercase letter $\AA$. We let $\vv(j)$ denote the $j$-th element in vector $\vv$ and $\AA(i,j)$ the element at position $(i, j)$ accordingly. 
For a vector $\vv$, we let $|\vv|$ denote the coordinate-wise absolute value. Given two scalars $\alpha$ and $\beta$, we let $\alpha \approx_{\kappa} \beta$ if $\alpha/\kappa \leq \beta \leq \kappa \alpha$.

\paragraph{Graphs.} In this article, we typically work with undirected graphs unless explicitly specified. That is while the min-cost flow problem is solved on a directed graph, we reduce via the IPM to sequence of subproblems on undirected graphs. On such graphs, we still use orientations of edges as we use gradients and flows over these graphs. 

We let $G = (V, E, \ll, \gg)$ denote a graph on the vertex set $V$ with edge set $E$, $\ll(\cdot): E \mapsto \R_{\ge 0}$ is the vector storing the lengths of the edges, and $\gg(\cdot): E \mapsto \R$ is the vector for edge gradients. We let $\BB_G \in \R^{E \times V}$, or $\BB$ if $G$ is clear from the context, denote the edge vertex incidence matrix where we attribute arbitrary directions to edges. We usually let $n = |V|$ and $m = |E|$. We further denote $\LL = \diag(\ll)$.

We call a flow vector $\cc \in \R^{E}$ a circulation if $\BB^\top \cc = \veczero$.
Given a graph $G$, we sometimes denote its vertex set by $V_G$ and its edge set by $E_G$.

\paragraph{Distances.}For a graph $G = (V, E, \ll)$ with edge lengths $\ll \in \R^E_{\geq 0}$, we let $\dist_G(u, v)$ be the length of the shortest weighted path between $u$ and $v$. We let $\diam(G)$ denote the maximum distance between two vertices in $G$, i.e. $\diam(G) = \max_{u,v} \dist_G(u,v)$. We denote the \emph{ball} of radius $D$ around $v$ as $B_G(v, D) := \{u \in V(G) : \dist_G(u, v) \le D\}$. 

\paragraph{Embeddings. } Given two graphs $G = (V_G, E_G)$ and $H = (V_H, E_H)$ with $V_G \subseteq V_H$ we call a map of edges in $G$ to paths between their respective endpoints in $H$ a (edge) embedding and often denote such an embedding by $\Pi_{G \mapsto H}$. Similarly, we often denote by $\Pi_{V_G \mapsto V_H}$ a map from vertices in $V_G$ to vertices in $V_H$ and we call it either a vertex embedding or a (vertex) map. We sometimes drop the subscript if it is clear from the context.

\begin{definition}[Vertex and Edge congestion]
    Given a graph $G = (V, E)$, a (sub-)graph $H = (V_H, E_H)$ with $V \subseteq V_H$ and an edge embedding $\Pi_{G \mapsto H}$ that maps edges $e = (u, v)$ to a $uv$-path $P$ in $H$ we let:
    \begin{itemize}
        \item The edge congestion $\econg(\Pi_{G \mapsto H}, e')$ of edge $e' \in E_H$ be the total number of paths that use edge $e'$. Further, we let $\econg(\Pi_{G \mapsto H})$ be the maximum edge congestion observed by an edge in $E_H$.
        \item The vertex congestion  $\vcong(\Pi_{G \mapsto H}, v)$ of $v \in V_H$ be the total number of paths in $\Pi$ that contain $v$. Similarly, we let $\vcong(\Pi_{G \mapsto H})$ be the maximum vertex congestion observed by a vertex in $V$.
    \end{itemize}
    We will sometimes consider broken embeddings, i.e. embeddings $\Pi_{G \mapsto H}$ where edges on an embedding path got removed or vertices got split. We therefore don't require the paths $P$ to correspond to paths in $H$, but they may simply be an ordered collection of edges in $H$. 
\end{definition}

\begin{definition}[Vertex map congestion]
    For a vertex map $\Pi_{V_G \mapsto V_H}$, we let $\vcong(\Pi_{V_G \mapsto V_H})$ be the maximum number of vertices in $V_G$ that map to the same vertex in $V_H$.
\end{definition}

\begin{restatable}[Flat embedding]{definition}{flatEmbedding}
\label{def:lift}
We say for graphs $G =(V,E)$ and $H=(V_H, E_H)$ that a map $\Pi_{V_H \mapsto V}$ is a \emph{flat embedding} of $H$ into $G$ if for every edge $e = (x, y) \in E_H$, for $x' = \Pi_{V_H \mapsto V}(x)$ and $y' = \Pi_{V_H \mapsto V}(y)$, either $(x',y')$ is an edge in $G$, or $x'=y'$. 
\end{restatable}

\paragraph{Oblivious routings.} Our data-structure is based on dynamically maintaining $\ell_1$-oblivious routings for graphs $G = (V, E, \ll)$, which are linear mappings $\AA \in \R^{E \times V}$ of demands $\dd \in \R^{V}$ to flows $\ff \in \R^{E}$ which route that demand. In other words, for $\dd \perp \vecone$ the flow $\ff = \AA \dd$ satisfies  $\BB^T \ff = \dd$, i.e. it routes the demands $\dd$. For every edge $e = (u,v)$ in $G$, we let $\vecone_e$ denote the flow that routes one unit across edge $e$, and $\bb_e$ for the demand that it routes.

\begin{definition}[Oblivious routing]
    \label{def:obl_routing}
    We call a matrix $\AA \in \R^{E \times V}$ a $\gamma$-approximate $\ell_1$-oblivious routing for graph $G = (V, E, \ll)$ if 
    \begin{enumerate}
        \item $\forall \dd$ so that $\dd^T \vecone = 0$: $\BB^T\AA \dd = \dd$
        \item $\norm{\LL \AA \BB^T \LL^{-1}}_{1 \rightarrow 1} \leq \gamma$
    \end{enumerate}
\end{definition}

\begin{remark}
    Notice that this coincides with the combinatorial definition of an $\ell_1$-oblivious routing over competitive ratios, i.e. 
    consider an arbitrary flow $\ff$ and observe
    \begin{align*}
        \frac{\norm{\LL \AA \BB^T \ff}_1}{\norm{\LL\ff}_1} = \frac{\|\LL \AA \BB^T \LL^{-1} \tilde{\ff}\|_1}{ \|\tilde{\ff}\|_1} \leq \gamma
    \end{align*}
    which follows by the definition of $\norm{\cdot }_{1 \rightarrow 1}$ and $\ff$ and thus $\tilde{\ff}$ are arbitrary.
\end{remark}

For an oblivious routing $\AA$, we let $\PP = \II - \AA\BB^\top$ denote the corresponding \emph{cycle projection matrix}, as it is a projection onto the space of circulations on $G$.

An averaging argument shows that there is a circulation $\PP\vecone_e$ whose quality is within a multiplicative factor of the quality of the best min-ratio circulation.
\begin{lemma}[Min-ratio circulation]
    \label{lem:min_ratio}
    Given a graph $G = (V, E, \ll, \gg)$, a cycle projection matrix $\PP \in \R^{E \times E}$ for $G$ and an arbitrary matrix $\MM \in \R^{k \times E}$ we have
    \begin{equation}
        \label{eq:min_ratio_bound}
        \min_{\substack{e \in E \\ \beta \in \{-1, 1\}}} \beta \cdot \gg^T \PP \vecone_e/\norm{\MM \vecone_e}_1 \leq \frac{1}{\norm{\MM \LL^{-1}}_{1 \rightarrow 1}} \min_{\DDelta: \BB^\top \DDelta = \veczero} \gg^T \DDelta/\norm{\LL \DDelta}_1
    \end{equation}
\end{lemma}
\begin{proof}
Let $\DDelta^*$ denote the minimizer of the RHS. Because $\DDelta^*$ is a circulation, we get \[ \gg^\top \DDelta^* = \gg^\top \PP\DDelta^* = \sum_{e \in E} \DDelta_e^* \cdot \gg^\top \PP \vecone_e. \]
For the denominator we also get that
\[ \|\LL\DDelta^*\|_1 = \sum_{e \in E} |\DDelta_e^*| \|\LL \vecone_e \|_1 \ge \frac{1}{\|\MM\LL^{-1}\|_{1\to1}} \sum_{e \in E} |\DDelta_e^*| \|\MM\vecone_e\|_1. \]
Now, the desired claim follows by a standard averaging argument.
\end{proof}

\paragraph{Trees.} For a tree $T = (V,E)$ we denote the path from $u$ to $v$ as $T[u,v]$. Given an extra edge $e$, we denote the cycle formed by the tree and edge $e$ as $\tcycle[e]$. Additionally, we denote the vector that sends a unit amount of flow along the path and cycle as $\vecone_{T[u,v]}$ and $\vecone_{\tcycle[e]}$ respectively (the latter vector sends a unit of flow along $e$ and then from $v$ to $u$). 

\paragraph{Dynamic Trees.} We frequently use dynamic trees to efficiently maintain our data structures. 

\begin{lemma}[Dynamic trees, Lemma 3.3 in \cite{chen2022maximum}, derived from \cite{ST83}]
\label{lem:dyn_trees}
There is a deterministic data structure $\mathcal{D}$ that maintains a dynamic tree $T \subseteq G = (V, E)$ under insertion/deletion of edges with gradients $\gg$ and lengths $\ll$, and supports the following operations:
\begin{enumerate}
    \item Insert/delete edges $e$ to $T$, under the condition that $T$ is always a tree, or update the gradient $\gg(e)$ or lengths $\ll(e)$. The amortized time is $\tilde{O}(1)$ per change.
    \item For a path vector $\DDelta = \vecone_{T[u,v]}$ for some $u, v \in V$, return $\l \gg, \DDelta \r$ or $\l \bell, |\DDelta|\r$ in time $\tilde{O}(1)$.
    \item \label{item:nonpositiveflow} Maintain a flow $\ff \in \R^E$ under operations $\ff \assign \ff + \eta\bDelta$ for $\eta \in \R$ and path vector $\bDelta = \vecone_{T[u,v]}$, or query the value $\ff(e)$ in amortized time $\tilde{O}(1)$.
    \item \label{item:positiveflow} Maintain a positive flow $\ff \in \R^E_{>0}$ under operations $\ff \gets \ff + \eta|\DDelta|$ for $\eta \in \R_{\geq 0}$ and path vector $\DDelta = \vecone_{T[u,v]}$, or or query the value $\ff(e)$ in amortized time $\tilde{O}(1)$.
    \item $\textsc{Detect}()$. For a fixed parameter $\eps$, and under positive flow updates (item \ref{item:positiveflow}), where $\DDelta^{(t)}$ is the update vector at time $t$, returns
    \begin{align}
        \label{eq:detect} S^{(t)} \defeq \left\{ e \in E : \ll(e) \sum_{t' \in [\last^{(t)}_e+1,t]} |\DDelta^{(t')}(e)| \ge \eps \right\}
    \end{align}
    where $\last^{(t)}_e$ is the last time before $t$ that $e$ was returned by $\textsc{Detect}()$. Runs in time $\tilde{O}(|S^{(t)}|)$.
\end{enumerate}
\end{lemma}

\paragraph{Degree reduction.} Throughout, it is often convenient to assume the graphs we work with have bounded degrees.
\begin{fact}[Standard BST dynamic degree reduction]
  \label{fac:dynlowdeg}
  Consider a dynamic graph $G = (V,E)$ undergoing edge insertions and
  deletions and insertions/deletions of isolated vertices.
  There exists a data structure
  $\textsc{LowDeg}()$ that maintains a graph $H$ and a forest $F$ such that 
  \begin{enumerate}
  \item $F$ contains exactly one connected component for each vertex $v \in V$, denoted $T_v$.
  \item $T_v$ is a balanced binary search tree with $\deg_G(v)$
    leaves, each corresponding uniquely to edges incident to $v$.
  \item The graph $H$ consists of the forest $F$ and
    for each edge $(u,v) \in E_G$
    the graph $H$ contains an edge between the leaves associated
    with $(u,v)$ in $T_u$ and $T_v$.
  \end{enumerate}
Using this construction, we have that
\begin{enumerate}
\item The maximum degree of $H$ at any time is 3.
\item The number updates to $H$ under updates to $G$ is bounded by $O(\log m)$ worst case per update to $G$.
\item When $G$ contains $m$ edges, $H$ contains at most $4m$
  edges.
\end{enumerate}
\end{fact}
\section{Dynamic Min-Ratio Cycle}
\label{sec:min_ratio_ds}

In this section we build a deterministic data structure that allows us to toggle along approximate min-ratio cycles in a dynamic graph. It is based on the data structures developed in \Cref{sec:hrg_ds} and \Cref{sec:DynMinRatioTreeCycle}. 

\paragraph{Flat embeddings.} Because min-ratio cycles in graphs may contain up to $n$ edges, previous works \cite{chen2022maximum} represented these cycles implicitly as a few paths on a dynamically changing spanning tree of the underlying graph $G$. The min-ratio cycle data structures we build in this work are not as naturally ``tree-based" as in previous works, and therefore we cannot maintain a spanning tree of the original graph. We instead maintain a forest satisfying the following.

\flatEmbedding*
Sometimes, we simply say that $H$ is \emph{flat} for brevity where we implicitly also mean that there is a vertex map $\Pi_{V(H)\mapsto V}$.
An intuitive way to understand \Cref{def:lift} is that $H$ flatly embeds into $G$ if paths in $H$ naturally correspond to (not necessarily simple) paths of the same gradient in $G$, and length upper bounds. This implies that for a circulation in $H$, the image of the circulation on $G$ has the same or better ratio. Ultimately, our algorithm will toggle cycles on a forest $F$ which flatly embeds into $G$ with low vertex and edge congestion. We then use a link-cut data structure to maintain the flow on $H$ which will be sufficient to approximately maintain the (absolute) flow added to each edge in $G$ given that the edge congestion of the flat embedding is low, and whenever the increase is significant, we update the flow on $G$ by explicitly summing over the flow on edges in $H$ that embed into the same edge. We defer details of this flow maintenance procedure to \Cref{sec:ipm}. 

\paragraph{High-level strategy.} Let us now describe our data structure to find min-ratio cycles. Our data structure is built recursively.
First, we show that maintaining the min-ratio cycle on a dynamic graph can be reduced to maintaining the min-ratio fundamental tree cycles for a collection of $m^{o(1)}$ dynamic trees.
Then, we show that maintaining min-ratio tree cycles can be reduced to min-ratio cycles on a graph $H$ that is roughly smaller by a factor $k$ than the input graph $G$. By carefully controlling the recourse of graph $H$ to be bounded near-linearly in the number of updates to $G$, we can then recurse on $H$.

To formalize the reduction, we first define dynamic min-ratio cycle data structures.
\begin{definition}[Dynamic min-ratio cycle data structure]
\label{def:DynamicMinRatioCycle}
Given a dynamic graph $G = (V, E, \ll, \bg)$, an $\alpha$-approximate \emph{dynamic min-ratio cycle} data structure $\calD$ supports the following operation:
\begin{itemize}
    \item $\textsc{InsertEdge}(e)/\textsc{DeleteEdge}(e)$: adds/removes edge $e$ to/from $G.$
    \item $\textsc{InsertVertex}(u)$: adds a new isolated vertex $u$ to $V.$
\end{itemize}
Under these updates, $\calD$ maintains a flat forest $F$ and $\Pi_{V(F) \mapsto V}$ of $G.$
After each update, $\calD$ outputs a cycle $\cc$ represented as some tree paths on $F$, specified by endpoints, and some off-tree edges such that
\begin{align*}
    \frac{\l\bg, \cc\r}{\norm{\LL \cc}_1} \le \frac{1}{\alpha} \cdot \min_{\BB^\top \bDelta = 0} \frac{\l\bg, \bDelta\r}{\norm{\LL \bDelta}_1}
\end{align*}
\end{definition}

We also need to formalize dynamic min-ratio tree cycle data structures.
\begin{definition}[Dynamic min-ratio tree cycle data structure]
\label{def:DynamicMinRatioTreeCycle}
Given a dynamic forest $T = (V, E_T)$ with a dynamic set of off-tree edges $E_{\mathrm{off}}$ that do not cross between components of $T$ (we use a dynamic graph $G$ to denote the union of $T$ and $E_{\mathrm{off}}$) and edge lengths $\ll$ and gradients $\bg$, an $\alpha$-approximate \emph{dynamic min-ratio tree cycle} data structure $\calD$ supports the following operation:
\begin{itemize}
    \item $\textsc{InsertTreeEdge}(e)/\textsc{DeleteTreeEdge}(e)$: adds/removes edge $e$ to/from $T.$
    \item $\textsc{InsertOffTreeEdge}(e)/\textsc{DeleteOffTreeEdge}(e)$: adds/removes edge $e$ to/from $E_{\mathrm{off}}.$
    \item $\textsc{InsertVertex}(u)$: adds a new isolated vertex $u$ to $G.$
\end{itemize}
Under these updates, $\calD$ maintains a flat forest $F$ and $\Pi_{V(F) \mapsto V}$ of the graph $(V, E_T \cup E_{\mathrm{off}}).$
After each update, $\calD$ outputs a cycle $\cc$ represented as some tree paths on $F$, specified by endpoints, and some off-tree edges such that
\begin{align*}
    \frac{\l\bg, \cc\r}{\norm{\LL \cc}_1} \le \frac{1}{\alpha} \cdot \min_{\bDelta \in \{\pm \mathbf{1}_{\TCyc[e]}: e \in E_{\mathrm{off}}\}} \frac{\l\bg, \bDelta\r}{\norm{\LL \bDelta}_1}
\end{align*}
\end{definition}

In \Cref{sec:hrg_ds}, we show the reduction from dynamic min-ratio cycles to tree cycles.
The reduction is based on dynamically maintaining the $\ell_1$-oblivious routing from \cite{rozhon2022undirected}.
The oblivious routing scheme is built using \emph{sparse neighborhood covers}, which we show how to efficiently maintain in a fully dynamic setting (\Cref{sec:SNC}).
The reduction is formalized as the following theorem.

\begin{restatable}[Reducing dynamic min-ratio cycle to dynamic min-ratio tree cycle]{theorem}{HRG} \label{thm:CycleToTreeCycle}
Consider an $\alpha$-approximate dynamic min-ratio tree cycle data structure $\calD^{TC}$ (\Cref{def:DynamicMinRatioTreeCycle}) that, on any dynamic graph $H$ with at most $m$ tree and non-tree edges in total, satisfies
\begin{itemize}
    \item it takes $T_{\mathrm{init}}(m)$-time to initialize.
    \item each update takes $T_{\mathrm{upd}}(m)$-amortized time.
    \item the flat forest $F^H$ it maintains has vertex congestion $\gamma_{\mathrm{vcong}}$.
    \item the output cycles $\cc$ are represented by at most $\gamma_{\mathrm{cycle}}$ tree paths on $F^H$ as well as $\gamma_{\mathrm{cycle}}$  off-tree edges.
\end{itemize}
Then, there is an $(\alpha \cdot \kappahrg \cdot \gamma_{\mathrm{route}})$-approximate dynamic min-ratio cycle data structure $\calD^{\mathrm{MRC}}$ (\Cref{def:DynamicMinRatioTreeCycle}), that, on any dynamic graph $G$ with at most $m$ edges, achieves
\begin{itemize}
\item $m \cdot \gamma_{\mathrm{tree}} + \gamma_{\mathrm{tree}} \cdot T_{\mathrm{init}}(m \cdot \gamma_{\mathrm{tree}})$ initialization time.
\item $\gamma_{\mathrm{tree}} \cdot T_{\mathrm{upd}}(m \cdot \gamma_{\mathrm{tree}})$ amortized update time.
\item the flat forest $F^G$ it maintains has vertex congestion at most $\gamma_{\mathrm{tree}} \gamma_{\mathrm{vcong}}.$
\item the output cycles are represented by at most $\gamma_{\mathrm{cycle}}$ tree paths on $F^G$ and non-tree edges.
\end{itemize}
where $\kappahrg = O(\log^{1/84} m), \gamma_{\mathrm{route}} = e^{O(\log^{83/84} m)}$ and $\gamma_{\mathrm{tree}} = e^{O(\log^{83/84} m \log\log m)}$ are the parameters in \Cref{lem:DynHRGTrees}.
\end{restatable}

Later in \Cref{sec:DynMinRatioTreeCycle}, we reduce the maintenance of min-ratio tree cycles to min-ratio cycles on a smaller graph.
This is done via dynamically maintaining the portal routing of off-tree edges, i.e., moving off-tree edges onto a small subset of carefully chosen vertices along the tree~\cite{ST04,KPSW19,CGHPS20}.
After moving off-tree edges to portals, we further reduce the number of these edges using spanners.
To control the number of updates propagating to the next level of recursion, we design a new dynamic spanner that has a small recourse under vertex updates (\Cref{sec:newSpanner}).
The reduction is formalized as the following theorem.

\begin{restatable}[Reducing dynamic min-ratio tree cycle to smaller min-ratio cycle]{theorem}{TreeCycleToCycle}
\label{thm:TreeCycleToCycle}
Consider an $\alpha$-approximate dynamic min-ratio cycle data structure $\calD^{\mathrm{MRC}}$ (\Cref{def:DynamicMinRatioCycle}) that, on any dynamic graph $H$ with at most $m$ edges, satisfies:
\begin{itemize}
    \item It takes $T_{\mathrm{init}}(m)$-time to initialize.
    \item Each update takes $T_{\mathrm{upd}}(m)$-amortized time.
    \item The flat forest $F^H$ it maintains has vertex congestion $\gamma_{\mathrm{vcong}}$.
    \item The output cycles $\cc$ are represented by $\gamma_{\mathrm{cycle}}$ tree paths on $F^H$ and off-tree edges.
\end{itemize}
Then, for any size reduction parameter $k$, there is a $(\alpha \cdot \gamma_{\mathrm{spanner}})$-approximate dynamic min-ratio tree cycle data structure $\calD^{\mathrm{TC}}$ (\Cref{def:DynamicMinRatioTreeCycle}), where $\gamma_{\mathrm{spanner}} =  e^{O(\log^{20/21} m \log\log m)}$ is the parameter given in \Cref{thm:newSpanner}, that, on any dynamic graph $G$ with at most $m$ total tree and non-tree edges, achieves
\begin{itemize}
\item $m \gamma_{\mathrm{spanner}} + T_{\mathrm{init}}(\gamma_{\mathrm{spanner}} \cdot m / k)$ initialization time.
\item $\gamma_{\mathrm{spanner}} \cdot \left(k^2 \gamma_{\mathrm{vcong}} + \frac{T_{\mathrm{init}}(\gamma_{\mathrm{spanner}} \cdot m / k)}{m/k} + T_{\mathrm{upd}}(\gamma_{\mathrm{spanner}} \cdot m / k)\right)$ amortized update time.
\item the flat embedding forest $F^G$ it maintains has vertex congestion at most $2 \gamma_{\mathrm{vcong}}.$
\item the output cycles are represented by at most $\max \{\gamma_{\mathrm{cycle}}, \gamma_{\mathrm{spanner}}\}$ tree paths on $F^G$ and non-tree edges.
\end{itemize}
\end{restatable}

Our main result in this section is a deterministic dynamic min-ratio cycle data structure with $m^{o(1)}$-amortized update time.
The data structure is based on the two reductions stated previously, \Cref{thm:CycleToTreeCycle} and \Cref{thm:TreeCycleToCycle}.
Ultimately, alternating the use of them allows us to recurse on a new dynamic graph problem with input size reduced by a factor of roughly $k.$
Setting $k$ properly, which is $e^{O(\log^{167/168} m \log\log m)}$ in our case, gives the desired result.
\begin{theorem} \label{thm:ds_query}
For some $\gamma_{\mathrm{MRC}} = e^{O(\log^{167/168} m \log\log m)}$, there is a $\gamma_{\mathrm{MRC}}$-approximate dynamic min-ratio cycle data structure $\calD$ that, on any dynamic graph $G$ with at most $m$ edges, satisfies
\begin{itemize}
\item initialization takes $m \cdot \gamma_{\mathrm{MRC}}$-time.
\item each update takes amortized $\gamma_{\mathrm{MRC}}$-time.
\item the flat forest $F^G$ it maintains has vertex and edge congestion at most $\gamma_{\mathrm{MRC}}.$
\item after each update, $\calD$ outputs a $\gamma_{\mathrm{MRC}}$-approximate min-raito cycle represented as $\gamma_{\mathrm{MRC}}$ tree paths on $F^G$ specified by endpoints and $\gamma_{\mathrm{MRC}}$ off-tree edges.
\end{itemize}

\end{theorem}
\begin{proof}
We build the data structure recursively using \Cref{thm:CycleToTreeCycle} and \Cref{thm:TreeCycleToCycle} with a reduction parameter $k$ set to be $e^{O(\log^{167/168} m \log\log m)}.$
We stop the recursion until the input graph has at most $k$ edges, where the problem can be solved to a $2$-approximation in $\O(k^2)$ time.

Let $T_{\mathrm{init}}$ and $T_{\mathrm{upd}}$ be our data structure's initialization and amortized update time.
\Cref{thm:CycleToTreeCycle} and \Cref{thm:TreeCycleToCycle} yield
\begin{align*}
T_{\mathrm{init}}(m) 
&\le m \gamma_{\mathrm{tree}} + \gamma_{\mathrm{tree}} T_{\mathrm{init}}\left(\frac{\gamma_{\mathrm{tree}} \gamma_{\mathrm{spanner}}m}{k}\right) \\
&\le m e^{O(\log^{83/84} m \log\log m)} + e^{O(\log^{83/84} m \log\log m)} \cdot T_{\mathrm{init}}\left(\frac{e^{O(\log^{83/84} m \log\log m)} \cdot m}{k}\right) \\
&\le m e^{O(\log^{83/84} m \log\log m)} + e^{O(\log^{83/84} m \log\log m)} \cdot T_{\mathrm{init}}\left(\frac{m}{e^{O(\log^{167/168} m \log\log m)}}\right) \\
&\le m e^{O(\log^{83/84} m \log\log m)}
\end{align*}
by our choice of $k.$
This also shows that the number of recursion levels is $d = O(\log^{1/168} m)$.

To bound the amortized update time using \Cref{thm:TreeCycleToCycle}, we need to know the vertex congestion of the flat forest maintained by our recursive data structure.
We notice that after a level of \Cref{thm:CycleToTreeCycle} and \Cref{thm:TreeCycleToCycle}, the vertex congestion increased by a factor of $2 \gamma_{\mathrm{tree}} = e^{O(\log^{83/84} m \log\log m)}.$
Therefore, the final vertex congestion after $d$ levels is at most
\begin{align*}
    \gamma_{\mathrm{vcong}} \le (2 \gamma_{\mathrm{tree}})^d = e^{O(\log^{167/168} m \log\log m)}
\end{align*}

We can also bound the amortized update time as follows:
\begin{align*}
T_{\mathrm{upd}}(m)
&\le \gamma_{\mathrm{tree}} \gamma_{\mathrm{spanner}} \left(k^2\gamma_{\mathrm{vcong}} + \frac{k}{m}T_{\mathrm{init}}\left(\frac{\gamma_{\mathrm{tree}} \gamma_{\mathrm{spanner}}m}{k}\right) + T_{\mathrm{upd}}\left(\frac{\gamma_{\mathrm{tree}} \gamma_{\mathrm{spanner}}m}{k}\right)\right) \\
&\le k^2\gamma_{\mathrm{vcong}} + (\gamma_{\mathrm{tree}} \gamma_{\mathrm{spanner}})^2 k + \gamma_{\mathrm{tree}} \gamma_{\mathrm{spanner}}T_{\mathrm{upd}}\left(\frac{\gamma_{\mathrm{tree}} \gamma_{\mathrm{spanner}}m}{k}\right) \\
&\le e^{O(\log^{167/168} m \log\log m)} + e^{O(\log^{83/84} m \log\log m)} \cdot T_{\mathrm{upd}}\left(\frac{m}{e^{O(\log^{167/168} m \log\log m)}}\right) \\
&\le e^{O(\log^{167/168} m \log\log m)}
\end{align*}

Now, we bound the approximation ratio of the output.
Each level of recursion, which includes one reduction of \Cref{thm:CycleToTreeCycle} and \Cref{thm:TreeCycleToCycle}, increases the ratio by a factor of $\kappahrg\gamma_{\mathrm{route}}\gamma_{\mathrm{spanner}} \le e^{O(\log^{83/84} m)}.$
After $d = O(\log^{1/168} m)$ levels of recursion, the approximation ratio becomes
\begin{align*}
    e^{O(d \cdot \log^{83/84} m)} = e^{O(\log^{167/168} m)}
\end{align*}

The cycle output by the data structure can be represented as $\gamma_{\mathrm{spanner}}$ tree paths on the final flat forest $F^G$ and $\gamma_{\mathrm{spanner}}$ off-tree paths.
This concludes the guarantees of our dynamic min-ratio cycle data structure $\calD$ with our choice of $\gamma_{MRC} = e^{O(\log^{167/168} m \log\log m)}.$
\end{proof}

To use the min-ratio data structure in the IPM framework, we need additional data structures that maintain a current (flow) solution, augment it with the output cycle, and detect large changes after augmentations.
Fortunately, since the cycle is compactly supported on a flat forest maintained by the min-ratio cycle data structure, we can leverage powerful dynamic tree data structures such as link-cut trees (\Cref{lem:dyn_trees}).

\begin{theorem}[Flow maintenance]
    \label{thm:flow_maint}
    There is a data structure that given a dynamic graph $G = (V, E, \ll, \gg, \cc)$ and forest $F^G$ with a flat embedding from $F^G$ to $G$ with edges congestion $\gamma_{\mathrm{MRC}}$ and some parameter $\epsilon > 0$ stores some implicit flow vector $\ff^{\mathrm{im}}$ and an explicit flow vector $\ff^{\mathrm{ex}}$ and receives updates of the form
    \begin{itemize}
        \item $\textsc{ApplyCycle}(E_C, \mu)$: Receives a cycle $C$ represented by the off-tree edges $E_C$ and tree paths in $F^G$. Sends $\mu$ flow along the formed tree paths in $F^G$ by updating $\ff^{\mathrm{im}}$, and returns $(\ff^{\mathrm{im}} + \ff^{\mathrm{ex}})^\top \vecone_{C}$ as well as a set of edges $E' \subseteq E$ for which $\ff^{\mathrm{ex}}(E') \gets \ff^{\mathrm{ex}}(E') +  \ff^{\mathrm{im}}(E')$ and $\ff^{\mathrm{im}}(E') \gets \veczero$. 
        \item Whenever an edge $e$ in $F^G$ is updated, it identifies the edge $e'$ in $G$ it maps to  and sets $\ff^{\mathrm{ex}}(e') \gets \ff^{\mathrm{ex}}(e') +  \ff^{\mathrm{im}}(e')$ and $\ff^{\mathrm{im}}(e') \gets \veczero$. 
    \end{itemize} 
    It maintains that $|\ff^{\mathrm{im}}(e) \ll(e)| \leq \epsilon |\ff^{\mathrm{ex}}(e)|$ and that the total number of edges returned after $t$ updates to $F_G$ and total flow weight sent $W$ is $W/\epsilon + t$ where the weight of a flow is the sum over the weight of its edges which is obtained by multiplying the amount of flow with the length of the edge. The runtime after $t$ operations is at most $\tilde{O}(t \gamma_{\mathrm{MRC}} + r \gamma_{\mathrm{MRC}})$ where $r$ is the total number of returned edges. 
\end{theorem}
\begin{proof}
    Follows from maintaining the implicit flows on link-cut trees as in \Cref{lem:dyn_trees}. 
\end{proof}

Next, we combine the data structures presented in this section. We obtain a solver data structure, which allows us to solve (incremental) min-cost flow. We first define the data structure, and then state our result. 

\begin{definition}[Solver]
    \label{def:solver_ds}    
    We call a data structure $\mathcal{D} = \textsc{Solver}(G, \ll, \gg, \cc, \ff, q, \Gamma, \epsilon)$ that is initialized with 
\begin{itemize}
    \item a graph $G = (V, E)$ and
    \item lengths $\ll \in \R^{E}_{\geq 0}$, gradients $\gg \in \R^{E}$, costs $\cc \in \R^{E}$, a flow $\ff \in \R^{E}$ routing demand $\dd$, and 
    \item a quality parameter $q > 0$, a step-size parameter $\Gamma > 0$ and a accuracy parameter $\epsilon > 0$. 
\end{itemize}
a $\gamma_{\mathrm{approx}}$ min-ratio cycle solver if it (implicitly) maintains a flow vector $\ff$ such that $\ff$ routes demand $\dd$ throughout and supports the following operations. 
\begin{itemize}
    \item $\textsc{ApplyCycle}()$: One of the following happens.  
        \begin{itemize}
            \item Either the data structure finds a circulation $\DDelta \in \R^{E}$ such that $\gg^\top \DDelta / \norm{\LL \DDelta}_1 \leq -q$ and $|\gg^\top \DDelta| = \Gamma$. In that case it updates $\ff \gets \ff + \DDelta$ and it returns a set of edges $E' \subseteq E$ alongside the maintained flow values $\ff(e')$ for $e' \in E'$. 

            For every edge $e$, between the times that it is in $E'$ during calls to $\textsc{ApplyCycle}()$, the value of $\ll(e)\ff(e)$ does not change by more than $\epsilon$.
            \item Or it certifies that there is no circulation $\DDelta \in \R^{E}$ such that $\gg^\top \DDelta / \norm{\LL \DDelta}_1 \leq -q/\gamma_{\mathrm{approx}}$ for some parameter $\gamma_{\mathrm{approx}}$.
        \end{itemize}
    \item $\textsc{UpdateEdge}(e, l, g)$: Updates the length and gradient of an edge $e$ that was returned by the last call to $\textsc{ApplyCyle}()$.
    \item $\textsc{InsertEdge}(e, l, g, c)$: Adds edge $e$ to $G$ with length $l$, gradient $g$, cost $c$. The flow $\ff(e)$ is intialized to $0$.
    \item $\textsc{ReturnCost}()$: Returns the flow cost $\cc^\top \ff$.
    \item $\textsc{ReturnFlow}()$: Explicitly returns the currently maintained flow $\ff$. 
\end{itemize}
The sum of the sizes $|E'|$ of the returned sets by $t$ calls to $\textsc{ApplyCycle}()$ is at most $t\Gamma/\epsilon q$. 
\end{definition}

\begin{theorem}
\label{thm:solver_ds}
There is a data structure $\mathcal{D} = \textsc{Solver}(G, \ll, \gg, \cc, \ff, q, \Gamma, \epsilon)$ as in \Cref{def:solver_ds} for $\gamma_{\mathrm{approx}} = e^{O(\log^{167/168} m)}$.
It has the following time complexity for $\gamma_{\mathrm{time}} = e^{O(\log^{167/168} m \log \log m)}$ where $m$ is an upper bound on the total number of edges in $G$:
\begin{enumerate}
    \item The initialization takes time $|E| \gamma_{\mathrm{time}}$.
    \item The operation $\textsc{ApplyCycle}()$ takes amortized time $(|E'| + 1) \cdot \gamma_{\mathrm{time}}$ when $|E'|$ edges are returned.
    \item The operations $\textsc{UpdateEdge}()$/$\textsc{InsertEdge}()/\textsc{ReturnCost}()$ take amortized time $\gamma_{\mathrm{time}}$. 
    \item The operation $\textsc{ReturnFlow}()$ takes amortized time $|E|\cdot \gamma_{\mathrm{time}}$.
\end{enumerate}
\end{theorem}
\begin{proof}
    Follows from \Cref{thm:ds_query} and \Cref{thm:flow_maint} since the total flow weight is at most $t\Gamma/q$ after $t$ updates. 
\end{proof}

We finally formally define the min-cost flow problem. 

\begin{definition}[Min-cost flow]
    Given a directed graph $G = (V, E)$ with edge capacities $\uu$ and costs $\cc$ and a demand vector $\dd \perp \vecone$, the min-cost flow problem seeks to find 
    \begin{equation*}
        \ff^{\star} = \argmin_{\ff: \BB^\top\ff = \dd} \cc^\top \ff.
    \end{equation*}
\end{definition}

We then describe how the solver data structure is used to solve min-cost flow. 

\begin{theorem}[Min-cost flow interior point method]
    \label{thm:static_ipm}
    There is an algorithm that given a directed graph $G = (V, E)$ with polynomially bounded edge capacities $\uu$ and costs $\cc$, a demand vector $\dd \perp \vecone$, and access to a $\gamma_{\mathrm{approx}}$ min-ratio cycle solver data structure (as defined in \Cref{def:solver_ds}) exactly solves the min-cost flow problem where:
    \begin{itemize}
        \item The number of calls to $\mathcal{D}.\textsc{ApplyCycle}()$/$\textsc{UpdateEdge}()$/$\textsc{InsertEdge}()/\textsc{ReturnCost}()$ is bounded by $\tilde{O}(m \gamma_{\mathrm{approx}}^{O(1)})$
        \item The number of calls to $\mathcal{D}.\textsc{ReturnFlow}()$ is bounded by $\tilde{O}(\gamma_{\mathrm{approx}}^{O(1)})$. 
    \end{itemize}
\end{theorem}
This result follows from previous work \cite{chen2022maximum} as discussed in \Cref{sec:ipm}.

Recall, we defined incremental threshold min-cost flow in \Cref{def:incrementalmincostflow}.
Next, we formalize how the solver data structure can be used to solve the incremental threshold min-cost flow problem. 

\begin{theorem}[Incremental threshold min-cost flow interior point method] \label{thm:inc_to_solver}
    There is an algorithm that given a directed graph $G = (V, E)$ with polynomially bounded edge capacities $\uu$ and costs $\cc$, a demand vector $\dd \perp \vecone$, and access to a $\gamma_{\mathrm{approx}}$ min-ratio cycle solver data structure $\mathcal{D}$ (as defined in \Cref{def:solver_ds}), exactly solves the incremental min-cost flow problem where:
    \begin{itemize}
        \item The number of calls to $\mathcal{D}.\textsc{ApplyCycle}()$/$\textsc{UpdateEdge}()$/$\textsc{InsertEdge}()/\textsc{ReturnCost}()$ is bounded by $\tilde{O}(m \gamma_{\mathrm{approx}}^{O(1)})$
        \item The number of calls to $\mathcal{D}.\textsc{ReturnFlow}()$ is bounded by $\tilde{O}(\gamma_{\mathrm{approx}}^{O(1)})$. 
    \end{itemize}
\end{theorem}
This result also follows from previous work \cite{vdBrand23incr} and is discussed in \Cref{sec:ipm}.

\Cref{thm:inc_to_solver} and \Cref{thm:solver_ds} together imply our main result \Cref{thm:thresh_mc}, an almost-linear time algorithm for threshold min-cost flow.
\section{Hierarchical Routing Graphs, Trees and Min-Ratio Cycles}
\label{sec:hrg_ds}

In this section, we reduce the dynamic min-ratio cycle problem (\Cref{def:DynamicMinRatioCycle}) to $m^{o(1)}$ dynamic min-ratio tree cycle problems (\Cref{def:DynamicMinRatioTreeCycle}).
That is, we show that one can maintain $m^{o(1)}$ trees such that the min-ratio tree cycle on these trees is a $m^{o(1)}$-approximate min-ratio cycle on the given graph.
We achieve the following result.

\HRG*

\subsection{Sparse Neighborhood Covers and Hierarchical Routing Graphs}
\label{subsec:snchrg}
The first step towards constructing the flat forest over $G$ is based on the $\ell_1$ oblivious routing of \cite{rozhon2022undirected}, which gives a way of extracting an oblivious routing from a collection of sparse neighborhood covers, one at each distance scale. Thus, the first piece that our data structure requires is an algorithm for maintaining a sparse neighborhood cover in a fully dynamic graph, which we show in \Cref{sec:SNC}.

\begin{restatable}[Fully-dynamic sparse neighborhood cover]{theorem}{dynSNC}
    \label{thm:dyn_snc}
Given an $m$-edge constant-degree input graph $G = (V,E,l)$ with polynomially-bounded lengths in $[1,L]$ and a diameter parameter $D \geq 1$ there is a data structure that supports a polynomially bounded number of updates of the following type:
\begin{itemize}
    \item $\textsc{InsertEdge}(e)/\textsc{DeleteEdge}(e)$: inserts/deletes edge $e$ to/from $G$ where insertions preserve that $G$ has constant degree.
\end{itemize}
Under these updates, the data structure maintains a forest $F$, map $\Pi_{V(F)\mapsto V}$, and a subset $S \subseteq V(F)$ that satisfy the following properties, for some $\gamma_{SNC} = e^{O(\log^{41/42} m \log\log m)}$:
\begin{enumerate}
\item \label{item:snclift} $F$ with map $\Pi_{V(F)\mapsto V}$ is a flat embedding of $G$.
\item \label{item:snccover} For any vertex $v \in V$ there is a tree $T \in F$ such that
\[ B_G(v, D/\gammasnc) \subseteq \Pi_{V(F)\mapsto V}(S \cap V(T)). \]
\item \label{item:sncdiam} For any tree $T \in F$ we have that $\diam_F(S \cap V(T)) \le \gammasnc \cdot D$.
\item \label{item:sncvcong} The congestion satisfies $\vcong(\Pi_{V(F)\mapsto V}) \le \gammasnc$.
\end{enumerate}
The forest $F$, map $\Pi_{V(F)\mapsto V}$, and subset $S \subseteq V(F)$ are all maintained explicitly, and each undergoes at most $\gammasnc$ amortized changes per update. The algorithm is deterministic, initializes in time $m \cdot \gammasnc$, and has amortized update time $\gammasnc$.
\end{restatable}
The reason for the set $S$ is that in our construction, the forest $F$ and trees $T \in F$ contain several extraneous vertices where we cannot control the diameter well. $S$ is the set of ``real vertices" which correspond to images of $v \in V(G)$ where we can control the diameter.
This way, the sets $\Pi_{V(F)\mapsto V}(S \cap V(T))$ correspond to the clusters in a sparse neighborhood cover, and the forest $F$ represents an approximate low-diameter tree connecting the vertices. As a subtle point, our data structure is only guaranteed to maintain $S$ explicitly, and \emph{not} each vertex in $S \cap V(T)$. This is because the algorithm for \Cref{thm:dyn_snc} requires detaching (potentially large) subtrees of $F$ and reattaching them to other subtrees.

Given this data structure, we are able to maintain an object, which we call a \emph{hierarchical routing graph} (HRG), that supports the oblivious routing of \cite{rozhon2022undirected}.

\begin{definition}[Hierarchical routing graph]
\label{def:HRG}
    Given a graph $G = (V, E, \ll, \gg)$ with polynomially bounded lengths $\ll$ in $[1, L = m^{O(1)}]$, a hierarchical routing graph (HRG) $H$ of $G$ depends on parameters $\gammahrg$, $\gammadiam$, $\kappahrg \ge 3$ and $\gammahrg \leq n$, so that $\gammadiam^{\kappahrg} \geq n^2 L$. Consider the following objects:
    \begin{itemize}
        \item a collection of vertex sets $\mathcal{V} = \{V_1, \ldots V_{\kappahrg}\}$ with bijective mapping between each $V_i$ and $V$ denoted by $\Pi_{V_i \mapsto V}$ for $\kappahrg = O(\log^{1/84} m)$,
        \item a collection of dynamic forests $\mathcal{F} = (F_1, \ldots, F_{\kappahrg - 1})$, called $\emph{routing forests}$, maps $\Pi_{V(F_i) \mapsto V}$, and subsets $S_i \subseteq V(F_i)$,
        \item a collection of edge sets $\cE = \cEout \cup \cEin$ for $\cEout := E^{out}_1 \cup \dots \cup E^{out}_{\kappahrg - 1}$ and $\cEin := E^{in}_2 \cup \dots \cup E^{in}_{\kappahrg}$ called linking edges and
        \item lengths and gradients $\ll_H$ and $\gg_H$ of all edges in $\cE$
    \end{itemize}
    These objects satisfy the following properties: 
    \begin{enumerate}
        \item \underline{Congestion:} \label{item:vconghrg} The vertex congestion $\vcong(\Pi_{V(F_i) \mapsto V}) \le \gammahrg$ for all $i \in [\kappahrg - 1]$,
        \item \underline{Routing trees are flat:} \label{item:flathrg} For all $i \in [\kappahrg - 1]$, $F_i$ and $\Pi_{V(F_i) \mapsto V}$ are a flat embedding of $G$.
        \item \underline{Routing tree diameter:} \label{item:diamhrg} For all $T \in F_i$, we have $\diam(S_i \cap T) \le \gammahrg \cdot \gammadiam^i$,
        \item \underline{Linking edge consistency:} \label{item:linkedge} Every out-edge $e \in \cEout$ is between a vertex $v \in V_i$ and $u \in S_i$, such that $\Pi_{V_i\mapsto V}(v) = \Pi_{V(F_i) \mapsto V}(u)$.
        Every in-edge $e \in \cEin$ is between a vertex $u \in S_{i-1}$ and $v \in V_i$ with $\Pi_{V(F_{i-1}) \mapsto V}(u) = \Pi_{V_i \mapsto V}(v)$.
        The length of edges $e \in E^{\mathrm{out}}_i$ for $i = 1, \dots, \kappahrg-1$ is $\ll_H(e) := \gammahrg \cdot \gammadiam^i$, and gradient is $\gg_H(e) = 0$. The length of edges $e \in E^{\mathrm{in}}_i$ for $i = 2, \dots, \kappahrg$ is $\ll_H(e) := \gammahrg \cdot \gammadiam^{i-1}$, and gradient is $\gg_H(e) = 0$.
        \item \label{item:outdegree} \underline{Outdegree:} Every vertex $v \in V_i$ is adjacent to at most $\gammahrg$ vertices in $F_i$.
        \item \label{item:hrgcover} \underline{Covering:} For all $v \in V$, we have $B_G(v,\gammadiam^i/\gammahrg) \subseteq \Pi_{V(F_i)\mapsto V}(S_i \cap T)$ for some $T \in F_i$.
        \end{enumerate}  
    The vertex set of $H$ consists of all $V_i$ and all $V(F_i)$. The edge set consists of all edges in all the $F_i$, and the linking edges.
    Given an additional edge $e = (u,v) \in E$ with length $\ll(e)$ and gradient $\gg(e)$ between two vertices $u, v \in V$, we let $H^e$ be the hierarchical routing graph $H$ with extra edge $e'$ between the two unique vertices $u', v' \in V_1$ that map to $u$ and $v$ respectively, and we set $\ll_H(e') = \ll(e)$ and $\gg_H(e') = \ll(e)$. For convenience, we refer to $e'$ as $e$ when the context is clear.
\end{definition}
\begin{figure}
    \centering
    \includegraphics[width=11cm]{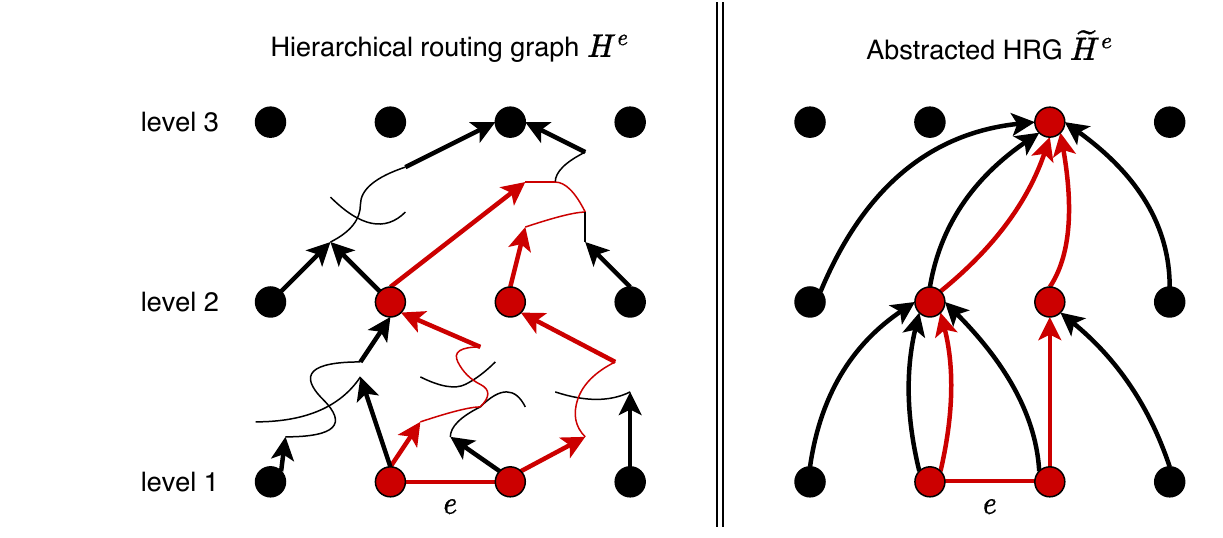}
    \caption{A hierarchical routing graph $H^e$ (\Cref{def:HRG}) and its abstracted HRG $\widetilde{H}^e$ (\Cref{def:abs_HRG}) alongside a monotone cycle circulation (\Cref{def:mon_cycle}) on $\widetilde{H}^e$ and its mapping on $H^e$. Notice that every tree is only attached to a single vertex on the next level. }
    \label{fig:HRG}
\end{figure}
Critically, note that $H$ flatly embeds into $G$, by the map $\Pi_{V(H) \mapsto V}$ which combines the maps $\Pi_{V(F_i) \mapsto V}$ and $\Pi_{V_i \mapsto V}$. Note that the linking edges $e \in \cE$ have endpoints with the same image, and hence still satisfy \Cref{def:lift}.

Next, we define an abstraction of the hierarchical routing graph $H$ that we never construct, but plays a crucial role in our analysis. It represents paths from $v \in V_i$ to $v \in V_{i+1}$ in $H$ as edges of length $\gammahrg\gammadiam^i$, which is (up to constants) the length of the underlying path in $H$ by properties \ref{item:diamhrg} and \ref{item:linkedge} in \Cref{def:HRG}. 

\begin{definition}[Abstracted HRG]
\label{def:abs_HRG}
    Given a hierarchical routing graph $H$ of a graph $G = (V,E,\ll,\gg)$, we define the abstracted HRG of $H$ as $\widetilde{H} = (\mathcal{V}, \tcE, \ll_{\widetilde{H}}, \gg_{\widetilde{H}})$, where the vertex set $\mathcal{V} = V_1 \cup V_2 \cup \dots \cup V_{\kappahrg}$, where $V_i$ are defined as in \Cref{def:HRG}. The edge set $\tcE$ and the lengths and gradients are defined as follows.
    \begin{itemize}
        \item \underline{Edges:} For every pair of edges $(u , x)$ and $(y, v)$ in $E(H)$ so that $u \in V_i$, $v \in V_{i + 1}$ and $x, y$ are in the same tree $T \in F_i$ the edge set $\tcE$ contains an edge $e = (u,v)$. 
        We let $\Pi_{\tilde{H} \mapsto H}(e)$ map the edge $e$ to the simple path between $u$ and $v$ given by concatenating the edge $(u, x)$ the path $T[x, y]$ and the edge $(y, v)$
        \item \underline{Lengths:} Edge $e$ as in the first bullet point receives length $\ll_{\widetilde{H}} = \gammahrg \cdot \gammadiam^i$. 
        \item \underline{Gradients:}  
        Edge $e$ as in the first bullet point receives gradient $\gg_{\widetilde{H}} = \gg_{H}(\Pi_{\tilde{H} \mapsto H}(e))$. 
    \end{itemize}
    Given an additional edge $e = (u,v) \in E$ between two vertices $u, v \in V$ of length $\ll(e)$ and gradient $\gg(e)$ we let $\widetilde{H}^e$ be the abstracted HRG $\widetilde{H}^e$ with an extra edge $e'$ between the two unique vertices $u', v' \in V_1$ that map to $u$ and $v$ respectively, and we set $\ll_{\widetilde{H}}(e') = \ll(e)$ and $\gg_{\widetilde{H}}(e') = \gg(e)$. For convenience, we refer to $e'$ as $e$ when the context is clear. Similarly, we denote $\widetilde{H}$ after the addition of multiple extra edges $\bar{E} \subseteq E$ as $\widetilde{H}^{\bar{E}}$ with respective lengths and gradients.
\end{definition}
\begin{remark}
    See \Cref{fig:HRG} for an example of a hierarchical routing graph $H$ and its abstracted HRG $\widetilde{H}$. The example illustrates that multi-edges may arise in $\widetilde{H}$ even if $H$ is a simple graph.
\end{remark}
Later, we will show that there is an oblivious routing (based on \cite{rozhon2022undirected}) which has the following form: for every vertex $v$, demand from $v$ is sent from its preimage in $V_1$ upwards on paths in any HRG that we construct.

\subsection{Dynamic Hierarchical Routing Graphs}

In this section, we show how we dynamically maintain a hierarchical routing graph $H$ using sparse neighborhood covers from \Cref{thm:dyn_snc} of increasing radius.

\paragraph{Dynamic hierarchical routing graph.} We state the main result of this section.

\begin{theorem}[Dynamic HRG]
\label{thm:dyn_hrg}
    Given an $m$-edge constant-degree input graph $G = (V, E, \ll, \gg)$ with polynomially-bounded lengths in $[1,L]$, there is a data structure that supports a polynomially bounded number of updates of the following type:
\begin{itemize}
    \item $\textsc{InsertEdge}(e)/\textsc{DeleteEdge}(e)$: inserts/deletes edge $e$ to/from $G$ where insertions preserve that $G$ has constant degree.
\end{itemize}
Under these updates, the data structure \textsc{DynamicHRG} maintains a hierarchical routing graph of $G$ with parameters $\gammahrg = e^{O(\log^{41/42} m \log\log m)}$, $\gammadiam = e^{O(\log^{83/84} m \log\log m)}$, and $\kappahrg = \log^{1/84} m$.
The algorithm is deterministic, can be initialized in time $O(m \gammahrg)$, and processes each update in amortized time $\gammahrg$.
\end{theorem}

\paragraph{Algorithm.}

We first describe the data structure \textsc{DynamicHRG} for maintaining a hierarchical routing graph. The algorithm follows the construction of the static $\ell_1$-oblivious routing of \cite{rozhon2022undirected}, but omits the computation of the distribution over the routing paths. The conditions of \Cref{def:HRG} will nearly be satisfied trivially given the construction.

Given a constant-degree dynamic graph $G = (V, E, \ll, \gg)$, we describe how our data structure maintains all the individual pieces of a hierarchical routing graph. 
\begin{enumerate}
    \item \underline{Vertex sets $\mathcal{V}$:} Copies $V_1, \ldots V_{\kappahrg}$ of the vertex set $V$ alongside bijective mappings $\Pi_{V_i \mapsto V}$ for all $i = 1, \ldots, \kappahrg$.
    \item \underline{Forests $\mathcal{F}$:} Dynamic SNC data structures $\mathcal{N}_1, \ldots, \mathcal{N}_{\kappahrg - 1}$ (\Cref{thm:dyn_snc}) operating on $G$, where $\mathcal{N}_i$ has diameter parameter $D_i := \gammadiam^i$. $F_i$ and $S_i$ in \Cref{def:HRG} is the forest and set maintained by $\mathcal{N}_i$.
    \item \underline{Linking edges $\cE$:} We first describe the sets $E_i^{\mathrm{out}}$ for $i = 1, \ldots, \kappahrg - 1$, and then the sets $E_i^{\mathrm{in}}$ for $i = 2, \ldots, \kappahrg$. 
    \begin{itemize}
        \item For every vertex pair $v \in V_i$ and $v' \in S_i$ that map to the same vertex in $V$, i.e., $\Pi_{V_i \mapsto V}(v) = \Pi_{V(F_i) \mapsto V}(v')$ we have an edge $(v, v')$ in $E_i^{\mathrm{out}}$. 
        \item For every $T \in F_{i-1}$, we add an edge to $E^{in}_i$ between an arbitrary vertex $v \in S_{i-1} \cap T$ and its copy in $V_i$, i.e, the unique vertex $v' \in V_i$ so that $\Pi_{V_i \mapsto V}(v') = \Pi_{V(F_{i - 1}) \mapsto V}(v)$.
    \end{itemize}
    \item \underline{Lengths and gradients:}
    We define the lengths $\ll_H$ and gradients $\gg_H$ of $H$ as in \Cref{def:HRG}. In particular, edges $e \in E(F_i)$ have the same gradient and longer lengths as their preimage in $G$, because $F_i$ flatly embeds into $G$. Lengths of linking edges $e \in \cE$ are described in \Cref{def:HRG}.
\end{enumerate}

The algorithm then reacts to changes in $G$  
by passing them to all the data structures $\mathcal{N}_i$. Since the forests $F_i$ and sets $S_i$ are maintained explicitly, the in-edges are maintained directly. To maintain the out-edges, we use a link-cut tree to detect some vertex in $S$ in each connected component of $F_i$ to add an out-edge from (see for example \cite{alstrup2005maintaining} on how to find a vertex from a specified set in a tree of a dynamic forest).

\paragraph{Correctness.}
Next, we show that the various parameters in the hierarchical routing graph are bounded as claimed in \Cref{def:HRG}.
\begin{lemma}
    Our algorithm maintains a hierarchical routing graph $H$ as in \Cref{def:HRG}. 
\end{lemma}
\begin{proof}
    We check the properties of \Cref{def:HRG} one at a time.
    \begin{itemize}
        \item Property~\ref{item:vconghrg} (congestion) follows from property~\ref{item:sncvcong} of \Cref{thm:dyn_snc} and the fact that each $F_i$ is maintained by $\mathcal{N}_i$.
        \item Property~\ref{item:flathrg} (flat routing trees) follows from property \ref{item:snclift} of \Cref{thm:dyn_snc}.
        \item Property~\ref{item:diamhrg} (routing tree diameter) follows from property \ref{item:sncdiam} of \Cref{thm:dyn_snc}.
        \item Property~\ref{item:linkedge} (linking edges) follows from our construction of linking edges.
        \item Property~\ref{item:outdegree} (out-degree) follows from property~\ref{item:sncvcong} of \Cref{thm:dyn_snc} and the way we construct $E^{out}_i$.
        \item Property~\ref{item:hrgcover} (covering) follows from property \ref{item:snccover} of \Cref{thm:dyn_snc}.
    \end{itemize}
\end{proof}

\paragraph{Runtime analysis.} We conclude with the runtime analysis of our dynamic HRG algorithm.
\begin{lemma}
    Initializing a dynamic hierarchical routing graph $H$ on a graph with $m$ edges takes time $\O(\kappahrg \cdot \gammasnc \cdot m) = e^{O(\log^{41/42} m \log\log m)} \cdot m$. Afterwards, each update can be processed in amortized time $\O(\kappahrg \cdot \gammasnc^2) = e^{O(\log^{41/42} m \log\log m)}$.
\end{lemma}
\begin{proof}
    The run-time is dominated by the cost of maintaining the forests $F_i$ and sets $S_i$ since each change to $F_i$ or $S_i$ causes only a constant number of easily computable changes in $\cE$. Therefore, the claimed runtime bound follows from \Cref{thm:dyn_snc}, and an overhead of $O(\log n)$ for using link-cut trees to detect adding out-edges.
\end{proof}

\subsection{Routings and Min-Ratio Circulations}
\label{sec:routing_quality}

In this section, we formally show that the edges in an HRG support an oblivious routing with competitive ratio at most $m^{o(1)}$. To do so, we walk through the construction of the $\ell_1$ oblivious routing of \cite{rozhon2022undirected} and analyze its competitive ratio in our setting. We will primarily work with the abstracted hierarchical routing graph $\tilde{H}$ (\Cref{def:abs_HRG}) in this section. We first define monotone paths on the $\tilde{H}$, which are the (compressed) paths upwards or downwards in the HRG.
\begin{definition}[Monotone path]
    We call a path $P$ on $\tilde{H}$ as in \Cref{def:abs_HRG} monotonically increasing (decreasing) if the layer index $i$ of the vertices it visits is strictly increasing (decreasing).
\end{definition}

We next state a theorem about a suitable notion of $\ell_1$ oblivious routings on $G$ represented by $\widetilde{H}$. These are restricted to route along monotone paths. 

\begin{theorem}
    \label{thm:routing_HRG}
    Given graph $G = (V, E, \ll, \gg)$ with edge-vertex incidence matrix $\BB$ and a hierarchical routing graph $H$ with abstraction $\widetilde{H}$ of $G$ the following hold: 
    \begin{enumerate}
        \item There is a set of monotonically increasing paths $\mathcal{P}$ in $\widetilde{H}$ that start at vertices in $V_1$ and end in $V_{\kappahrg-1}$, and a matrix $\PPi_{V \mapsto \mathcal{P}} \in \R^{|\mathcal{P}| \times |V|}$ mapping each vertex in $v \in V$ to a distribution of paths in $\mathcal{P}$ starting from its copy in $V_1$ such that all paths with nonzero weight starting from a vertex $v$ in connected component $A$ of the graph $G$ end at the same vertex $r_A$. 
        \item There is a matrix $\PPi_{P \mapsto \tcE} \in \R^{|\tcE| \times |\mathcal{P}|}$ mapping paths to their edges. 
        \item For $\gamma_{\mathrm{route}} := O(\kappahrg \gammahrg^6 \gammadiam)$, we have $\norm{\widetilde{\LL}\widetilde{\PP}\LL^{-1}}_{1 \rightarrow 1} \leq \gamma_{\mathrm{route}}$ for \[ \widetilde{\PP} := \PPi_{E \mapsto (V_1, V_1)} - \PPi_{\mathcal{P} \mapsto \tcE}\PPi_{V \mapsto \mathcal{P}}\BB^\top \] and $\widetilde{\LL} = \diag(\ll_{\widetilde{H}^E})$ where $\PPi_{E \mapsto (V_1, V_1)}$ maps edges $(u,v) \in E$ to edges between $(u', v')$ where $u = \Pi_{V_1\to V}(u')$ and $v = \Pi_{V_1 \to V}(v')$ (recall that $\Pi_{V_1\to V}$ is a bijection).
    \end{enumerate}
    We finally define the cycle projection matrix $\PP := \II - \PPi_{\tilde{H} \mapsto H}\PPi_{\mathcal{P} \mapsto \tcE}\PPi_{V \mapsto \mathcal{P}} \BB^\top $, where $\PPi_{\tilde{H} \mapsto H}$ is the matrix corresponding to the map $\Pi_{\tilde{H} \mapsto H}$ which sends edges in $\tilde{H}$ to paths in $H$.
\end{theorem}

Before we prove \Cref{thm:routing_HRG}, we state a corollary that makes the relation to min-ratio circulations explicit.

\begin{corollary}[HRG circulation]
    \label{cor:min_ratio_HRG}
    Given a graph $G = (V, E, \ll, \gg)$ with edge vertex incidence matrix $\BB$ and a HRG $H$ with abstraction $\widetilde{H}$ of $G$ we have 
    \begin{equation*}
        \min_{e \in |E|} \gg_{H}^\top  \PP \vecone_e/\norm{\widetilde{\LL}\widetilde{\PP} \vecone_e}_1 \leq \frac{1}{ \gamma_{\mathrm{route}}} \min_{\DDelta: \BB^\top \DDelta = \veczero} \gg^\top  \DDelta/\norm{\LL \DDelta}_1.
    \end{equation*}
\end{corollary}
\begin{proof}
    Setting $\MM = \widetilde{\LL}\widetilde{\PP}$ and instantiating $\PP$ as in \Cref{thm:routing_HRG} we obtain 
    \begin{equation*}
        \min_{e \in |E|} \gg_H^\top  \PP \vecone_e/\norm{\widetilde{\LL} \widetilde{\PP}  \vecone_e}_1 \leq \frac{1}{\gamma_{\mathrm{route}}} \min_{\DDelta: \BB^\top \DDelta = \veczero} \gg^\top \DDelta/\norm{\LL \DDelta}_1
    \end{equation*}
    directly by \Cref{lem:min_ratio} since $\norm{\widetilde{\LL}\widetilde{\PP}\LL^{-1}}_{1 \rightarrow 1} \leq \gamma_{\mathrm{route}}$ by \Cref{thm:routing_HRG}.
\end{proof}

\paragraph{Proof of \Cref{thm:routing_HRG}. }

We finally show that the abstracted routing graph $\widetilde{H}$ of $H$ supports a $\gamma_{\mathrm{route}}$-approximate $\ell_1$-oblivious routing to prove \Cref{thm:routing_HRG}. The proof closely follows section 4 of \cite{rozhon2022undirected} with the exception that we merely use their construction to prove the existence of a supported routing instead of extracting it. We note that the achieved approximation ratio in \cite{rozhon2022undirected} is polylogarithmic while ours is a large subpolynomial factor. This is caused by the additional overhead for maintaining sparse neighborhood covers dynamically, and that we force our HRG to have fewer layers so that we can extract $m^{o(1)}$ trees out of it later.

We first fix a hierarchical routing graph $H$ of $G$. To enable our analysis we define similar notation as \cite{rozhon2022undirected} to better interface with their results. 

\begin{definition}
    \label{def:l1_notation}
    Our analysis requires the definition of the following objects.  
    \begin{itemize}
        \item \underline{Clusters:} We let $\mathcal{C}_i := \{\Pi_{V(F_i) \mapsto V}(S_i \cap T)|T \in F_i\}$ denote the set $S_i \cap T$ over trees $T$ after mapping them to $G$. Recall that by property \ref{item:hrgcover} of \Cref{def:HRG} every ball of radius $\gammadiam^{i}/\gammahrg$ in $G$ is contained in a cluster $C \in \mathcal{C}_i$.
        \item \underline{Roots:} Recall that for every tree $T \in F_i$ (corresponding to a cluster $C \in \mathcal{C}_i)$ there is an unique edge $(u, v) \in E_i^{\mathrm{out}}$ with $u \in T$. We call $v \in V_{i+1}$ the \emph{root} of cluster $C$, and denote it as $r_C := v$.
        \item \underline{Edges and flows:} To simplify our notation, we assume an extra clustering $\mathcal{C}_0$ such that each vertex $v \in V_1$ is the root of its own singleton cluster in $\mathcal{C}_0$. The diameter of these clusters is $D_0 = \gammadiam^0 = 1$. Then, we let $e_{C, C'}$ be the edge from $r_C \in V_i$ to $r_{C'} \in V_{i + 1}$ in $\widetilde{H}$.\footnote{To disambiguate multi-edges, we take the one using the tree giving rise to cluster $C'$.}
    \end{itemize}
\end{definition}

In the following, we define the set of monotone paths $\mathcal{P}_v$ for every $v \in V_1$ on the abstracted routing graph $\widetilde{H}$ alongside a probability weighting $p_P$ for each $P \in \mathcal{P}_v$. We first define $p^i_C(v)$, which is proportional to the (weighted) amount of paths in $\mathcal{P}_v$ that end at $r_C$ for some cluster $C \in \mathcal{C}_i$.

\begin{definition}
    \label{def:routing_local}
    For $C \in \mathcal{C}_i$ where $i = 1, \ldots, \kappahrg - 2$ we let 
    \begin{align*}
        p^{(i)}_C(v) &:= \max\left\{0, \frac{\dist_G(v, V \setminus C)}{D_i} - 
        \frac{1}{4 \gammahrg}
        \right\} \\
        w_i(v) &:= \sum_{C \in \mathcal{N}_i} p^{(i)}_C(v)
    \end{align*}
    
    On the last level $\kappahrg - 1$ we choose a cluster $C_A \in \mathcal{C}_{\kappahrg - 1}$ per connected component $A$ of $G$ that contains all vertices in $A$. Such a cluster always exists since we choose $\kappahrg$ such that $\gammadiam^{\kappahrg - 1} \geq \gammahrg \cdot n \cdot L$ in \Cref{def:HRG}. We give vertices $v \in A$ weight $p_{C_A}^{\kappahrg - 1}(v) = 1$ and set $w_{\kappahrg - 1}(v) = 1$.
\end{definition}

Then, every vertex $v_i \in V_i$ at level $i$ that has flow originating from vertex $v \in V$ sends a $p^{i}_C(v)/w_i(v)$ fraction of said flow to $r_C$. We show that all such vertices are contained in the cluster $C$. Concretely, this yields the following path set. 

\begin{definition}[Paths and weights]
    \label{def:routing_global}
    Each flow path $P \in \mathcal{P}_v$ is given by $v = r_{C_0}$, $r_{C_1}$, $\ldots$, $r_{C_{\kappa_{\mathrm{HRG} - 1}}} = r$ where $C_i \in \mathcal{C}_i$. The weight of $P$ is given by 
    \begin{equation*}
        p_P := \prod_{i = 1}^{\kappahrg - 1} \frac{p_{C_i}^{(i)}(v)}{w_i(v)}.
    \end{equation*}
\end{definition}

\begin{remark}
    Notice that the choice of $p_{C_A}^{(\kappahrg - 1)}$ in \Cref{def:routing_local} ensures that all paths with nonzero weight end at the same vertex $r = r_{C_A}$ for all vertices in the same connected component $A$. 
\end{remark}

Thus we have described the routing projection matrix $\widetilde{\PP} = (\proj_{E \mapsto (V_1, V_1)} - \proj_{\mathcal{P} \mapsto \tilde{\mathcal{E}}}\proj_{V \mapsto \mathcal{P}})\BB^\top$ since we gave a mapping from vertices (in $V_1$ which is a copy of $V$) to distributions of paths and from paths to edges in $\widetilde{H}$. We first define the flow routed by a unit demand following \cite{rozhon2022undirected}. 

\begin{definition}
    We let $f_{C, C'}(v) := \frac{p^{(i)}_C(v)}{w_i(v)} \cdot \frac{p^{(i + 1)}_{C'}(v)}{w_{i + 1}(v)}$ be the amount of flow sent along edge $e_{C, C'}$ for a unit demand on $v$.
\end{definition}

We first show that all the edges $e_{C,C'}$ with nonzero flow $f_{C,C'}(v)$ are in $\widetilde{H}$. To do so, we show that $C \subseteq C'$ in that case. 

\begin{lemma}
    If $f_{C,C'}(v) \neq 0$ for some $C \in \mathcal{C}_i, C' \in \mathcal{C}_{i + 1}$ and $v \in V$ we have $C \subseteq C'$ given $\gammadiam \geq 8 \gammahrg$.
\end{lemma}
\begin{proof}
    For $i + 1 = \kappahrg - 1$ the claim follows directly. For other $i$, because $f_{C,C'}(v) \neq 0$ we have
    \begin{equation*}
        \dist_G(v, V \setminus C')/D_{i + 1} - 1/(4\gammahrg) > 0
    \end{equation*}
    we get that 
    \begin{equation*}
        \dist_G(v, V \setminus C') \geq D_{i + 1}/(4\gammahrg) \geq 2 D_{i}.
    \end{equation*}
    where the last inequality follows from $\gammadiam \geq 8 \gammahrg$ (see \Cref{thm:dyn_hrg}). By definition, if $p_C^{(i)}(v) \neq 0$ then $C$ is a cluster of diameter $D_i$ that contains $v$ and is thus contained in $C'$.
\end{proof}

We then state the main lemma of the analysis of \cite{rozhon2022undirected}.

\begin{lemma}[See Lemma 4.5 in \cite{rozhon2022undirected}]
    \label{lem:main_lemma_flow}
    For every $C \in \mathcal{C}_{i - 1}$ and $C' \in \mathcal{C}_{i}$ for $i = 1, \ldots, \kappahrg - 1$ we have 
    \begin{equation*}
        |f_{C, C'}(u) - f_{C, C'}(v)| \leq O(\gammasnc^2 \dist_G(u,v)/D_{i - 1})
    \end{equation*}
\end{lemma}

Given \Cref{lem:main_lemma_flow} a decomposition of the flow into shortest paths yields the following lemma. 

\begin{lemma}[See Corollary 4.5 in \cite{rozhon2022undirected}]
    \label{lem:final_bound_routing}
    Given $\gammadiam \geq 8 \gammahrg$ and any demand $\dd \in \R^{|V|}$, we have that 
    \begin{equation*}
        \sum_{i \in [\kappahrg - 1]} \sum_{C \in \mathcal{C}_{i - 1}, C' \in \mathcal{C}_i} \sum_{v \in V_1} |f_{C, C'}(v)\dd(v)| \cdot D_{i}
    \end{equation*}
    is at most $O(\kappahrg \gammahrg^5 \gammadiam)$ times the optimal routing length $\mathrm{OPT}(\dd)$ on $G$ routing demands $\dd$. 
\end{lemma}

Given \Cref{lem:final_bound_routing}, we conclude with a proof of \Cref{thm:routing_HRG}. 

\begin{proof}[Proof of \Cref{thm:routing_HRG}]
    We show the properties in $\Cref{thm:routing_HRG}$ one by one. 
    \begin{enumerate}
        \item The first property holds by inspecting \Cref{def:routing_local} and \Cref{def:routing_global}.
        \item The matrix mapping paths to their edges can be extracted from \Cref{def:routing_global} in a straightforward fashion. 
        \item We finally show an upper bound on $\norm{\widetilde{\LL}\widetilde{\PP}\LL^{-1}}_{1 \rightarrow 1}$. We have 
        \begin{align*}
        \norm{\widetilde{\LL}\widetilde{\PP}\LL^{-1}}_{1 \rightarrow 1} &\stackrel{(a)}{ =} \max_{\ff \in \R^{|E|}} \frac{\norm{\widetilde{\LL}\widetilde{\PP}\LL^{-1}\ff}_1}{\norm{\ff}_1} \\
        &\stackrel{(b)}{=} \max_{\tilde{\ff} \in \R^{|E|}} \frac{\norm{\widetilde{\LL}\widetilde{\PP}\tilde{\ff}}_1}{\norm{\LL\tilde{\ff}}_1} \\
        &\stackrel{(c)}{=} \max_{\tilde{\ff} \in \R^{|E|}} \frac{\norm{\widetilde{\LL}(\PPi_{E \mapsto (V_1, V_1)} - \PPi_{\mathcal{P} \mapsto \tcE}\PPi_{V \mapsto \mathcal{P}}\BB^\top)\tilde{\ff}}_1}{\norm{\LL\tilde{\ff}}_1} \\
        &\stackrel{(d)}{\le} 1 + \max_{\tilde{\ff} \in \R^{|E|}} \frac{\norm{\widetilde{\LL} \PPi_{\mathcal{P} \mapsto \tcE}\PPi_{V \mapsto \mathcal{P}}\BB^\top\tilde{\ff}}_1}{\norm{\LL\tilde{\ff}}_1} \\
        &\stackrel{(e)}{=} 1 + \max_{\dd: \dd^\top \vecone = 0} \frac{\norm{\widetilde{\LL} \PPi_{\mathcal{P} \mapsto \tcE}\PPi_{V \mapsto \mathcal{P}}\dd}_1}{\mathrm{OPT}(\dd)} \\
        &\stackrel{(f)}{\leq} 1 + \max_{\dd: \dd^\top \vecone = 0} \frac{\sum_{i \in [\kappahrg - 1]} \sum_{C \in \mathcal{C}_{i - 1}, C' \in \mathcal{C}_i} \sum_{v \in V_1} |f_{C, C'}(v)\dd(v)| \cdot O(\gammasnc \cdot  D_{i})}{\mathrm{OPT}(\dd)} \\
        &\stackrel{(g)}{\leq} O(\kappahrg \gammahrg^6 \gammadiam). 
        \end{align*}
        Equation (a) follows from the definition of the norm, and equation (b) follows by substituting $\ff$ with $\LL\ff = \tilde{\ff}$. Then (c) follows from the definition of $\widetilde{\PP}$, (d) follows by the triangle inequality and that for $e \in E$ corresponding to $e' \in (V_1, V_1)$, we have $\ll(e) = \ll_{\tilde{H}^E}(e')$, (e) follows by substituting $\dd$ for $\BB^\top\tilde{\ff}$ and realizing that if there was a flow $\tilde{\ff}'$ with $\norm{\LL \tilde{\ff}'} \leq \norm{\LL \tilde{\ff}}$ routing demands $\dd$ then this one has higher objective value. Inequality (f) follows from the definition of $f_{C, C'}(v)$ and $\widetilde{\LL}$. Finally, inequality (g) follows from \Cref{lem:final_bound_routing}.
    \end{enumerate}
\end{proof}
    
\subsection{Flow Decomposition}

\Cref{cor:min_ratio_HRG} ensures that there is an edge $e \in E$ such that some circulation in $H^e$ which consists of $e' \in (V_1, V_1)$ and a distribution over monotone paths has a good competitive ratio. We refer to such circulations as \emph{routing circulations} (\Cref{def:routing_circ}). We then describe a way of decomposing these routing circulations into simple cycles that consist of a monotone increasing path, decreasing path, and possible edge $e$. These are more suitable for building a fast query data structure. 
\begin{figure}
    \centering
    \includegraphics[width = 12cm]{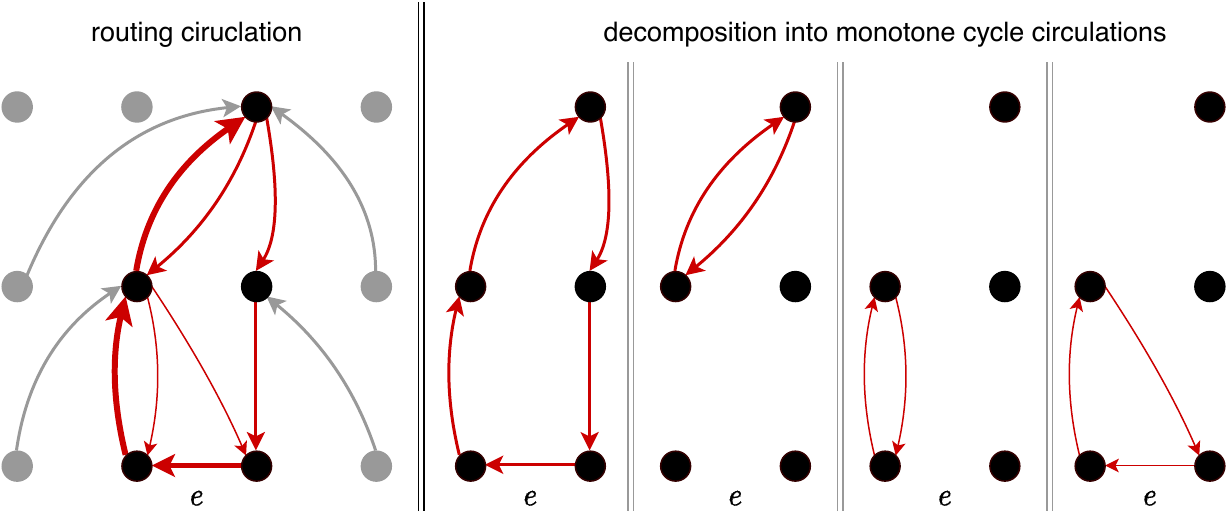}
    \caption{A routing circulation (\Cref{def:routing_circ}) depicted in on the left hand side in red. The thickness of the arrows corresponds to the amount of flow. On the right, we show a lossless decomposition of this flow into monotone cycles (\Cref{def:mon_cycle}). Notice that we incur some loss in the decomposition in general.}
    \label{fig:flow_decomp}
\end{figure}

Next, we formally define routing circulations on $\widetilde{H}^e$.
\begin{definition}[Routing circulation] 
    \label{def:routing_circ}
    Given an abstracted HRG $\widetilde{H}^e = (\mathcal{V}, \tcE, \ll_{\widetilde{H}^e})$ where $e = (u,v)$ we call the flow $\ff$ a routing circulation on $\widetilde{H}^e$ if 
    \begin{enumerate}
        \item the flow $\ff$ is a circulation, i.e. $\BB_{\widetilde{H}^e}^\top\ff = \veczero$ and
        \item the flow $\ff$ is composed of flow on the edge $e$, and a sum of path flows such that each path $P$ starts at either $u$ or $v$ and is monotonically increasing. Each path $P$ originating at $v$ carries some positive amount of flow, and each path originating at $u$ carries some negative amount of flow (or vice versa). Notice that a positive amount of flow sends flow up the hierarchy, and a negative amount sends flow down the hierarchy. 
    \end{enumerate}
\end{definition}
\begin{remark}
     Notice that $\widetilde{\mathcal{\PP}}\vecone_e$ as in \Cref{sec:routing_quality} is a routing circulation on $\widetilde{H}^e$. 
\end{remark}

Next we define monotone cycles, which are a crucial object in our algorithm. They are simple cycles in the abstracted HRG that first go monotonically up, and then down the hierarchy (or vice versa). They also may contain an extra edge $e \in (V_1, V_1)$.

\begin{definition}[Monotone cycles]
    \label{def:mon_cycle}
    Given a abstracted HRG $\widetilde{H}^e = (\mathcal{V}, \tcE, \ll_{\widetilde{H}^e})$
    with extra edge $e = (u,v)$ we define a monotone cycle to be a simple cycle composed of a monotonically increasing path, monotonically decreasing path, and optionally the edge $e$.
\end{definition}

Next, we show that any routing circulation (\Cref{def:routing_circ}) can be decomposed into cycle flows on monotone cycles (\Cref{def:mon_cycle}) without increasing the weight of the flow by more than a sub-logarithmic factor. See \Cref{fig:flow_decomp} for an illustration of such a decomposition. The reason that a lossless decomposition does not exist is because in the oblivious routing, flow that has merged into a vertex may split further down the line. However, because the lengths in the abstracted HRG are geometrically increasing upwards, we can intuitively charge a flow that merges and splits to the highest level where it splits.
\begin{figure}
    \centering
    \includegraphics[width = 7cm]{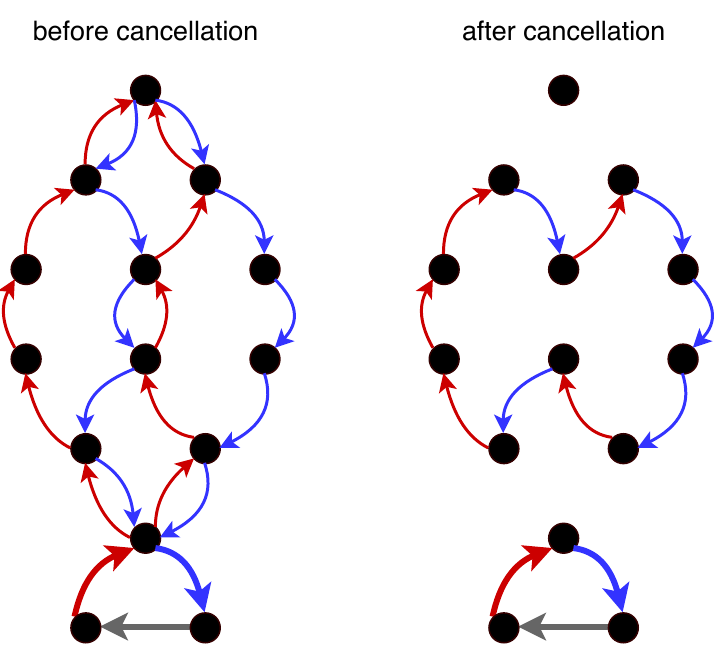}
    \caption{An example of a circulation on monotone paths that yields a non-monotone cycle after cancellation. Blue edges carry flow down, and red edges carry flow up. It is not possible to decompose this circulation into monotone cycles without increasing the flow weight. To achieve a good decomposition, we crucially exploit that the lengths of edges at level $i$ are $\gammahrg \gammadiam^i$, and thus edges between layers with high indices are much longer.}
    \label{fig:non_lossless}
\end{figure}

\begin{lemma}[Flow decomposition on abstracted HRG]
\label{lem:flow_decomp_abs}
    Given a routing circulation $\ff$ on an abstracted HRG $\widetilde{H}^e$ with extra edge $e = (u,v)$ there exists a decomposition $\ff = \sum_{i = 1}^k \cc_i$ for some $k$ such that:
    \begin{enumerate}
        \item each $\cc_i$ is a circulation supported on a single monotone cycle of $\widetilde{H}^e$, and 
        \item $\sum_{i = 1}^k \norm{\widetilde{\LL} \cc_i}_1 \leq 4 \kappahrg \norm{\widetilde{\LL} \ff}_1$.
    \end{enumerate} 
    where $\widetilde{\LL} = \diag(\ll_{\widetilde{H}^e})$.
\end{lemma}
\begin{proof}
    We iteratively decompose a routing circulation $\ff$ on $\widetilde{H}^e$. Throughout, we will let the ``mass'' of a circulation or path flow mean the amount on a single edge of the flow. Assume without loss of generality that there is one unit of flow on edge $e$ from $v$ to $u$. Because $\ff$ is a routing circulation (\Cref{def:routing_circ}), the part of $\ff$ restricted to $H$, i.e., without the edge $e$, can be written as the sum of flows representing a distribution over monotonically increasing paths from $u$ (which we denote as $\mathcal{P}_u$), and monotone decreasing paths into $v$ (which we denote as $\mathcal{P}_v$).
    Now we iteratively decompose the flow into circulations. For each $i = \kappahrg, \dots, 1$ we will build a collection of monotone cycle circulations $\mathcal{M}_i$ and circulation $\ff_i$ satisfying:
    \begin{itemize}
        \item $\ff_i + \sum_{j=i}^{\kappahrg} \sum_{\cc \in \mathcal{M}_j} \cc = \ff$.
        \item Every edge $e'$ adjacent to some vertex $v' \in V_j$ for $j\ge i+1$ has $\ff_i(e') = 0$.
        \item $\ff_i$ is a routing circulation (\Cref{def:routing_circ}). Let $\mathcal{P}_u^i$ be the monotone increasing paths from $u$, and $\mathcal{P}_v^i$ be the monotone decreasing flow paths into $v$, which only contain vertices in $V_1, \dots, V_i$.
        \item $\|\tilde{\LL}\ff_i\|_1 \le \|\tilde{\LL}\ff\|_1$.
        \item $\sum_{\cc \in \mathcal{M}_i} \|\tilde{\LL}\cc\|_1 \le 4\|\tilde{\LL}\ff\|_1$.
    \end{itemize}
These imply the lemma by the second and the final bullet. The original flow $\ff$ satisfies these properties for $i = \kappahrg$. Now, given a flow $\ff_{i+1}$ satisfying these properties, we will construct a flow $\ff_i$ and monotone cycle circulations $\mathcal{M}_i$. We first define $\tcE_i := \{(u',v') \in \tcE| u' \in V_i, v' \in V_{i + 1}\}$ to be the set of edges in $\widetilde{H}^e$ between $V_i$ and $V_{i+1}$. Let $F_i$ be the total absolute value of the flow $\ff_{i+1}$ on $\tcE_i$, i.e., 
\[ F_i = \sum_{e' \in \tcE_i} |\ff_{i + 1}(e')|. \]
While there is some nonzero $\ff_{i+1}(e')$ for $e' \in \tcE_i$, do the following. Let $e_1$ have minimal nonzero $|\ff_{i+1}(e_1)|$. Without loss of generality, assume that it has positive flow from $v_1 \in V_i$ to $y$ in $V_{i+1}$.
Because $\ff_{i+1}$ is a circulation with no flow on edges from $V_{i+1}$ to $V_{i+2}$ by induction, there must be positive flow on some $e_2$ from $y$ to $v_2 \in V_i$. Because $\ff_{i+1}$ is a routing circulation by induction (see bullet 3 above), we know that there is a path from $u$ to $v_1$ to $y$, which we call $P_u \in \mathcal{P}_u^{i+1}$, and a path from $y$ to $v_2$ to $v$, which we call $P_v \in \mathcal{P}_v^{i+1}$. Let $\mu$ the smallest among $|\ff_{i+1}(e_1)|$, and the mass of $P_u$ or $P_v$. We will peel off a circulation $\cc$ depending on two cases:
\begin{itemize}
    \item If $P_u, P_v$ intersect at $x \neq y$, then set $\cc$ to be the circulation of mass $\mu$ from $x \to v_1 \to y \to v_2 \to x$. Remove mass $\mu$ from path $P_u$, and add mass $\mu$ of path $P_u[u,x]$ to $\ff_i$. Remove mass $\mu$ from $P_v$, and add mass $\mu$ of path $P_v[x, v]$ to $\ff_i$.
    \item If $P_u, P_v$ do not intersect other than at $y$, let $\cc$ denote $\mu$ mass of the cycle from $v \to u \to v_1 \to y \to v_2 \to v$. Remove $\mu$ flow from $v$ to $u$, and $\mu$ mass from both paths $P_u$ and $P_v$.
\end{itemize}
The process terminates because during each iteration we either set some $\ff_{i+1}(e') = 0$, or shorten some path in $\mathcal{P}_u^{i+1}$ or $\mathcal{P}_v^{i+1}$ to under level $i$.

Let us check the five conditions of the iterative procedure. The first and second condition follow by construction. For the third condition, define $\mathcal{P}_u^i$ as the union of: (1) paths $P_u[u, x]$ when paths $P_u, P_v$ intersected at $x \neq y$, (2) paths remaining in $\mathcal{P}_u^{i+1}$ shortcut to level $i$. Define $\mathcal{P}_v^i$ analogously. These paths show that $\ff_i$ is a routing circulation by construction.

To check the fourth, note that by setting all flows on $\tcE_i$ to $0$ in $\ff_{i+1}$, we decrease $\ell_1$-cost $\ff_{i+1}$ (i.e., $\|\LL \ff_{i+1}\|_1$) by $F_i\gammahrg\gammadiam^i$. For edges in lower layers, the total weight increase is $F_i \sum_{j<i} \gammahrg\gammadiam^j \le 2F_i\gammahrg\gammadiam^{i-1}$. Also, $|\ff_i(e)| \leq |\ff_{i+1}(e)|$. So $\|\tilde{\LL}\ff_i\|_1 \le \|\tilde{\LL}\ff_{i+1}\|_1 \le \|\tilde{\LL}\ff\|_1$.
The total weight of edges in the cycles $\cc \in \mathcal{M}_i$ for edges in $\tilde{H}$ is at most $2F_i \sum_{j\le i} \gammahrg\gammadiam^j \le 4F_i \gammahrg\gammadiam^i.$ The total contribution of edge $e$ is also at most its original flow in $\ff_{i+1}$. So $\sum_{\cc\in\mathcal{M}_i} \|\tilde{\LL}\cc\|_1 \le 4\|\tilde{\LL}\ff_{i+1}\|_1 \le 4\|\tilde{\LL}\ff\|_1.$
\end{proof}
\begin{remark}
    It is not possible to decompose the flow in a lossless manner and exploiting that the lengths are geometrically increasing upwards is crucial. This is illustrated in \Cref{fig:non_lossless}. 
\end{remark}
The flow decomposition in \Cref{lem:flow_decomp_abs} directly translates to a decomposition in the HRG $H$.  
\begin{definition}[Flow mapping and monotone cycles on $H$]
\label{def:flow_map_mon}
    Given a routing circulation $\ff$ on the abstracted HRG $\widetilde{H}^e$ we let $\ff_{H^e}$ denote the flow mapped via $\Pi_{\widetilde{H} \mapsto H}$ i.e. for $e' \neq e$ we have 
    \begin{equation*}
        \ff_{H^e}(e') = \sum_{e'' \mathrm{ s.t. } e' \in \Pi_{\widetilde{H} \mapsto H}(e'')} \ff(e'')
    \end{equation*}
    and $\ff_{H^e}(e) = \ff(e)$. We define the monotone cycle circulations on $H^e$ as the set of $\ff_{H^e}$ for monotone cycle circulations $\ff$ on $\tilde{H}^e$.
\end{definition}

We conclude that we can decompose a flow into monotone cycles on a HRG $H^e$ given an upper bound on its weight on $\widetilde{H}^e$. This upper bound translates to the decomposition on $H^e$ since we chose edge lengths in the abstracted HRG $\widetilde{H}$ as upper bounds of corresponding path lengths in $H$.
\begin{claim}
\label{claim:upto3}
For an edge $(u, v) \in \tilde{H}$ with $u \in V_i, v \in V_{i+1}$, the corresponding path in $H$ has length at most $3\gammahrg\gammadiam^i$.
\end{claim}
\begin{proof}
Let the path in $H$ be composed of edges $(u, x)$, a path from $x$ to $y$ in the a tree $T \in F_i$, and $(y, v)$. By construction of the linking edges $\cE$, both $(u, x)$ an $(y, v)$ have length $\gammahrg\gammadiam^i$. We know that $x, y \in S_i \cap T$ by property \Cref{item:linkedge} of \Cref{def:HRG}, so $\mathrm{dist}_T(x, y) \le \diam(S_i \cap T) \le \gammahrg\gammadiam^i$. Combining these gives the desired bound.
\end{proof}
Together with the results from the previous sections, this suffices to show that a monotone cycle of some $H^e$ is an approximate min-ratio circulation.
\begin{corollary}
    \label{decomp:min_ratio}
    Let $H$ be a hierarchical routing graph of $G = (V, E, \ll, \gg)$. Then, there exists some $e$ and monotone cycle circulation $\cc$ on $H^e$ such that 
    \begin{equation*}
      \gg_{H}^\top  \cc  /\norm{\LL_H\cc}_1 \leq \frac{1}{O(\kappahrg \gamma_{\mathrm{route}})} \min_{\DDelta: \BB^\top \DDelta = \veczero} \gg^\top \DDelta/\norm{\LL \DDelta}_1.
    \end{equation*}
\end{corollary}
\begin{proof}
    By \Cref{cor:min_ratio_HRG} there exists some routing circulation $\ff$ on some $\tilde{H}^e$ such that 
    \begin{equation*}
        \gg_{H}^\top  \PPi_{\widetilde{H} \mapsto H} \ff /\norm{\widetilde{\LL}\ff}_1 \leq \frac{1}{\gamma_{\mathrm{route}}} \min_{\DDelta: \BB^\top \DDelta = \veczero} \gg^\top  \DDelta/\norm{\LL \DDelta}_1
    \end{equation*}
    Notice that the quantities on both sides of the inequality are negative without loss of generality. Then, we have  
    \begin{equation*}
        \frac{\sum_{i = 1}^k \gg_H^\top \PPi_{\widetilde{H} \mapsto H} \cc_i}{\sum_{i = 1}^k \norm{\widetilde{\LL} \cc_i}_1} \leq \frac{\gg_{H}^\top \ff}{4\kappahrg\norm{\widetilde{\LL} \ff}_1}
    \end{equation*}
    by \Cref{lem:flow_decomp_abs}. We then use the standard averaging inequality
    \begin{equation*}
        \min_{i \in [k]} \frac{\gg_H^\top \PPi_{\widetilde{H} \mapsto H} \cc_i }{\norm{\LL_H \cc_i}_1} \leq \frac{\sum_{i = 1}^k \gg_H^\top \cc_i}{\sum_{i = 1}^k \norm{\LL_H \cc_i}_1}.
    \end{equation*}
    Finally, we let $j$ be a minimizer of the left-hand side of the previous inequality, and define $\cc' = \PPi_{\widetilde{H} \mapsto H} \cc_j$. We then have 
    \begin{equation*}
         \frac{\gg_H^\top \cc'}{\norm{\LL_H \cc'}_1} \leq \frac{1}{3} \min_{i \in [k]} \frac{\gg_H^\top \Pi_{\widetilde{H} \mapsto H}(\cc_i)}{\norm{\LL_H \cc_i}_1}
    \end{equation*}
    by \Cref{def:HRG} and \Cref{def:abs_HRG} since the lengths on $\widetilde{H}$ are upper bounds for the lengths of the embedded paths in $H$ up to a factor of $3$ by \Cref{claim:upto3}. This concludes the proof of this corollary since $\cc'$ is a monotone cycle circulation on $H^e$ by \Cref{def:flow_map_mon}. 
\end{proof}
Recall that the HRG $H$ flatly embeds into $G$. Thus, the monotone cycle circulation $\cc$ on $H$ is also a circulation in $G$ via the edge map $\Pi_{V(H) \to V(G)}$.

\subsection{Tree Decomposition and Maintenance}

In this section, we show that a hierarchical routing graph $H$ of a graph $G = (V, E, \ll, \gg)$ can be decomposed into a collection of $m^{o(1)}$ many subtrees $\mathcal{T}$ of $H$ and an easy-to-maintain set of off-tree edges such that every monotone cycle on a graph $H^e$ for every edge $e \in E$ is represented by a tree cycle for some off-tree edge $e'$ in some tree $T \in \mathcal{T}$. Together with the previous section, this reduces the query to solving the problem on tree cycles.
\begin{definition}[Hierarchical routing tree]
    \label{def:HRT}
    We call graph $T$ a \emph{hierarchical routing tree} (HRT) if it is a hierarchical routing graph (\Cref{def:HRG}) and additionally every vertex in $V_i$ is adjacent to a single vertex in $F_i$, i.e., the out-degree (property \ref{item:outdegree} in \Cref{def:HRG}) is upper bounded by $1$ instead of $\gammahrg$. 
\end{definition}

We first define the sets of off-tree edges.  
\begin{definition}
\label{def:off_tree_edges}
    Given a hierarchical routing graph $H$ of a graph $G = (V, E, \ll, \gg)$ we define the sets of off-tree edges $E_{\mathrm{graph}}$ and $E_{\mathrm{pair}}$
    as follows:
    \begin{itemize}
        \item $E_{\mathrm{graph}}$ contains the images of edges $e \in G$ into $(V_1, V_1)$, i.e., \[ E_{\mathrm{graph}} = \{(\Pi_{V\to V_1}(u), \Pi_{V\to V_1}(v)) : (u, v) = e \in E(G) \}. \]
        The gradients and lengths are the same as in $G$.
        \item For all $v \in V_i$, and all pairs $(v, u_1), (v, u_2)$ of out-edges of $v$ with $u_1, u_2 \in F_i$, add the edge $e = (u_1, u_2)$ to $E_{\mathrm{pair}}$. Define $\ll(e) = 2\gammahrg\gammadiam^i$ and $\gg(e) = 0$.
    \end{itemize}
\end{definition}
Recall that there are two types of monotone cycles in $\tilde{H}^e$: those containing $e$, and those consisting only of edges in $\tilde{H}$. The former are naturally tree cycles in some tree induced by edges in $E_{\mathrm{graph}}$. Similarly, the edges in $E_{\mathrm{pair}}$ are used to turn cycles consisting only of edges in $\tilde{H}$ into tree cycles. Note that even after adding these edges to a HRG $H$, the result is still flat over $G$. This is because the $E_{\mathrm{graph}}$ clearly respect flatness, and by out-edge consistency (property \ref{item:linkedge} of \Cref{def:HRG}), we know that for a paired edge $(u_1, u_2) \in E_{\mathrm{pair}}$ coming from $v$, that $\Pi_{V(H) \to V(G)}(u_1) = \Pi_{V(H) \to V(G)}(u_2) = \Pi_{V(H) \to V(G)}(v)$.

We observe that the total number of edges in $E_{\mathrm{graph}}$ and $E_{\mathrm{pair}}$ is almost linear in the number of edges of $G$. 

\begin{lemma}
    \label{lem:num_edges}
    The number of edges $|E_{\mathrm{graph}}| + |E_{\mathrm{pair}}| \leq |E| + \kappahrg\gammahrg^2 |V|.$
\end{lemma} 
\begin{proof}
    We have $|E_{\mathrm{graph}}| = |E|$. A vertex $v \in V_i$ has outdegree at most $\gammahrg$ by property \ref{item:outdegree} of \Cref{def:HRG}, so it contributes at most $\gammahrg^2$ paired edges. The bound follows because the total number of vertices in $V_1, \dots, V_{\kappahrg}$ is $n\kappahrg$.
\end{proof}

In the next lemma, we show that we can construct a collection of at most $m^{o(1)}$ hierarchical routing trees that are subgraphs of a hierarchical routing graph $H$ so that every monotone cycle corresponds to a tree cycle formed by some tree and off-tree edge from either $E_{\mathrm{graph}}$ or $E_{\mathrm{pair}}$.

We first describe a randomized procedure for generating such trees. Our deterministic construction will derandomize this approach in a standard way. Given a hierarchical routing graph $H$, consider sampling a random edge between $v \in V_i$ and $u \in F_i$ for every vertex $v \in V_i$ and every $i$. This yields a hierarchical routing tree $T$ formed by dropping all edges from $H$ that were not sampled. For simplicity, consider a monotone cycle that includes an edge $e \in E$. Then, for every vertex $v \in V_1$ that is part of the cycle, we have a chance of at least $1/\gammahrg$ to choose the correct edge to $F_i$. Since the cycle contains at most $2 \kappahrg$ such vertices, the probability of preserving the min-circulation in the sampled tree $T$ is at least $1/\gammahrg^{2 \kappahrg}$. Therefore, we expect to find a tree containing the cycle after sampling roughly $\gammahrg^{2 \kappahrg}$ trees. To make this process deterministic, we exploit that the choices of the vertices at level $V_i$ can be highly correlated and then enumerate over all choices. 

We define the collection of routing trees. To do so we associate an identifier consisting of $\log_2 |V_i| = \log_2 |V|$ bits to each vertex $v \in V_i$.

\begin{figure}
    \centering
    \includegraphics[width = 9cm]{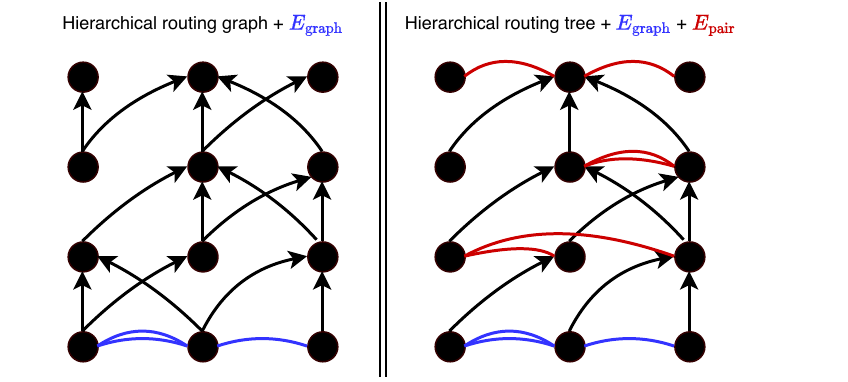}
    \caption{This figure displays a abstracted HRG $\tilde{H}$ with extra edges $E_{\mathrm{graph}}$, and one of its sub-trees. The red edges $E_{\mathrm{pair}}$ are generated by out-pairs in $\tilde{H}$. We generate a set of trees such that every monotone cycle is a tree cycle formed by some tree and an off-tree edge in $E_{\mathrm{pair}}$ or $E_{\mathrm{graph}}$.}
    \label{fig:enter-label} 
\end{figure}

\begin{definition}
\label{def:trees}
    Given a hierarchical routing graph $H$ of $G = (V, E, \ll, \gg)$, we define the tree collection $\mathcal{T}$ to consist of a hierarchical routing tree $T_{\pp, \aa, \bb}$ for every triple $(\pp, \aa, \bb) \in [\log_2 |V|]^{\kappahrg} \times [\gammahrg]^{\kappahrg} \times [\gammahrg]^{\kappahrg}$. We now describe the construction of $T_{\pp, \aa, \bb}$. The vertex $v \in V_i$ chooses its edge to $F_i$ as follows. If the $\pp(i)$-th bit of the identifier of $v$ is 1, then it chooses the $\aa(i)$-th edge between $v$ and $F_i$ for some arbitrary but consistent ordering of the edges. If on the other hand the $\pp(i)$-th bit is $0$, it chooses the $\bb(i)$-th edge between $v$ and $F_i$. Whenever the $\bb(i)$-th or $\aa(i)$-th edge does not exist, substitute the missing one with an arbitrary edge to maintain connectivity. This concludes the description of the trees $T_{\pp, \aa, \bb}$ and thus the set $\mathcal{T}$.
\end{definition}

We show that a hierarchical routing graph can be decomposed into a small set of trees such that each routing cycle is a tree cycle for one of the trees. 

\begin{lemma}[HRG decomposition]
    \label{lma:HRG_decom}
    Given a hierarchical routing graph $H$ of $G = (V, E, \ll, \gg)$ its off-tree edges $E_{\mathrm{graph}}$ and $E_{\mathrm{pair}}$ (see \Cref{def:off_tree_edges}) and the set $\mathcal{T}$ of tree sub-graphs of $H$ (see \Cref{def:trees}), for every monotone cycle circulation $\cc$ on $H^e$ for some edge $e \in E$ there exists an edge in $e' \in E_{\mathrm{graph}} \cup E_{\mathrm{pair}}$ so that $\cc = \beta \cdot \vecone_{\TCyc[e']}$ for some $T \in \mathcal{T}$ and $\beta \in \R$. 
\end{lemma}
\begin{proof}
    Let $\cc$ be a monotone cycle circulation in the abstracted HRG, with corresponding monotone cycle circulation $\cc_{H^e}$ in $H^e$. Note that a circulation on a cycle has the same amount of flow on each edge. We distinguish two similar but distinct cases. 
    \begin{enumerate}
        \item \underline{Case 1: The cycle contains an off-HRG edge $e = (u,v)$.} Decompose $\cc$ into two paths of some length $k$, plus edge $e$. Say that it contains vertices $u = u_1, u_2, \dots, u_k$ and $v = v_1, v_2, \dots, v_k$ where $u_i, v_i \in V_i$ and $u_k = v_k$. Thus, $\cc_{H^e}$ is a tree cycle of a tree $T \subseteq H$ (as in \Cref{def:trees}) if the out-edge of $u_i$ in $T$ goes to a connected component in $F_i$ with an in-edge to $u_{i+1}$ (and the analogous statement for $v_i, v_{i+1}$), for all $i = 1, \dots, k-1$. We construct a tuple $(\pp, \aa, \bb)$ so that $T_{\pp,\aa,\bb}$ satisfies this. 
        Since $u_i \neq v_i$, let $\pp(i)$ be so that the $\pp(i)$-th bit of $u_i$ and $v_i$ differ, for $i = 1, \dots, k-1$. Say the $\pp(i)$-th bit of $u_i$ is $1$ and $v_i$ is $0$. Then, let $\aa(i)$ be the identifier of the out-edge of $u_i$ to the connected component with an in-edge to $u_{i+1}$, and $\bb(i)$ be the identifier of the out-edge of $v_i$ to the connected component with an in-edge to $v_{i+1}$ for $i = 1, \dots, k-1$. If the $\pp(i)$-th bit of $u_i$ is $0$ and $v_i$ is $1$, reverse this. Evidently, this gives a tree $T_{\pp,\aa,\bb}$ (with off-tree edge $e$) containing $\cc_{H^e}$.
        This concludes the first case. 
        \item \underline{Case 2: the cycle is internal to $H$.} If the cycle is instead an internal monotone cycle of $H$ we consider the unique vertex $v \in V_i$ at the lowest level $i$ on the cycle. This vertex has two adjacent edges $(u_1, v)$ and $(v, u_2)$ on the cycle. There is an off-tree edge between $u_1$ and $u_2$ in $E_{\mathrm{pair}}$. Then the rest of the argument is analogous to the first case. 
    \end{enumerate}
\end{proof}

Next, we show that given our algorithm that maintains a hierarchical routing graph $H$, there is an efficient algorithm that maintains the sets $\mathcal{T}$.

\begin{lemma}[Tree maintenance]
\label{lem:DynHRGTrees}
    Given a dynamic HRG $H$ of a dynamic graph $G = (V, E, \ll, \gg)$ with polynomially-bounded lengths in $[1,L]$ and off HRG edges $E_{\textrm{graph}}$, there is a data structure that supports a polynomially bounded number of updates of the following type:
\begin{itemize}
    \item $\textsc{InsertEdge}_{H}(e)/\textsc{DeleteEdge}_{H}(e)$: inserts/deletes edge $e$ to/from $H$ after each update to $G$ such that $H$ remains a hierarchical routing of the dynamic graph $G$. 
    \item $\textsc{InsertOffEdge}(e)/\textsc{DeleteOffEdge}(e)$: adds/removes edge $e$ to/from $E_{\textrm{graph}}$
\end{itemize}
Under these updates, the data structure maintains a collection of rooted trees $\mathcal{T}$ and a set of auxiliary edges $E_{\mathrm{off}} = E_{\mathrm{graph}} \cup E_{\mathrm{pair}}$ as in \Cref{def:off_tree_edges}. Each tree $T \in \mathcal{T}$ is flat over $G$, explicitly represented by map $\Pi_{V(T)\to V(G)}$.
Furthermore
\begin{equation}
    \min_{T \in \mathcal{T}} \min_{\cc \in \{\vecone_{\TCyc[e]}, -\vecone_{\TCyc[e]}\}} \frac{\gg_T^\top \cc}{\norm{\LL_T \cc}_1} \leq \frac{1}{O(\kappahrg\gamma_{\mathrm{route}})} \min_{\DDelta: \BB^\top \DDelta = \veczero} \gg^\top \DDelta/\norm{\LL \DDelta}_1, \label{eq:treequality}
\end{equation}
for $\gamma_{\mathrm{route}} = e^{O(\log^{83/84}m)}$ and $\kappahrg = \log^{1/84}m$. The algorithm is deterministic, can be initialized in time $O(m\gamma_{\mathrm{tree}})$ for $\gamma_{\mathrm{tree}} = e^{O(\log^{83/84} m \log\log m)}$ and processes updates in amortized time $\gamma_{\mathrm{tree}}$.
\end{lemma}
\begin{proof}
    We will set $\gammahrg = e^{O(\log^{41/42}m\log\log m)}$, $\gammadiam = e^{O(\log^{83/84}m)}$, $\kappahrg = \log^{1/84}m$, and first invoke the data structure of \Cref{thm:dyn_hrg}.
    With this setup, we first show that the sets $E_{\mathrm{graph}}$ and $E_{\mathrm{pair}}$ can be maintained efficiently. Indeed, every update to $G$ leads to exactly one update of $E_{\mathrm{graph}}$, and every update to $H$ causes updates to $E_{\mathrm{pair}}$ if it affects an edge in $E_i^{\mathrm{out}}$. In that case, it causes $\gammahrg$ changes. Therefore, the amortized update time for maintaining the sets $E_{\mathrm{graph}}$ and $E_{\mathrm{pair}}$ is $\O(\gammahrg \cdot \kappahrg) =  e^{O(\log^{41/42} m \log\log m})$. 
    
    Then we consider the set of trees $\mathcal{T}$, which are defined by \Cref{def:trees}. Recall that $T$ is a subgraph of $H$, so we can directly decide which updates in $H$ to propagate to $T$. This also implies that $T$ is flat over $G$. \eqref{eq:treequality} follows from the fact that $\mathcal{T}$ captures all monotone cycles (\Cref{lma:HRG_decom}), and the quality of the best monotone cycle (\Cref{decomp:min_ratio}).
    
    The runtime is given by $O(|\mathcal{T}| \gammahrg \cdot \kappahrg) \le \gamma_{\mathrm{tree}}$, because $|\mathcal{T}| = O(\gammahrg^2 \log m)^{\kappahrg}$. The lemma follows.  
\end{proof}

\begin{proof}[Proof of \Cref{thm:CycleToTreeCycle}]
The theorem follows from \Cref{thm:dyn_hrg}, \Cref{lem:DynHRGTrees} and \Cref{cor:min_ratio_HRG}.
We use \Cref{thm:dyn_hrg} to maintain a HRG of $G$ and \Cref{lem:DynHRGTrees} to maintain a collection of $\gamma_{\mathrm{tree}}$ flat rooted trees $\calT.$
For each tree in the collection $T \in \calT$ and its set of off-tree edges $E_{off}$, we use the given min-ratio tree cycle data structure (\cref{def:DynamicMinRatioTreeCycle}) to maintain a flat forest $F$ and approximate tree cycle represented as paths on $F.$
Because each $T$ and $E_{\mathrm{off}}$ are flat in $G$, $F$ is also flat on $G$.
Therefore, our min-ratio cycle data structure maintains the flat forest $F^G$ as a disjoint union (on disjoint vertex sets) of flat forests given by each of the $\gamma_{\mathrm{tree}}$ tree cycle data structures.
\end{proof}

\section{Portal Routed Graphs and Min-Ratio Tree Cycles}
\label{sec:DynMinRatioTreeCycle}

In this section, we build dynamic min-ratio tree cycle data structures by reducing to a min-ratio cycle problem on a substantially smaller graph.

\TreeCycleToCycle* 

The data structure for \Cref{thm:TreeCycleToCycle} is presented in \Cref{algo:TreeCycleToCycle}.
On a high-level, we construct and maintain a vertex sparsifier for the approximate min-ratio tree cycle problem.
In particular, given any size reduction parameter $k$, we maintain a smaller graph $\calP$ on roughly $m/k$ vertices and reduce the problem of finding an approximate min-ratio cycle on $\calP.$
In the dynamic setting, we show that $\calP$ can be maintained dynamically with roughly $k^2$ update time, but much lower recourse.
Because $\calP$ could still contain $m$ edges, we maintain an edge sparsifier with respect to distances, i.e. a spanner, $\wh{G}$ of $\calP$ containing roughly $m/k$ edges.
Then, we recursively use $\calD^{\mathrm{MRC}}$ to maintain an approximate min-ratio cycle on $\wh{G}$, which is transformed into an approximate min-ratio tree cycle on $G$.

The construction of the vertex sparsifier first identifies a set of $m/k$ vertices, called \emph{portals}, and moves every off-tree edge onto the portals.
This results in a graph $G'$ containing a tree and off-tree edges between $m/k$ portals.
The min-ratio tree cycle is preserved in $G'$.
Then, one can replace every maximal path containing degree-2 vertices with an edge and repeatedly eliminate degree-1 vertices.
The process ends with $\calP$, a graph on the set of portals and some additional Steiner nodes, which preserves every cycle of $G'$, as well as the min-ratio tree cycle of $G.$

\subsection{Portal Routing and Portal Routed Graphs}

For a cleaner presentation, we force the set of portals to contain all the Steiner nodes, i.e., the set of portals is \emph{branch-free} on the tree.
\begin{definition}[Branch-free set]
\label{def:BranchFree}
Given a tree/forest $T$, a set $R \subseteq V(T)$ is \emph{Branch-free} if for any $u \not\in R$, the number of vertices $v \in R$ such that $T[u, v]$ containing no other vertex in $R$ is at most $2.$
\end{definition}
\begin{figure}
    \centering
    \includegraphics[width = 15cm]{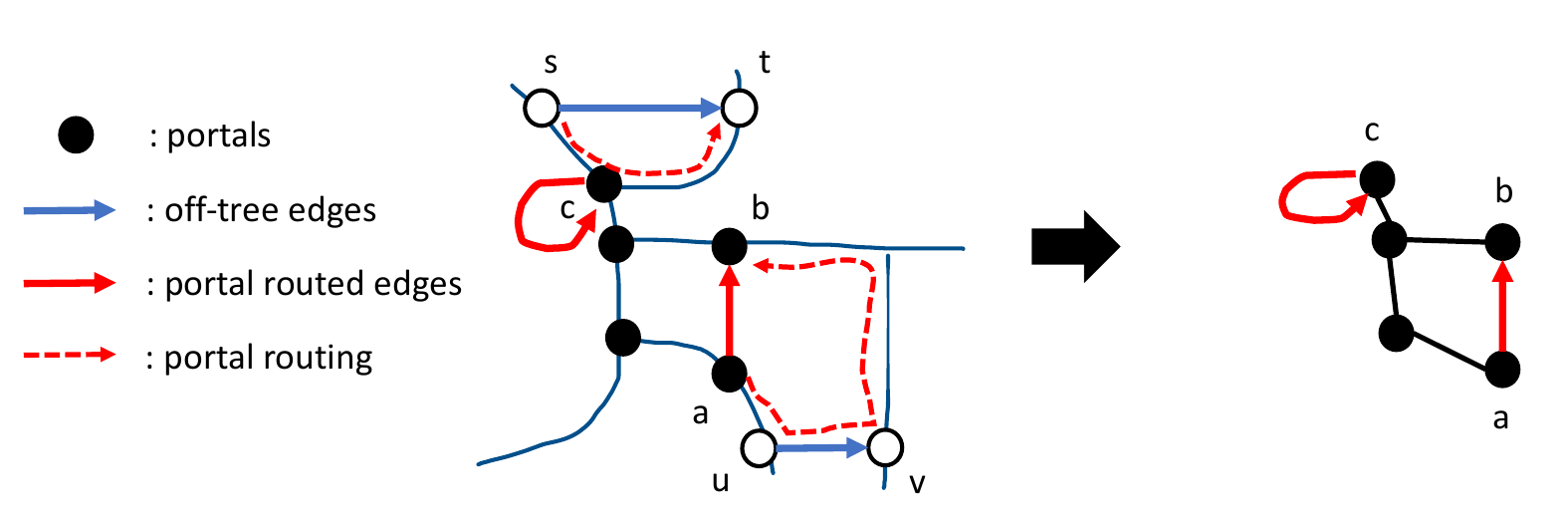}
    \caption{This example contains a tree, two off-tree edges $(u, v)$ and $(s, t)$, and a branch-free set of five portals. The right-hand side shows the portal routed graph. The edge $(u, v)$ is moved to $(a, b)$ and the edge $(s, t)$ is moved to a self-loop $(c, c).$}
    \label{fig:PortalRouting}
\end{figure}

To define how we move the off-tree edges, we first define the portal routing of each off-tree edge $e$ by short-cutting the tree cycle $\TCyc[e]$ at the given set of portals.
See \Cref{fig:PortalRouting} for an illustration of portal routing.
\begin{definition}[Portal routing]
\label{def:PortalRouting}
Given a tree/forest $T = (V, E_T)$, a branch-free set of \emph{portals} $P \subseteq V$, and a set of off-tree edges $E_{\mathrm{off}}$, we define the \emph{portal routing} $\calP(e)$ for each off-tree edge $e = (u, v)$ as follows:
\begin{align*}
\calP(e) = \begin{cases}
    \TCyc[e] &\text{if there are less than 2 portals on the path } $T[e]$ \\
    T[a, u] \oplus e \oplus T[v, b] &\text{otherwise}
\end{cases}
\end{align*}
where $a$ and $b$ are the first and the last portal on the tree path $T[u, v].$
We also define $e^{\calP} \defeq (a, b)$ if there are at least two portals on the path $T[e]$.
\end{definition}

Now, we are ready to define the \emph{portal routed graph} $\calP$, our notion of a vertex sparsifier for the min-ratio cycle problem.
We move each off-tree edge $e$ to the closest portals that short-cut its tree cycle $\TCyc[e].$
Then, we replace each tree path between two portals $T[p_1, p_2]$ with an edge $(p_1, p_2).$
This naturally defines an embedding from $\calP$ into $G.$
The edge lengths and gradients on $\calP$ are defined according to the embedding so that cycles in $G$ are preserved in $\calP$.
See \Cref{fig:PortalRouting} for an illustration of portal routed graphs.
\begin{definition}[Portal routed graph and embeddings]
\label{def:PRG}
Given a graph $G = (V, E)$ which contains a tree/forest $T = (V, E_T)$ and a set of off-tree edges $E_{\textrm{off}}$, a branch-free set of \emph{portals} $P \subseteq V$, and edge lengths $\ll$ and gradients $\bg$, we define the \emph{Portal routed graph} $\calP(G, T, P)$ as a graph with vertex set $P$ and an embedding $\Pi$ into $G.$
It contains the following two types of edges:
\begin{itemize}
\item \underline{Tree-path edges:} For each pair of portals $p_1, p_2 \in P$ such that $T[p_1, p_2]$ contains no other portals, we add an edge $e^{\calP} = (p_1, p_2)$ to $\calP(G, T, P)$ with length $\ll^{\calP}(e^{\calP}) \defeq \ll(T[p_1, p_2])$ and gradient $\bg^{\calP}(e^{\calP}) \defeq 0.$
We embed the edge into $G$ using the tree path, i.e., $\Pi(e^{\calP}) \defeq T[p_1, p_2].$
\item \underline{Portal routed edges:} For each off-tree edge $e \in E_{\textrm{off}}$ such that $e^{\calP}$ is well-defined, we add an edge $e^{\calP}$ to $\calP(G, T, P)$ with length $\ll^{\calP}(e^{\calP}) \defeq \ll(\calP(e))$ and gradient $\bg^{\calP}(e^{\calP}) \defeq \l\bg, \vecone_{\TCyc[e]}\r.$
We embed the edge into $G$ using the portal routing, i.e., $\Pi(e^{\calP}) \defeq \calP(e).$
\end{itemize}
When the graph $G$, tree/forest $T$, and the portal set $P$ are clear from the context, we denote $\calP(G, T, P)$ by $\calP$ for a clean presentation.
\end{definition}
\begin{remark}
\label{rem:gradient}
In our usage, we maintain $P$ as an incremental set and never change the gradient of any existing edge in $\calP$, even if the tree $T$ is changed.
All tree-path edges have gradient zero.
For a portal routed edge, its gradient $\bg^{\calP}(e^{\calP})$ is defined w.r.t. the initial tree $T^{init}$.
\end{remark}

The following lemma argues that all tree cycles containing at least one portal are preserved in $\calP.$
Later in our min-ratio tree cycle data structure, we handle the tree cycles touching no portals separately.
\begin{lemma}[$\calP$ preserves min-ratio tree cycles]
\label{lem:MinRatioTreeCyclePRG}
Under the setting of \Cref{def:PRG}, for any off-tree edge $e \in E_{\textrm{off}}$ such that $e^{\calP}$ is in $\calP$ and is not a self-loop and any sign $s \in \{\pm 1\}$, there is a cycle $\cc^{\calP}$ in $\calP$ s.t.
\begin{align*}
    \frac{\l\bg^{\calP}, \cc^{\calP}\r}{\norm{\LL^{\calP} \cc^{\calP}}_1}
    = \frac{\l\bg, \Pi_{\calP \mapsto G}(\cc^{\calP})\r}{\norm{\LL \Pi_{\calP \mapsto G}(\cc^{\calP})}_1}
    = s \cdot \frac{\l\bg, \vecone_{\TCyc[e]}\r}{\ll(\TCyc[e])}
\end{align*}
\end{lemma}
\begin{proof}
Fix the off-tree edge $e = (u, v) \in E_{\textrm{off}}$ and consider the tree cycle $\TCyc[e].$
Let $p_1, p_2, \ldots, p_k$ be the portals along the tree path $T[u, v].$
One can decompose $\TCyc[e]$ as follows:
\begin{align*}
    \TCyc[e] = T[p_1, p_2] \oplus \cdots \oplus T[p_{k - 1}, p_k] \oplus (T[p_k, v] \oplus \rev(e) \oplus T[u, p_1])
\end{align*}
Notice that each $T[p_i, p_{i+1}]$ corresponds to a tree-path edge $(p_i, p_{i+1})$ in $\calP$ and the part in the parenthesis corresponds to the portal-routed edge $e^{\calP} = (p_1, p_k)$ in reverse direction.
Therefore, we can define the cycle $\cc^{\calP}$ to be
\begin{align*}
\cc^{\calP} \defeq (p_1, p_2) \oplus \cdots \oplus (p_{k-1}, p_k) \oplus \rev(e^{\calP})
\end{align*}
We know that the cycle embeds into $G$ as $\TCyc[e]$, i.e., $\Pi_{\calP \mapsto G}(\cc^{\calP}) = \TCyc[e]$ and the lemma follows from the definition of edge gradients and lengths in the portal routed graph (\Cref{def:PRG}).
\end{proof}

\begin{figure}
    \centering
    \includegraphics[width = 15cm]{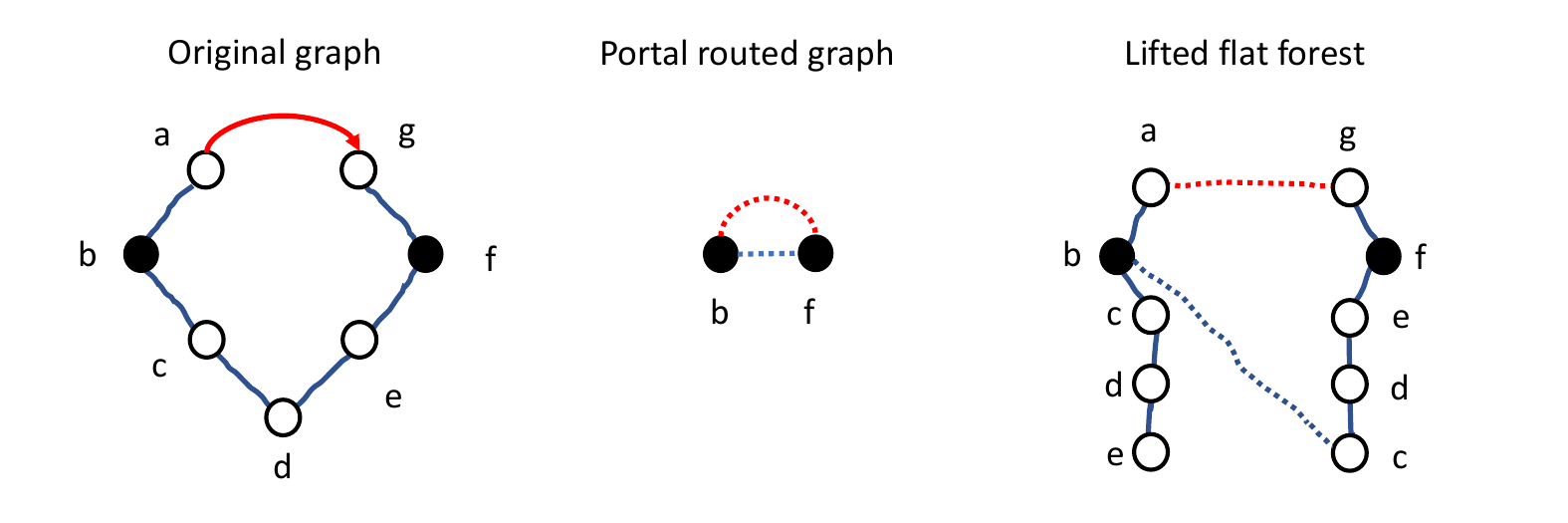}
    \caption{The original graph has two portals $b$ and $f$ and one off-tree edge $(a, g)$. If the forest in the portal routed graph contains a single portal routed edge, it is lifted back to $G$ with the edge $(a, g)$ added. If it contains a single tree-path edge, it is lifted to $G$ with the edge $(b, c)$ added.}
    \label{fig:TreeLiftPRG}
\end{figure}

In our data structure for proving \Cref{thm:TreeCycleToCycle}, we use the given min-ratio cycle data structure $\calD^{\mathrm{MRC}}$ on $\calP$ in a black-box manner.
$\calD^{\mathrm{MRC}}$ additionally maintains a forest $F$ that flatly embeds into $\calP$ and outputs cycles as several off-tree edges and paths on $F.$
In order to bring the cycle back to $G$, we first need to transform $F$ into a forest $F^G$ that flatly embeds into $G.$
Then, we can bring the cycle output by $\calD^{\mathrm{MRC}}$ back to $G$ as some off-tree edges and paths on $F^G.$
See \Cref{fig:TreeLiftPRG} for an illustration.
\begin{definition}[Lift a forest from $\calP$]
\label{def:TreeLiftPRG}
Under the setting of \Cref{def:PRG}, given a forest $F$ and $\Pi_{V(F) \mapsto P}$ that are a flat embedding of $\calP$, we define its lift in $G$, denoted by $F^G$, as follows:
\begin{itemize}
\item \underline{Dangling subtrees:} For every vertex $u \in V(F)$ that maps to a portal $p$, we add $u$ in $F^G$ as well as a copy of $T_p$ dangling down $u.$ $T_p$ is the maximal sub-tree in $T$ containing only one portal which is $p.$
We denote the copy by $T_u.$
\item \underline{Lifting tree-path edges:} For every edge $e^F = (u, v)$ in $F$ that maps to a tree-path edge $(p_1, p_2)$ in $\calP$, we add an edge $(u, a)$ to $F^G$ where $a$ is in $T_v$, the subtree dangling under $v$, that corresponds to the vertex incident to $p_1$ in the tree-path $T[p_1, p_2].$
\item \underline{Lifting portal routed edges:} For every edge $e^F = (u, v)$ in $F$ that maps to a portal routed edge $e^{\calP} = (p_1, p_2)$, we add to $F^G$ an edge $(a, b)$ where $a \in T_u$ and $b \in T_v$ correspond to the endpoints of the pre-image of $e^{\calP}$ respectively.
\end{itemize}
The mapping $\Pi_{V(F^G) \mapsto V}$ is defined accordingly.
\end{definition}

We show that bringing $F$ to $F^G$ increases the vertex congestion by a factor of $2.$
The constant factor blow-up is acceptable as the number of portals is a much smaller portion of $V(G)$ if we set the reduction parameter $k = \omega(1).$
\begin{lemma}[Vertex congestion of $F^G$]
\label{lem:VCongTreeLiftPRG}
Under the setting of \Cref{def:TreeLiftPRG}, $F^G$ is a forest and a flat embedding of $G$ w.r.t. $\Pi_{V(F^G) \mapsto V}.$
If every portal in $\calP$ had vertex congestion $\gamma$ from $F$ to $G$,the mapping from $F^G$ to $G$ has vertex congestion $2 \gamma$.
\end{lemma}
\begin{proof}
To see that $F^G$ is a forest, one can contract each dangling subtree of $F^G$ onto its portal vertex.
The result is exactly $F$, which is a forest.
Therefore, $F^G$ must be a forest.

Next, we argue about the vertex congestion of $F^G.$
For any non-portal vertex $v$ and portal vertex $p$, $v$ appears in $T_p$ if $T[v, p]$ contains no other portal vertices.
For $v$, this can only hold for at most two portals $p_1, p_2$ due to the branch-free-ness of the portal set.
Because each portal appears at most $\gamma$ times in $F$ and $F^G$, $u$ appears in $F^G$ at most $2 \gamma$ times, $\gamma$ for each portals.
\end{proof}

After bringing the forest $F$ maintained by $\calD^{\mathrm{MRC}}$ back to $G$ as $F^G$, we can bring back the cycle using $F^G.$
\begin{lemma}[Lift a cycle from $\calP$]
\label{lem:LiftTreeCycle}
Consider a graph $G = (V, E)$ with edge lengths $\ll$ and gradients $\bg$, a tree/forest $T$, a branch-free set of portals $P \subseteq V$, and a portal routed graph $\calP \defeq \calP(G, T, P).$
Given a forest $F$ and $\Pi_{V(F) \mapsto P}$ which are a lift of $\calP$, and a cycle $\cc$ on $\calP$ which can be represented as $\gamma$ tree paths on $F$ and $\gamma$ off-tree edges, the cycle $\Pi_{\calP \mapsto G}(\cc)$ can be represented as $\gamma$ tree paths on $F^G$ and $\gamma$ off-tree edges.

In addition, we have $\l\bg^{\calP}, \cc\r = \l\bg, \Pi_{\calP \mapsto G}(\cc)\r$ and $\|\LL \Pi_{\calP \mapsto G}(\cc)\| \le \|\LL^{\calP} \cc\|_1.$
\end{lemma}
\begin{proof}
We can write $\cc$ as
\begin{align*}
    \cc = F[p_1, p_2] \oplus (p_2, p_3) \oplus F[p_3, p_4] \oplus \ldots \oplus F[p_{2\gamma-1}, p_{2\gamma}] \oplus (p_{2\gamma}, p_1)
\end{align*}
By \Cref{def:TreeLiftPRG}, $\Pi_{\calP \mapsto G}$ maps each $F[p_i, p_{i+1}]$ part to $F^G[p_i, p_{i+1}].$
$\Pi_{\calP \mapsto G}$ maps each $(p_i, p_{i+1})$ to $T[p_i, u] \oplus e \oplus T[v, p_{i+1}]$ where $(p_i, p_{i+1})$ is the portal routed edge of $e = (u, v) \in G$.

$\l\bg^{\calP}, \cc\r = \l\bg, \Pi_{\calP \mapsto G}(\cc)\r$ comes from how we define $\bg^{\calP}$ and the fact that $\Pi_{\calP \mapsto G}(\cc)$ remains a circulation.
$\|\LL \Pi_{\calP \mapsto G}(\cc)\| \le \|\LL^{\calP} \cc\|_1$ comes from triangle inequality and how we define $\ll^{\calP}.$ 
The lemma follows.
\end{proof}

\subsection{Dynamic Portal Routed Graphs}

Using these definitions and properties, we are ready to state the lemma for the dynamic portal routed graph data structure.
\begin{lemma}[Dynamic portal routed graphs]
\label{lem:DynPRG}
Given a size reduction parameter $k$, a dynamic tree/forest $T = (V, E_T)$ with a set of off-tree edges $E_{\mathrm{off}}$, edge lengths $\ll$ and gradients $\bg$, there is a data structure that supports up to $O(m / k)$ number of updates of the following type:
\begin{itemize}
    \item $\textsc{InsertTreeEdge}(e)/\textsc{DeleteTreeEdge}(e)$: adds/removes edge $e$ to/from $T.$
    \item $\textsc{InsertOffTreeEdge}(e)/\textsc{DeleteOffTreeEdge}(e)$: adds/removes edge $e$ to/from $E_{\mathrm{off}}.$
    \item $\textsc{InsertVertex}(u)$: adds a new isolated vertex $u$ to $V.$
    \item $\textsc{AddPortal}(u)$: adds $u$ as a new portal vertex.
\end{itemize}
In the beginning, the data structure initializes a set of portals $P$ of size $O(m/k)$ and a portal routed graph $\calP(G = (V, E_T \cup E_{\mathrm{off}}), T, P)$.
After each update, the data structure adds at most $4$ new portals and updates the portal routed graph $\calP$ with $O(1)$ updates of the following type:
\begin{itemize}
    \item $\textsc{InsertEdge}(e)/\textsc{DeleteEdge}(e)$: adds/removes edge $e$ to/from $\calP.$
    \item $\textsc{InsertVertex}(v, E_v)$: adds to $\calP$ a set of edges $E_v$ incident to $v$ and adds $v$ if it is new to $\calP.$
    \item $\textsc{SplitAndMerge}(u, v, E_{\textrm{move}})$: $E_{\textrm{move}}$ is a subset of edges incident to either $u$ or $v.$ This update adds a new vertex $w$ to $\calP$ and moves $E_{\textrm{move}}$ to $w$, i.e., re-mapping endpoints of edges in $E_{\textrm{move}}$ from either $u$ or $v$ to $w.$
    When $u = v$, this update simply splits the vertex $u.$ 
\end{itemize}
Furthermore, suppose there is an algorithm $\calA$ that explicitly maintains a forest $F$ and $\Pi_{V(F) \mapsto P}$ that are flat in $\calP$ with maximum vertex congestion $\gamma_{\mathrm{vcong}}$, i.e., $\cA$ outputs changes to $F$ and $\Pi_{V(F) \mapsto P}$ after each update to $\calP$.
The data structure maintains $F^G$ and $\Pi_{V(F^G) \mapsto V}$ (\Cref{def:TreeLiftPRG}).
In particular, each edge update to $F$ can be handled in $m^{o(1)}$-time and causes one edge update to $F^G.$
Each vertex insertion/deletion to $F$ can be handled in $m^{o(1)}$-time as well.

The data structure takes $m^{1+o(1)}$-time to initialize and $k \cdot \gamma_{\mathrm{vcong}} \cdot m^{o(1)}$-time to handle each update.
\end{lemma}

The data structure implementing \Cref{lem:DynPRG} is presented as \Cref{algo:DynPRG}.
The key idea is that edge updates between portals can be directly passed to $\calP$. We therefore first add $u$ and $v$ to $P$ whenever edge $e$ gets updated, and it suffices to implement the $\textsc{AddPortal}$ operation efficiently.
At initialization, the data structure first computes the set of portals $P$ that decomposes the tree into roughly $m/k$ pieces, each having $k$ incident off-tree edges.

Adding a new portal only affects the portal routing of at most $k$ off-tree edges and performs a split-and-merge operation on two vertices in $\calP.$
We also adds $\gamma_{\mathrm{vcong}}$ copies of subtrees dangling under the new portals to $F^G$, the flat forest in $G.$
Therefore, explicitly writing down the updates to $\calP$ takes only $O(k)$-time.

In this data structure, we use the tree decomposition tool from \cite{ST03,ST04}, which decomposes the tree into roughly $m/k$ edge-disjoint pieces and each piece is adjacent to roughly $k$ off-tree edges.
This is a simple depth-first-search procedure over the tree and we formalized it as follows:
\begin{lemma}[Tree decomposition, {\cite[Theorem 10.3]{ST04}}]
  \label{lemma:treeDecomp}
  There is a deterministic linear-time algorithm that on a graph $G=(V, E)$, a rooted spanning tree $T$, and a reduction parameter $k$, outputs a decomposition $\cW$ of $T$ into edge-disjoint sub-trees such that:
  \begin{enumerate}
  \item $\Abs{\cW} = O(m/k).$ \label{item:sizeBound}
  \item  \label{item:branchFreeBoundary} $R \defeq \partial \cW \subseteq V$, defined as the subset of vertices appear in multiple components, is branch-free.
  \item For every component $C \subseteq V$ of $\cW$, the number of edges adjacent to non-boundary vertices of $C$ is at most $40k.$\label{item:totalWgtBound}
  \end{enumerate}
\end{lemma}

\begin{algorithm}[!ht]
\caption{Dynamically maintain a portal routed graph $\calP(G, T, P)$.}
\label{algo:DynPRG}
\SetKwProg{Globals}{global variables}{}{}
\SetKwProg{Proc}{procedure}{}{}
\Globals{}{
$T^{init}$: The forest given initially. \\
$P$: the set of portals. \\
$\calP$: the portal routed graph \\
$\calA$: the forest maintenance algorithm in $\calP.$ \\
$F$: the flat forest in $\calP$ maintained by a given algorithm $\calA.$ \\
$F^G$: the flat forest in $G$ \\
}
\Proc{$\textsc{Initialize}(T, E_{\textrm{off}}, \ll, \bg, k, \gamma_{\mathrm{vcong}})$}{
    $T^{init} \gets T$ \\
    $\cW \gets \textsc{TreeDecompose}(G, T, k)$ \\
    $P \gets \partial \cW$ \\
    $\calP \gets \calP(G, T, P)$ \\
    \For{$p \in P$}{
        Let $T_p$ be the dangling subtree of $p$ w.r.t. $T^{init}.$ \\
        Add $\gamma_{\mathrm{vcong}}$ copies of dangling sub-trees $\{T_{p, 1}, \ldots, T_{p, \gamma_{\mathrm{vcong}}}\}$ to $F^G.$
    }
    Initialize $F$ using $\calA$ and add the corresponding edges to $F^G.$
}
\Proc{$\textsc{AddPortal}(u)$}{
    \If{adding $u$ as a new portal makes $P$ not branch-free w.r.t. $T^{init}$}{
        Let $u^+$ be the vertex such that $P \cup \{u, u^+\}$ is branch-free. \\
        $\textsc{AddPortal}(u^+)$
    }
    Let $E_{u}$ be the set of off-tree edges whose portal routing passes through $u.$ \\
    \If{$u$ lies between two portals $p_1$ and $p_2$ in $T^{init}$}{
        Update all $\gamma_{\mathrm{vcong}}$ copies of dangling subtrees of $T_{p_1}$ and $T_{p_2}$ \\
        $\calP.\textsc{SplitAndMerge}(p_1, p_2, E_u \cap (N_{\calP}(p_1) \cup N_{\calP}(p_2)))$
    }
    \Else{
        Let $p$ be the portal closest to $u.$ \\
        Update all $\gamma_{\mathrm{vcong}}$ copies of dangling subtrees of $T_p$ \\
        $\calP.\textsc{SplitAndMerge}(p, p, E_u \cap N_{\calP}(p))$
    }
    Update lengths and gradients for edges in $E_u.$ \\
    $P \gets P \cup \{u\}$ \\
    $\calP.\textsc{InsertVertex}(u, E_u)$ \\
    Add $\gamma_{\mathrm{vcong}}$ copies of dangling sub-trees $\{T_{u, 1}, \ldots, T_{u, \gamma_{\mathrm{vcong}}}\}$ to $F^G.$
}
\Proc{$\textsc{Update}(U)$}{
\For{$u \in V(U)$}{
    $\textsc{AddPortal}(u)$ \\
}
Pass the update to $\calP.$
}
\Proc{$\textsc{ForestUpdateEdge}(e = (u, v))$}{
    Let $i$ and $j$ be the indices of the copies of $u$ and $v$ where the forest edge update happens in $F$. \\
    Update the corresponding edge between $T_{u, i}$ and $T_{v, j}$ in $F^G$
}
\end{algorithm}

In our implementation, we use the following claims regarding maintaining branch-free sets on a tree.
The first claim shows that after adding a vertex to a branch-free set, we can maintain its branch-freeness by adding at most one additional vertex, which can be found efficiently using dynamic tree data structures.

\begin{claim}
\label{claim:BranchFreeAddPortal}
Given a tree/forest $T$ and a branch-free set of vertices $P \subseteq V(T)$, for every vertex $u \in V(T)$, one can make $P \cup \{u\}$ branch-free by adding at most one additional vertex $u^+ \in V(T)$ to $P$.
\end{claim}
\begin{proof}
Imagine we root $T$ at $u$. Then, we choose
$u^+$ as the vertex farthest from $u$ that has two children, each having a descendent in $P$. 
There can only be at most one such vertex $u^+$, as otherwise, $P$ was not branch-free in the first place.
\end{proof}

The second claim shows that, if the portal set remains branch-free after adding a new portal, the portal routed graph undergoes a $\textsc{SplitAndMerge}$ update as described in \Cref{lem:DynPRG}.

\begin{claim}
\label{claim:howPRGchange}
Under the setting of \Cref{def:PRG}, given some vertex $u$ such that $P \cup \{u\}$ remains branch-free, $\calP(G, T, P \cup \{u\})$ can be obtained from $\calP(G, T, P)$ with one $\textsc{SplitAndMerge}$ operation followed by decreasing lengths of some edges incident to the $u$ in $\calP(G, T, P \cup \{u\})$ and at most two edge insertions and one deletion.
The total number of affected edges whose length we decrease is $O(k)$.
\end{claim}
\begin{proof}
According to the definition of portal routed graphs (\Cref{def:PRG}), there are two types of (possibly affected) edges in $\calP$.
\begin{itemize}
\item \underline{Affected tree-path edges:} If $u$ appears in between two portals $p_1, p_2$ in $T$ such that $T[p_1, p_2]$ contains no other portals, we remove $(p_1, p_2)$ and add two tree-path edges $(p_1, u)$ and $(u, p_2)$ to $\calP$. This causes one edge deletions and two edge insertions. 
\item \underline{Affected portal routed edges:} Let $C$ be the component of $\cW$, the decomposition of $T$ induced by $P$, that contains $u$.
One can observe every off-tree edge $e$ whose portal routing $\calP(e)$ touches $u$ is incident to some vertex in $C.$
The number of such edges is $O(k)$ by the description of the initialization of the portal set $P$ and \Cref{lemma:treeDecomp}.
Because $P$ is branch-free, $C$ has at most two boundary portals.
Every portal routed edge incident to $u$ in $\calP(G, T, P \cup \{u\})$ was either incident to one of $\partial C$ or not well-defined.
Therefore, one can obtain the new portal routed graph with a $\textsc{SplitAndMerge}$ operation from $\partial C$ followed by inserting new edges/decreasing edge lengths around $u.$
\end{itemize}
\end{proof}

\begin{proof}[Proof of \Cref{lem:DynPRG}]
The correctness follows from \Cref{claim:BranchFreeAddPortal} and \Cref{claim:howPRGchange}.

We next analyze the runtime. Each update is handled in $k \gamma_{\mathrm{vcong}} m^{o(1)}$-time because, by \Cref{claim:howPRGchange}, the number of affected tree-path and portal routed edges are $O(1)$ and $O(k)$ respectively.
In addition, we add $\gamma_{\mathrm{vcong}}$ dangling subtrees for each new portal.
Each dangling subtree has size $O(k).$
The bound on the update time follows.

The maintenance of $F^G$ follows from \Cref{def:TreeLiftPRG}.
In particular, whenever inserting/deleting an edge into/from $F$, we insert/delete the corresponding edge that lives in $F^G.$
The vertex congestion bound follows from \Cref{lem:VCongTreeLiftPRG}.
\end{proof}

\subsection{Sparsified Portal Routed Graphs via Dynamic Low-Recourse Spanners}

The other ingredient for the data structure of \Cref{thm:TreeCycleToCycle} is a dynamic spanner with low-recourse.
We use the dynamic spanner to maintain a sparsifier $\wh{G}$ of the portal routed graph $\calG$ so that we run the given min-ratio cycle data structure $\calD^{\mathrm{MRC}}$ on an instance of size roughly $1/k$ smaller.
In addition, we want the sparsifier to handle updates to $\calP$ such as $\textsc{SplitAndMerge}$ and $\textsc{InsertVertex}$ while keeping the recourse small.
Naively treating these two operations as edge updates results in roughly $k$ changes to $\wh{G}$ and an update cost $\Omega(k T_{upd}(m/k))$ where $T_{upd}$ is the amortized update time of $\calD^{\mathrm{MRC}}.$
Here we present a new dynamic spanner that handles both operations with low recourse and maintains an explicit embedding back to $\calP.$
See \Cref{sec:newSpanner} for more discussion and technical details.
\begin{restatable}{theorem}{NewSpanner}
\label{thm:newSpanner}
Given an $m$-edge $n$-vertex input graph $G = (V,E,l)$ with lengths in $[1,L]$, a degree threshold $\Delta$ such that $G$ initially has maximum degree at most $\Delta$, there is a data structure $\textsc{DynamicSpanner}$ that supports a polynomially-bounded number of updates of the following type:
\begin{itemize}
    \item $\textsc{InsertEdge}(e)$: adds edge $e$ to $G$, with the guarantee that $\deg^{\master}\text(v) \le \Delta$ in the graph $G$ where $\deg_{\master}$ is defined as the degree in $G$ ignoring the deletions/splits.
    \item $\textsc{DeleteEdge}(e)$: removes edge $e$ from $G$
    \item $\textsc{SplitVertex}(v, E_{move}, E_{crossing})$: we assume that $E_{move}$ is a set of edges incident to vertex $v$, and $E_{crossing}$ is a set of self-loops incident to $v$ such that $E_{move} \cap E_{crossing} = \emptyset$. 
    
    The operation splits the vertex $v \in V$ into vertex $v$ and a new vertex $v'$. It then moves all edges in $E_{move}$ to $v'$ by re-mapping all their endpoints from $v$ to $v'$. Finally, it re-maps all edges in $E_{crossing}$ such that thereafter each such self-loop at $v$ is mapped to an edge of the same length between $v$ and $v'$. Returns a pointer to the new vertex $v'$.
    \item $\textsc{InsertVertex}(v, E_{inc})$: Adds a vertex $v$ to the graph $G$ along with $E_{inc}$ a set of edges that is incident on $v$. 
\end{itemize}
For $\gamma_{\mathrm{spanner}} = e^{O(\log^{20/21} m \log\log m)}$, the algorithm maintains a subgraph $H \subseteq G$ and a graph embedding $\Pi_{G \mapsto H}$ such that 
\begin{enumerate}
    \item at any time, for every $u,v\in V$, $\dist_G(u,v) \leq \dist_H(u,v) \leq \gamma_{\mathrm{spanner}} \cdot \dist_G(u,v)$, and
    \item at any time, $H$ consists of at most $\Otil(n' \log L)$ edges where $n'$ is the number of vertices in the final graph $G$, and
    \item the total recourse of $H$, i.e., the total number of insertions/deletions and isolated vertex insertions to $H$, over a sequence of $\hat{q}$ invocations  of operations $\textsc{DeleteEdge}, \textsc{SplitVertex}$, $\textsc{InsertVertex}$ and an arbitrary number of (legal) invocations of $\textsc{InsertEdge}$ is at most $\gamma_{\mathrm{spanner}}(n'+\hat{q})\log L$,
    \item every edge $e = (u,v) \in E(G)$ is mapped to a $uv$-path $\Pi_{G \mapsto H}(e)$ consisting of at most $\gamma_{\mathrm{spanner}}$ many edges and ensures that at any time $\econg(\Pi_{G \mapsto H}) \leq \gamma_{\mathrm{spanner}}\cdot \Delta \log L$. The number of embedding paths changed over a sequence of $q$ updates to $G$ is at most $\gamma_{\mathrm{spanner}}\cdot \Delta \cdot q \log L$.
\end{enumerate}
The algorithm maintains $H$ and $\Pi_{G \mapsto H}$ explicitly, is initialized in time $O(m)$ and thereafter processes each update in amortized time $\Delta \cdot \gamma_{\mathrm{spanner}} \log L$.
\end{restatable}

When using \Cref{thm:newSpanner} to maintain a sparsifier $\wh{G}$ of the portal routed graph $\calP$, the min-ratio cycle in $\calP$ could disappear from $\wh{G}.$
In this case, we show that one of the \emph{spanner cycles}, the cycle consists of an edge $e \in \calP$ and its embedding path $\Pi_{\calP \mapsto \wh{G}}(e)$, has a comparable ratio.
\begin{lemma}[Min-ratio cycle given spanner]
\label{lem:MinRatioCycleSpanner}
Given a graph $G = (V, E)$ with edge lengths $\ll$ and gradients $\bg$, consider the spanner $H \subseteq G$ and $\Pi_{G \mapsto H}$ maintained by \Cref{thm:newSpanner}.
Let $\cc^H$ be an $\alpha$-approximate min-ratio cycle in $H$ and $\cc^{\Pi}$ be the min-ratio cycle of the form $e \oplus \rev(\Pi_{G \mapsto H}(e)), e \in G$ (also denoted as $H[e]$, the \emph{spanner cycle of $e$}).
We have
\begin{align*}
    \min\left\{\frac{\l\bg, \cc^H\r}{\norm{\LL \cc^H}_1}, \frac{\l\bg, \cc^{\Pi}\r}{\norm{\LL \cc^{\Pi}}_1}\right\} \le \frac{1}{3\alpha \cdot \gamma_{spanner}} \min_{\BB^\top \bDelta = 0} \frac{\l\bg, \bDelta\r}{\norm{\LL \bDelta}_1}
\end{align*}
That is, either $\cc^H$ or $\cc^{\Pi}$ is a $(3\alpha \gamma_{\mathrm{spanner}})$-approximate min-ratio cycle in $G.$
\end{lemma}
\begin{proof}
Let $\bDelta^* \in \R^{E}$ be the min-ratio cycle on $G.$
We first construct $\bDelta^*$ as a circulation on $H$ by canceling out non-spanner edges.
That is, consider
\begin{align*}
    \bDelta^*_H \defeq \bDelta^* - \sum_{e \in E \setminus H} \bDelta^*_e \vecone_{H[e]}
\end{align*}
Observe that $\bDelta^*_H$ is supported on $H$ because each non-spanner edge $e \in G \setminus H$ appears only in one spanner cycle, which is exactly $H[e].$
Furthermore, we can bound the length of $\bDelta^*_H$ using triangle inequality as follows:
\begin{align*}
\norm{\LL \bDelta^*_H}_1
&\le \norm{\LL \bDelta^*}_1 + \sum_{e \in E \setminus H} |\bDelta^*_e| \norm{\LL \vecone_{H[e]}}_1 \\
&\le \norm{\LL \bDelta^*}_1 + \sum_{e \in E \setminus H} |\bDelta^*_e| \gamma_{\mathrm{spanner}} \cdot \ll_e \\
&\le 2\gamma_{\mathrm{spanner}} \cdot \norm{\LL \bDelta^*}_1
\end{align*}
where we use the fact that the length of $H[e]$ is at most $\gamma_{\mathrm{spanner}} \ll_e.$

Rearrangement yields that
\begin{align*}
\bDelta^* &= \bDelta^*_H + \sum_{e \in E \setminus H} \bDelta^*_e \vecone_{H[e]} \mathrm{ , and} \\
\l\bg, \bDelta^*\r &= \l\bg, \bDelta^*_H\r + \sum_{e \in E} \bDelta^*_e \l\bg, \vecone_{H[e]}\r
\end{align*}

Using the bound on the length of $\bDelta^*_H$, we have
\begin{align*}
\norm{\LL \bDelta^*_H}_1 + \sum_{e \in E \setminus H} |\bDelta^*_e| \norm{\LL \vecone_{H[e]}}_1 \le 3\gamma_{\mathrm{spanner}} \cdot \norm{\LL \bDelta^*}_1
\end{align*}

An averaging argument yields that
\begin{align*}
\min\left\{\frac{\l\bg, \bDelta^*_H\r}{\norm{\LL \bDelta^*_H}_1}, \min_{e \in G \setminus H}\left\{\frac{\l\bg, \vecone_{H[e]}\r}{\norm{\LL \vecone_{H[e]}}_1}\right\}\right\}
&\le \frac{\l\bg, \bDelta^*_H\r + \sum_{e \in E} \bDelta^*_e \l\bg, \vecone_{H[e]}\r}{\norm{\LL \bDelta^*_H}_1 + \sum_{e \in E \setminus H} |\bDelta^*_e| \norm{\LL \vecone_{H[e]}}_1} \\
&\le \frac{1}{3\gamma_{\mathrm{spanner}}} \cdot \frac{\l\bg, \bDelta^*\r}{\norm{\LL \bDelta^*}_1}
\end{align*}
That is, either the min-ratio cycle in $H$ or the best spanner cycle is a $(3\gamma_{\mathrm{spanner}})$-approximate min-ratio cycle in $G.$
The lemma follows when considering any $\alpha$-approximate solution in $H.$
\end{proof}

Now, we are ready to prove \Cref{thm:TreeCycleToCycle} using the dynamic portal routed graph (\Cref{lem:DynPRG}) and the dynamic spanner (\Cref{thm:newSpanner}).
We first initialize and maintain a portal routed graph $\calP$ using \Cref{lem:DynPRG}.
Since the data structure can only deal with $O(m/k)$ updates, we rebuild the whole data structure after every $m/k$ updates.
Then, we use \Cref{thm:newSpanner} to maintain $\wh{G}$, a spanner of $\calP.$
The given min-ratio cycle data structure $\calD^{\mathrm{MRC}}$ is then used to maintain $\alpha$-approximate min-ratio cycles on $\wh{G}.$
The cycle maintained by $\calD^{\mathrm{MRC}}$ is mapped back to $G$ as discussed in \Cref{lem:LiftTreeCycle}.

When maintaining $\wh{G}$ using the dynamic spanner, notice that $\calP$ undergoes the $\textsc{SplitAndMerge}$ operation, which is not supported in  \Cref{thm:newSpanner}.
However, we can mimic the operation with two vertex splits and map the two new vertices into the same one.

Unfortunately, the portal routed graph might have max degree as large as $\Omega(m)$ and it might lead to efficiency issues when using \Cref{thm:newSpanner}, whose update time depends on the max degree.
To deal with the issue, we instead maintain a spanner $H$ on an auxiliary graph $\calP^{\mathrm{aux}}$ which is flat in $\calP$ and is obtained from $\calP$ by split large degree vertices in the first place.
We split those vertices in a way aligning with the decomposition of $T$ using the portal set $P$ so that adding a new portal only affects $O(1)$ vertices in  $\calP^{\mathrm{aux}}.$

\begin{definition}[Auxiliary portal routed graph $\calP^{\mathrm{aux}}$]
\label{def:Paux}
Given a graph $G = (V, E)$ consisting of a tree $T$ and a set of off-tree edges, consider a branch-free set of $O(m/k)$ portals $P$ such that each component of $T \setminus P$ is incident to at most $O(k)$ off-tree edges.
Let $\cW$ be the edge-disjoint decomposition of $T$ induced by $P.$
We construct $\calP^{\mathrm{aux}}$ from $\calP = \calP(G, T, P)$ as follows:
For each portal $p \in P$, we split it into $p_{\mathrm{root}}$ and $p_C$ for each component $C \in \cW$ that contains $p.$
$p_{\mathrm{root}}$ is incident to the portal-routed edge whose pre-image is already incident to $p.$
$p_C$ is incident to the portal routed edges whose pre-image is incident to $C.$

We do not include tree-path edges in $\calP^{\mathrm{aux}}.$
\end{definition}

After adding a new portal, $\calP^{\mathrm{aux}}$ is still changed by a constant number of $\textsc{SplitVertex}$ operations and a few edge updates.
\begin{claim}
\label{claim:paux}
Given some vertex such that $P \cup \{u\}$ remains branch-free, $\calP(G, T, P \cup \{u\})^{\mathrm{aux}}$ is obtained from $\calP(G, T, P)^{\mathrm{aux}}$ with 6 $\textsc{SplitVertex}$ operation followed by decreasing edge lengths incident to the newly created vertices and a constant number of edge deletion/insertions.
\end{claim}
\begin{proof}
Let $\cW$ be the edge-disjoint decomposition of $T$ induced by $P.$
Let $C \in \cW$ be the component containing $u.$
After adding $u$ as a new portal, $C$ is further decomposed into two components $C_1, C_2$ where $u \in \partial C_1$ and $u \in \partial C_2.$
In the new auxiliary portal routed graph $\calP'$, we create three new vertices $u_{\mathrm{root}}$, $u_{C_1}$, and $u_{C_2}.$

For any pre-existing portal $p \in \partial C$, we move some of its incident edge in $\calP$ to $u.$
In the $\calP^{\mathrm{aux}}$, they goes to either $u_{\mathrm{root}}$, $u_{C_1}$, or $u_{C_2}.$
As there are at most two such portals $p$, this creates $6$ $\textsc{SplitVertex}$ operations and we identify the 6 new vertices as $u_{\mathrm{root}}$, $u_{C_1}$, or $u_{C_2}.$
We also add some new edges whose portal routing contains only $u$ but not any pre-existing portals.
Then, as adding portals shortcuts portal routing and decreases edge lengths, we update the graph by decreasing some edge lengths incident to these three new vertices $u_{\mathrm{root}}$, $u_{C_1}$, and $u_{C_2}.$
\end{proof}

Notice that $\calP^{\mathrm{aux}}$ can be maintained easily as we have \Cref{lem:DynPRG} that maintains $\calP.$
We use \Cref{thm:newSpanner} to maintain a spanner on $\calP^{\mathrm{aux}}$ which can be made a spanner on $\calP$.

\begin{algorithm}[!ht]
\caption{Dynamic Min-Ratio Tree Cycle via Dynamic Min-Ratio Cycle}
\label{algo:TreeCycleToCycle}
\SetKwProg{Globals}{global variables}{}{}
\SetKwProg{Proc}{procedure}{}{}
\Globals{}{
$\calD^{\calP}$: dynamic portal routed graph data structure from \Cref{lem:DynPRG} \\
$\calP$: the portal routed graph maintained by $\calD^{\calP}$ \\
$\calD^{\mathrm{MRC}}$: the given dynamic min-ratio cycle data structure \\
$F$: the lift forest maintained by $\calD^{\mathrm{MRC}}$ \\
$F^G$: the lift of $F$ maintained by $\calD^{\calP}$ \\
$\calD^{\mathrm{spanner}}$: the dynamic spanner data structure from \Cref{thm:newSpanner} \\
$\wh{G}$: the dynamic spanner of $\calP$ maintained using $\calD^{\mathrm{spanner}}$ \\
}
\Proc{$\textsc{Initialize}(T, E_{\textrm{off}}, \ll, \bg, k)$}{
    $\calP \gets \calD^{\calP}.\textsc{Initialize}(T, E_{\textrm{off}}, \ll, \bg, k)$ \\
    Build the spanner of the portal routed edges of $\calP^{\mathrm{aux}}$ as $H^{\mathrm{aux}}$ using $\calD^{\mathrm{spanner}}$\\
    Let $\wh{G}$ be the union of the tree-path edges of $\calP$ and $H$ obtained from $H^{\mathrm{aux}}$ by identifying vertices for the same portal as one. \\
    Initialize $\calD^{\mathrm{MRC}}$ on $\wh{G}$ \\
    Use $\calD^{\calP}$ to maintain $F^G$, the lift of $F$ in $G$ \\
}
\Proc{$\textsc{Update}(U)$}{
\If{there have been $m/k$ updates since the last initialization}{
$\textsc{Initialize}(T, E_{\textrm{off}}, \ll, \bg, k)$
}
Pass the update $U$ to $\calD^{\calP}$, which maintains $\calP$. \\
Pass the updates of $\calP$ to $\calD^{\mathrm{spanner}}$, which maintains $H^{\mathrm{aux}}.$ \\
Pass the updates of $\wh{G}$ to $\calD^{\mathrm{MRC}}.$ \\
$\calD^{\mathrm{MRC}}$ updates $F$ and outputs a cycle $\bDelta.$ \\
$\calD^{\calP}$ maintains $F^G$ according to the updates to $F.$ \\
Let $\cc_{\mathrm{spanner}}$ be the best spanner cycle. \\
$\cc^G_{\mathrm{spanner}} \gets \Pi_{\calP \mapsto G}(\cc_{\mathrm{spanner}})$ \\
$\bDelta^G \gets \Pi_{\calP \mapsto G}(\bDelta)$ \\
Let $\cc^G_{tree}$ be the best tree cycle $\TCyc[e]$ that contains no portals. \\
\Return{Best among $\cc^G_{\mathrm{spanner}}$, $\cc^G_{tree}$ and $\bDelta^G.$}
}
\end{algorithm}

Now, we are ready to argue the correctness and performance of \Cref{algo:TreeCycleToCycle} and prove \Cref{thm:TreeCycleToCycle}.

\begin{proof}[Proof of \Cref{thm:TreeCycleToCycle}]
First, we argue the correctness of the algorithm.
\Cref{lem:DynPRG} correctly maintains the portal routed graph $\calP$ under updates.
\Cref{claim:paux} and \Cref{thm:newSpanner} correctly maintains $\wh{G}$ as a spanner of $\calP.$
\Cref{lem:MinRatioCycleSpanner} states that either $\cc_{\mathrm{spanner}}$ or $\bDelta$ is a $(3\alpha\gamma_{\mathrm{spanner}})$-approximate min-ratio cycle on $\calP.$
Since we can correctly lift them back to $G$ using \Cref{lem:DynPRG} and \Cref{lem:LiftTreeCycle}, either $\cc_{\mathrm{spanner}}^G$, $\bDelta^G$, or $\cc_{tree}^G$ has ratio no worse than the min-ratio tree cycle in $G$ up to a $(3\alpha\gamma_{\mathrm{spanner}})$ factor due to \Cref{lem:MinRatioTreeCyclePRG}.
This concludes the correctness of the algorithm.

The lift forest $F^G$ has vertex congestion $\le 2\gamma_{\mathrm{vcong}}$ due to \Cref{lem:VCongTreeLiftPRG}.
The representation of the output cycle is guaranteed by \Cref{lem:LiftTreeCycle} and the fact that $\cc_{tree}^G$ is also a tree cycle on $F^G.$

Next, we analyze the runtime.
The initialization takes time $m \gamma_{\mathrm{spanner}} + T_{init}(m\gamma_{\mathrm{spanner}}/k)$ due to the initialization of the dynamic portal routed graph data structure $\calD^{\calP}$ (\Cref{lem:DynPRG}), the dynamic spanner $\calD^{\mathrm{spanner}}$ (\Cref{thm:newSpanner}), and the given dynamic min-ratio cycle data structure on the graph $\wh{G}$ which has size $m^{1+o(1)}/k.$

Finally, we analyze the amortized update time.
Throughout the course of $Q$ updates, the total time spent on initialization is
\begin{align*}
\frac{Q}{m/k}\left(m\gamma_{\mathrm{spanner}} + T_{init}\left(\frac{m\gamma_{\mathrm{spanner}}}{k}\right)\right)
\end{align*}
Total update time due to $\calD^{\calP}$ and $\calD^{\mathrm{spanner}}$ are, by \Cref{lem:DynPRG} and \Cref{thm:newSpanner}, $O(Q k^2 \gamma_{\mathrm{vcong}})$ and $Q k \gamma_{\mathrm{spanner}}$ respectively,
as each update changes $\calP$ and $\calP^{\mathrm{aux}}$ with $O(1)$ operations (\Cref{claim:howPRGchange} and \Cref{claim:paux}).
Across $Q$ updates, the total number of updates to the sparsified portal routed graph $\wh{G}$ is
\begin{align*}
\left(\frac{Q}{m/k} + 1\right) \frac{m\gamma_{\mathrm{spanner}}}{k} = \gamma_{\mathrm{spanner}} \left(Q + \frac{m}{k}\right)
\end{align*}
by \Cref{thm:newSpanner} and the fact that $\calD^{\mathrm{spanner}}$ is initialized for $Q / (m/k) + 1$ times.
Each update to $\wh{G}$ is handled by the given min-ratio data structure in $T_{upd}(m^{1+o(1)}/k)$-time.
Therefore, for $Q$ updates, the total update time is
\begin{align*}
\underbrace{\frac{Q}{m/k}\left(m\gamma_{\mathrm{spanner}} + T_{init}\left(\frac{m\gamma_{\mathrm{spanner}}}{k}\right)\right)}_{\text{total rebuild time}} &+ \underbrace{Q k^2 \gamma_{\mathrm{vcong}} \gamma_{\mathrm{spanner}}}_{\text{total update time due to $\calD^{\calP}$ and $\calD^{\mathrm{spanner}}$}} \\
&+ \underbrace{\gamma_{\mathrm{spanner}} \left(Q + \frac{m}{k}\right)T_{upd}\left(\frac{m\gamma_{\mathrm{spanner}}}{k}\right)}_{\text{total update time due to  $\calD^{\mathrm{MRC}}$}}
\end{align*}
Rearrangement yields
\begin{align*}
\frac{m\gamma_{\mathrm{spanner}}}{k}T_{upd}\left(\frac{m\gamma_{\mathrm{spanner}}}{k}\right) + Q\gamma_{\mathrm{spanner}}\left(k^2 \gamma_{\mathrm{vcong}} + \frac{k}{m}T_{init}\left(\frac{m\gamma_{\mathrm{spanner}}}{k}\right) + T_{upd}\left(\frac{m\gamma_{\mathrm{spanner}}}{k}\right)\right)
\end{align*}
We can charge the term independent of $Q$ to the initialization and conclude the amortized update time bound.
\end{proof}

\section{Fully-Dynamic Sparse Neighborhood Cover}
\label{sec:SNC}

In this section, we give an algorithm to maintain a sparse neighborhood cover (SNC) on a fully-dynamic graph. Our algorithm heavily builds on the framework obtained in \cite{KMP23}, but nonetheless, is nontrivial in composing various components from the paper.

\dynSNC*

\subsection{Additional Results from \texorpdfstring{\cite{KMP23}}{kmp23}}
\label{subsec:kmp23}

We will build our fully dynamic SNC by combining various pieces of \cite{KMP23}, who already designed an algorithm for maintaining a SNC under edge deletions only.
\begin{theorem}[see {\cite[Theorem 5.3]{KMP23}}]
\label{thm:dec_snc}
Given an $m$-edge constant-degree input graph $G = (V,E,l)$ with polynomially-bounded lengths in $[1,L]$ and a diameter parameter $D \geq 1$, there is a data structure $\textsc{DecrSNC}$ that supports the following update:
\begin{itemize}
    \item $\textsc{DeleteEdge}(e)$: removes edge $e$ from $G$. 
\end{itemize}
Under this update, the data structure maintains a set of partitions $\mathcal{P}_0, \mathcal{P}_1, \ldots, \mathcal{P}_{k}$ for $k = O(\log m)$ such that for some $\gammadecrsnc =  e^{O(\log^{20/21} m \log\log m)}$:  
\begin{enumerate}
    \item  \label{prop:snc_decrCover} every vertex $v \in V$, there is some index $0 \leq i \leq k$, such that $B(v, D/\gammadecrsnc) \subseteq C$ for some cluster $C \in \mathcal{P}_i$, and
    \item  \label{prop:snc_decrPartition} for every $0 \leq i \leq k$, and cluster $C \in \mathcal{P}_i$, we have $\diam(G[C]) \leq D$, and
    \item  \label{prop:almost_refining} for every $0 \leq i \leq k$, the partition $\mathcal{P}_i$ is maintained such that for every such partition $\mathcal{P}_i$, we have that the sizes of all sets $C$ that appear in $\mathcal{P}_i$ and are not a subset of a partition set in the previous version of $\mathcal{P}_i$ is at most $m \cdot \gammadecrsnc$.
\end{enumerate}
The data structure is deterministic, reports each change to the partitions explicitly, and takes total time $m \cdot \gammadecrsnc$ over any sequence of updates.
\end{theorem}
Next, we need a result that allows us to maintain a forest $F$ which flatly embeds into each cluster $C$ maintained by the decremental SNC of \Cref{thm:dec_snc}.
\begin{theorem}[see {\cite[Theorem 4.9]{KMP23}}]
\label{thm:mainTheoremLowDiamTree}
Given an $m$-edge input graph $G = (V,E,l)$ with polynomial lengths in $ [1,L]$ and maximum degree $3$. There is a data structure $\textsc{LowDiamTree}$ that maintains a forest $F$ that flatly embeds into $G$, and supports a polynomially-bounded number of updates of the following type:
\begin{itemize}
    \item $\textsc{InsertEdge}(e)/ \textsc{DeleteEdge}(e)$: adds/removed edge $e$ into/from $G$. If the edge is inserted, its associated length $l(e)$ has to be in $[1,L]$ and the maximum degree is not allowed to exceed $3$; if it is deleted, it has to be ensured that thereafter graph $G$ is still connected.
\end{itemize}
Under these updates, the algorithm explicitly maintains $F$ and $\Pi_{V(F)\mapsto V}$ that are a flat embedding into $G$, where $l_F$ of $F$ is defined by $l_F(e) = l(\Pi_{V(F) \mapsto V}(e))$ for every edge $e \in E(F)$, and a vertex map $\Pi_{V\mapsto V(F)}$ such that for $\gamma_{lowDiamTree} = e^{O(\log^{20/21} m \log\log m})$, at any time:
\begin{enumerate}
    \item $\diam_F(\Pi_{V \mapsto V(F)}(V)) \leq \gamma_{lowDiamTree}\cdot \diam(G)$, and
    \item  \label{prop:lowCongLowDiam} we have $\vcong(\Pi_{V(F) \mapsto V}) \leq \gamma_{lowDiamTree}$, and
    \item \label{prop:fewEdgesLowDiam} $F$ consists of at most $\gamma_{lowDiamTree} \cdot m$ vertices and edges.
\end{enumerate}
The algorithm maintains the flat hierarchical forest $F$ and all maps explicitly. Vertex maps are such that once an element is added to the preimage, its image remains fixed until the element is again removed.

The algorithm is deterministic, can be initialized in time $m \cdot \gamma_{lowDiamTree}$, and thereafter processes each edge insertion/deletion in amortized time $\gamma_{lowDiamTree}$.
\end{theorem}
Finally, we require an algorithm for maintaining a length vertex sparsifier $H$ onto a set of terminals $A$. Additionally, given any procedure that maintains a forest that flatly embeds into $H$, the algorithm maintains a forest that flatly embeds into $G$ which preserves distances between vertices in $A$.
\begin{theorem}[see {\cite[Theorem 4.12]{KMP23}}]
\label{thm:mainTheoremVSExt}
Given an $m$-edge input graph $G = (V,E,l)$ with polynomial lengths in $ [1,L]$ and maximum degree at most $3$. Then, for some $\gamma_{vertexSparsifier} = e^{O(\log^{20/21} m \log\log m})$, there is a data structure $\textsc{MaintainVertexSparsifier}$ that initially outputs an empty set $A$, and graph $H$ consisting of at most $\gamma_{vertexSparsifier}$ vertices and edges, and supports a polynomial number of updates of the following type:
\begin{itemize}
    \item $\textsc{InsertEdge}(e)/ \textsc{DeleteEdge}(e)$: adds/removed edge $e$ into/from $G$. If the edge is inserted, its associated length $l(e)$ has to be in $[1,L]$ and the maximum degree of $G$ is not allowed to exceed $3$.
    \item $\textsc{AddTerminalVertex}(a)/ \textsc{RemoveTerminalVertex}(a)$: adds/ removes the vertex $a \in V(G)$ to/from the terminal set $A$.
\end{itemize}
The algorithm processes the $t$-th update and outputs a batch of updates $U_H^{(t)}$ consisting of edge insertions/deletions, and isolated vertex insertions/deletions and that when applied to the current vertex sparsifier $H$, yields the next one, such that at all times
\begin{itemize}
    \item we have $A \subseteq V(H) \subseteq V(G)$, and
    \item for all vertices $u,v \in V(H)$, we have $\dist_G(u,v) \leq \dist_H(u,v)$ and further if $u,v \in A$ then we also have $\dist_H(u,v) \leq \gamma_{vertexSparsifier} \cdot \dist_G(u,v)$, and
    \item the number of edges and vertices in $H$ is at most $(1+|A|) \cdot \gamma_{vertexSparsifier}$, and
    \item we have $\sum_{t' \leq t} |U^{(t')}_H| \leq \gamma_{vertexSparsifier} \cdot t$.
\end{itemize}
The algorithm is deterministic, and initially takes time $m \cdot \gamma_{vertexSparsifier}$. Every update is processed in worst-case time $\gamma_{vertexSparsifier}$.

Further, say the algorithm is given as input a dynamic flat forest $F$ over $H$ and vertex maps $\Pi_{V(H) \mapsto V(F)}$, $\Pi_{V(F) \mapsto V(H)}$, along with parameters $\gamma_{congRep}$ and $\gamma_{recRep}$ such that at any time the vertex congestion $\vcong(\Pi_{V(F) \mapsto V(H)})$ is bounded by $\gamma_{congRep}$ and the number of changes to $F$ caused by an update to $G$ is upper bounded by $\gamma_{recRep}$. We require the vertex maps to be such that whenever a vertex is added to the preimage, its image remains constant for the rest of the algorithm.

Then, the algorithm can maintain a flat forest $F'$ over $G$ along with vertex maps $\Pi_{V(H) \mapsto V(F')}$, $\Pi_{V(F') \mapsto V(G)}$ such that at any time $\vcong(\Pi_{V(F') \mapsto V(G)})$ is bounded by $\gamma_{congRep} \cdot \gamma_{vertexSparsifier}$ and the number of changes to $F'$ per update to $G$ is $\Otil(\gamma_{recRep} + \gamma_{congRep}  \cdot \gamma_{vertexSparsifier})$, and we have for any two vertices $u,v \in V(H)$ that \begin{align}
    &l_G(\Pi_{V(F') \mapsto V(G)}(F'[\Pi_{V(H) \mapsto V(F')}(u), \Pi_{V(H) \mapsto V(F')}(v)]) \\
    \leq~&l_H(\Pi_{V(F) \mapsto V(H)}(F[\Pi_{V(H) \mapsto V(F)}(u),\Pi_{V(H) \mapsto V(F)}(v)]). \label{eq:distupper}
\end{align} Further, we have that the vertex maps are such that whenever a vertex is added to the preimage, its image remains constant for the rest of the algorithm.

Given that inputs $F$, $\Pi_{V(H) \mapsto V(F)}$, and $\Pi_{V(F) \mapsto V(H)}$ are maintained, the algorithm to maintain $F'$, $\Pi_{V(H) \mapsto V(F')}$, and $\Pi_{V(F) \mapsto V(G)}$ requires additional initialization time $\tilde{O}(m \cdot \gamma_{congRep})$ and processes every update with additional worst-case time $\Otil(\gamma_{recRep} + \gamma_{congRep} \cdot \gamma_{vertexSparsifier})$.
\end{theorem}
While \cite[Theorem 4.12]{KMP23} bounds edge congestion, we instead write it as a vertex congestion bound. These are equivalent because $G$ has maximum degree $3$.

\subsection{Maintaining a Sparse Neighborhood Cover} 

For the algorithm, we assume that the number of updates to $G$ to the data structure is at most $m$ and the number of vertices is at most $2m$. These assumptions can be made without loss of generality as the general result can be obtained by simply rebuilding the data structure after $m$ updates to $G$ and by inserting edges with infinite lengths between connected components in $G$.

\paragraph{Data Structures.} We first describe the algorithm to maintain the data structures required to maintain the hierarchical forest $F$, described in \Cref{thm:dyn_snc}. We again use a standard batching technique over levels $0, 1, \ldots, K$ where we take $K = \lceil \log^{1/42} m \rceil$.

\begin{algorithm}[!ht]
\caption{$\textsc{Init}(G, K, D)$}
\label{alg:initSNC}
\label{alg:newfully_dynamic_SNC}
$A_0 \gets V$; $A_1, A_2, \ldots, A_K \gets \emptyset$.\\
\ForEach{$0 \leq i \leq K$}{
    Initialize a data structure $\mathcal{H}_{i}$ as described in \Cref{thm:mainTheoremVSExt} on the graph $G$ where we let $H_i$ denote the corresponding vertex sparsifier maintained by $\mathcal{H}_i$.\\
    $\textsc{InitLevel}(i)$.\label{lne:invInitLevel1}
}
\end{algorithm}

In \Cref{alg:initSNC}, we describe how to initialize the data structures. We start by initializing the sets $A_0, A_1, \ldots, A_K$ where we initialize $A_0$ to be the vertex set of $G$ and all other sets to be empty. We then call for every level $0 \leq i \leq K$, the procedure $\textsc{InitLevel}(i)$ given in \Cref{alg:initLevel}. In this procedure, we initialize for level $i$ three different data structures: a data structure $\mathcal{H}_i$ that maintains a vertex sparsifier $H_i$ over the vertex set $A_i$; a data structure $\mathcal{D}_i$ that maintains a decremental sparse neighborhood cover on $\hat{H}_i$ where $\hat{H}_i$ is the decremental version of $H_i$ where insertions are simply ignored; and finally a data structure $\mathcal{T}_{i,j,C}$ that maintains for every cluster $C$ in some partition of the sparse neighborhood cover a tree $T_{i,j,C}$ that spans the cluster $C$ and has low diameter.

\begin{algorithm}[!ht]
\caption{$\textsc{InitLevel}(i)$}
\label{alg:initLevel}
Delete all current terminal vertices from the terminal vertex set maintained by $\mathcal{H}_{i}$ via the operation $\textsc{RemoveTerminalVertex}(\cdot)$. Then, add all vertices in $\hat{A}_i$ as terminals to the data structure $\mathcal{H}_{i}$ via the operation $\textsc{AddTerminalVertex}(\cdot)$ where $\hat{A}_i$ is taken to be the current set $A_i$. \\
Initialize data structure $\mathcal{D}_i$ as described in \Cref{thm:dec_snc} on graph $\hat{H}_i$ with parameter $D_i \defeq D/(4 \gamma_{decrSnc} \cdot \gamma_{vertexSparsifier})^{K - i + 1}$ and let $\mathcal{P}_{i, 0}, \mathcal{P}_{i, 1}, \ldots, \mathcal{P}_{i, k}$ for some $k = O(\log m)$ denote the partitions maintained by this data structure where $\hat{H}_i$ is initialized to $H_i$ at the current time and then undergoes all deletions that $H_i$ undergoes but not the insertions.\label{lne:initDi}\\
\ForEach{$0 \leq j \leq k$ and cluster $C \in \mathcal{P}_{i, j}$}{
    Initialize data structure $\mathcal{T}_{i,j, C}$ as described in \Cref{thm:mainTheoremLowDiamTree} on the graph $\hat{H}_{i,j, C}$ and maintain the hierarchical tree $T_{i,j, C}$, and embeddings $\Pi_{V(T_{i,j,C}) \mapsto C}$ and $\Pi_{C\mapsto V(T_{i,j,C})}$, where $\hat{H}_{i,j,C}$ is initialized to the current graph $\hat{H}_i[C]$ and then undergoes all deletions that $\hat{H}_i[C]$ undergoes.
}
\end{algorithm}

\begin{algorithm}[!ht]
\caption{$\textsc{Update}(t)$}
\label{alg:newfully_dynamic_SNC_update}
Let $e$ be the edge in $G$ affected by the $t$-th update.\\
\ForEach{$0 \leq i \leq K$}{
    Forward the $t$-th update to $G$ to the data structure $\mathcal{H}_{i}$ that updates graph $H_i$ via the update batch $U_{H_i}^{(t)}$.\\
    Add the endpoints of $e$ and all edges and vertices affected by the update in $H_i$ to the sets $A_{i+1}, A_{i+2}, \ldots, A_K$ (recall that $V(H_i) \subseteq V(G)$, as stated in \Cref{thm:mainTheoremVSExt}).
    \label{lne:moveToHigherLevel}\\    
    Forward the deletions in $U_{H_i}^{(t)}$ to the data structure $\mathcal{D}_i$ that maintains the partitions $\mathcal{P}_{i, 0}, \mathcal{P}_{i, 1}, \ldots, \mathcal{P}_{i, k}$. We denote by $\mathcal{P}_{i, 0}^{OLD}, \mathcal{P}^{OLD}_{i, 1}, \ldots, \mathcal{P}^{OLD}_{i, k}$ the partitions before the update was processed and by $\mathcal{P}_{i, 0}^{NEW}, \mathcal{P}^{NEW}_{i, 1}, \ldots, \mathcal{P}^{NEW}_{i, k}$ the ones obtained after the update.\label{lne:maintainDi}\\
    \ForEach{$0 \leq j \leq k$ and cluster $C \in \mathcal{P}_{i, j}^{OLD}$}{
        Forward the deletions in $U_{H_i}^{(t)}$ to the data structure $\mathcal{T}_{i,j, C}$ that maintains the flat tree $T_{i,j,C}$ and embeddings $\Pi_{V(T_{i,j,C}) \mapsto C}$, $\Pi_{C \mapsto V(T_{i,j,C})}$.
    }
    \ForEach{$0 \leq j \leq k$ and cluster $C \in \mathcal{P}_{i, j}^{NEW} \setminus \mathcal{P}_{i, j}^{OLD}$}{
        \If{there is a cluster $C'$ in partition $\mathcal{P}_{i, j}^{OLD}$ with $C \subseteq C'$ and $|C| > |C'|/2$}{
            Let $\mathcal{T}_{i,j, C}$ refer to the data structure $\mathcal{T}_{i,j, C'}$ after deleting all edges incident to vertices in $C' \setminus C$ from the data structure; let $T_{i,j,C}$ denote the corresponding tree.
        }\Else{
              Initialize data structure $\mathcal{T}_{i,j, C}$ as described in \Cref{thm:mainTheoremLowDiamTree} on the graph $\hat{H}_{i,j, C}$ and maintain the flat tree $T_{i,j, C}$ and embeddings $\Pi_{V(T_{i,j,C}) \mapsto C}$, $\Pi_{C \mapsto V(T_{i,j,C})}$ where $\hat{H}_{i,j,C}$ is initialized to the current graph $\hat{H}_i[C]$ and then undergoes all deletions that $\hat{H}_i[C]$ undergoes.\label{lne:initNewTree}
        }
    }
}

\If(\label{lne:ifRebuilt}){$\exists i$ s.t. $0 \leq i < K$ and $|A_{i+1}| \geq m^{1-(i+1)/K}$}{
    Let $i$ be the smallest index that satisfies the if-condition.\\
    $A_{i+1}, A_{i+2}, \ldots, A_K \gets \emptyset$.\label{lne:reset}\\
    \lForEach{$i \leq j \leq K$}{ $\textsc{InitLevel}(j)$. \label{lne:invInitLevel2}}
}

\end{algorithm}

The maintenance of all of these objects is given in \Cref{alg:newfully_dynamic_SNC_update} explains how to process the $t$-th update to $G$. The procedure forwards changes to $G$ to all data structures at every level, updating the vertex sparsifier, partitions, and the low-diameter spanning trees to those on the updated graph. However, it handles only the decremental updates, i.e. deletions of edges and vertices. To handle insertions, it adds the endpoints of all inserted edges and vertices to all sets $A_{i+1}, A_{i+2}, \ldots, A_K$ at higher levels (see \Cref{lne:moveToHigherLevel}). 

Whenever too many insertions occur at a certain level, the data structures at the level are rebuilt (see the if-statement starting in \Cref{lne:ifRebuilt}).

\paragraph{Maintaining the objects in \Cref{thm:dyn_snc}.} We discuss how to take the above \Cref{alg:initSNC,alg:initLevel,alg:newfully_dynamic_SNC_update} and use them to maintain $F, \Pi_{V(F)\to V}$, and subset $S$ as requierd by \Cref{thm:dyn_snc}. Then we analyze the overall construction in the next section (\Cref{subsec:analyzesnc}). For a cluster $C \in \mathcal{P}_{i,j}$, we maintain a vertex $\pi(i,j,C) \in C \cap A_{i+1}$ if it exists, and otherwise set $\pi(i,j,C) = \perp$. For $i = K, K-1, \dots, 0$, we will define the flat forest $F_i$ on $G$ for terminals $A_i$ (which comes with a map $\Pi_{V(F_i) \mapsto V(G)}$), and a subset $S_i \subseteq V(F_i)$, which will contain duplicates of vertices in $A_i$ (at most $m^{o(1)}$ times). Ultimately we will set $S = S_0$, $F = F_0$, and $\Pi_{V(F) \mapsto V(G)}$ as $\Pi_{V(F_0) \mapsto V(G)}$.
For the base case, let $F_{K+1}$ and $S_{K+1}$ be empty sets.

To go from level $i+1$ to $i$, first use \Cref{thm:mainTheoremVSExt} with $F$ as the disjoint union of $T_{i,j,C}$ with maps $\Pi_{V(T_{i,j,C}) \to V(\hat{H}_i)}$ (which is a flat embedding of $\hat{H}_i$ by \Cref{thm:mainTheoremLowDiamTree}), to get embeddings $\Pi_{V(T_{i,j,C}') \to V(G)}$ for some flat trees $T_{i,j,C}'$ over $G$.

Initialize $F_i = F_{i+1}$. Now for each $C \in \mathcal{P}_{i,j}$ for $j = 0, 1, \dots, k$ add a disjoint copy of $T_{i,j,C}'$ to $F_i$. If $C \cap A_{i+1} \neq \emptyset$, then $\pi(i,j,C) \neq \perp$. Now, for every vertex $s \in S_{i+1}$ corresponding to a duplicate of $\pi(i,j,C)$, take vertex $\Pi_{V(\hat{H}_i) \mapsto V(T_{i,j,C}')}(\pi(i,j,C)) \in V(T_{i,j,C}')$ and merge it with $F_{i+1}$ at vertex $s \in V(F_{i+1})$.
Finally, define $S_i$ to be the union of all copies of vertices in $A_i$ in the copies of $T_{i,j,C}'$, which we can tell by the maps $\Pi_{V(H) \to V(T_{i,j,C}')}$ given by \Cref{thm:mainTheoremVSExt}.

\subsection{Analysis of Fully-Dynamic SNC Algorithm}
\label{subsec:analyzesnc}
We start by showing that the ball of radius $D/\gammasnc$ around each $v \in V(G)$ is covered by our construction. The remaining aspects of \Cref{thm:dyn_snc} follow more directly by our construction. We start by establishing some auxiliary claims.

\begin{claim}\label{clm:dichotomyClusters}
Recall that $D_i \defeq D/(4 \cdot \gamma_{decrSnc} \cdot \gamma_{vertexSparsifier})^{K - i + 1}$. Then, for any vertex $v \in V(G)$, for every $0 \leq i \leq K$, then
\begin{itemize}
    \item $A_{i} \cap B_G(v, D_i/ (2\gamma_{decrSnc}\gamma_{vertexSparsifier})) \subseteq C$ for some cluster $C \in \mathcal{P}_{i,j}$ for some $0 \leq j \leq k$, or
    \item for each vertex $w \in A_{i} \cap B_G(v, D_i/ (2\gamma_{decrSnc}\gamma_{vertexSparsifier}))$, there is a cluster $C \in \mathcal{P}_{i,j}$ for some $0 \leq j \leq k$ that contains $w$, and $C \cap A_{i+1} \neq \emptyset$ (we define $A_{K+1} = \emptyset$).  
\end{itemize}
\end{claim}
\begin{proof}
The proof is trivial for the case where $A_{i} \cap B_G(v, D_{i}) = \emptyset$. Therefore, let us assume that there is at least one such vertex.

We have for every level $0 \leq i \leq K$ that the data structure $\mathcal{D}_i$ as described in \Cref{thm:dec_snc}, with parameter $D_i$, maintains a collection of partitions $\mathcal{P}_{i, 0}, \mathcal{P}_{i, 1}, \ldots, \mathcal{P}_{i, k}$ of decremental graph $\hat{H}_i$ such that for every vertex $w \in V(\hat{H}_i)$, there is an index $0 \leq j \leq k$, such that $B_{\hat{H}_i}(w, D_i/\gamma_{decrSnc}) \subseteq C$ for some $C \in \mathcal{P}_{i,j}$.

Here, as can be seen from how $\mathcal{D}_i$ is initialized (see \Cref{lne:initDi}) and maintained (see \Cref{lne:moveToHigherLevel}), that $\hat{H}_i$ differs from $H_i$ in the way that $\hat{H}_i$ was only updated by the deletions to $H_i$ since $\mathcal{D}_i$ was initialized, and thus insertions to $H_i$ since, are not present in $\hat{H}_i$. On the other hand, from how we update the set $A_{i+1}$ (see \Cref{lne:moveToHigherLevel} and \Cref{lne:ifRebuilt}), we have that every endpoint of an inserted edge to $H_i$ since $\mathcal{D}_i$ was initialized is added to set $A_{i+1}$. 

Now, let $w$ be an arbitrary vertex in $A_{i} \cap B_G(v, D_i/ (2\gamma_{decrSnc}\gamma_{vertexSparsifier}))$. We have that every vertex $w' \in A_{i} \cap B_G(v, D_i/(2\gamma_{decrSnc}\gamma_{vertexSparsifier}))$ satisfies $\dist_{H_i}(w,w') \leq \gamma_{vertexSparsifier} \cdot \dist_G(w,w') \leq   \gamma_{vertexSparsifier}  \cdot (\dist_G(w, v) + \dist_G(v, w')) \leq D_i/\gamma_{decrSnc}$. 

Now, we either have that $B_{\hat{H}_i}(w, D_i/\gamma_{decrSnc}) = B_{H_i}(w, D_i/\gamma_{decrSnc})$ in which case, we can conclude from \Cref{thm:dec_snc} that there indeed is a cluster $C$ in some partition $\mathcal{P}_{i,j}$ that contains all vertices $w' \in A_{i} \cap B_G(v, D_i/(2\gamma_{decrSnc}\gamma_{vertexSparsifier}))$. 

On the other hand, if $B_{\hat{H}_i}(w, D_i/\gamma_{decrSnc}) \neq B_{H_i}(w, D_i/\gamma_{decrSnc})$, then we have that there was an insertion to $H_{i}$ with at least one of the endpoints in the ball $B_{\hat{H}_i}(w, D_i/\gamma_{decrSnc})$. This implies that $B_{\hat{H}_i}(w, D_i/\gamma_{decrSnc}) \cap A_{i+1} \neq \emptyset$. And therefore, we have from \Cref{thm:dec_snc} that there indeed is a cluster $C$ in some partition $\mathcal{P}_{i,j}$ that contains $B_{\hat{H}_i}(w, D_i/\gamma_{decrSnc})$ and thus $C \cap A_{i+1} \neq \emptyset$, as desired.
\end{proof}

\begin{claim}\label{clm:coverClaim}
Recall that $D_i \defeq D/(4 \cdot \gamma_{decrSnc} \cdot \gamma_{vertexSparsifier})^{K - i + 1}$. For every vertex $v \in V(G)$ and $0 \leq i \leq K$, there is a tree $T \in F_i$ such that \[ B_G(v, D_i/ (2\gamma_{decrSnc}\gamma_{vertexSparsifier})) \cap A_i \subseteq \Pi_{V(F_i) \to V(G)}(S_i \cap V(T)). \]
\end{claim}
\begin{proof}
We prove the claim by induction. Fix any vertex $v \in V(G)$ and level $0 \leq i \leq K$. We have from \Cref{clm:dichotomyClusters} that we are in one of two scenarios: in the first scenario, we have that $B_G(v, D_i/ (2\gamma_{decrSnc}\gamma_{vertexSparsifier})) \cap A_i \subseteq C$ for some $C \in \mathcal{P}_{i,j}$ for some $0 \leq j \leq k$. And further note that we have a data structure $\mathcal{T}_{i,j,C}$ that maintains a hierarchical tree $T_{i,j,C}$ over $C$ on graph $\hat{H}_{i,j,C} = \hat{H}[C]$, and that a disjoint copy of $T_{i,j,C}$ is in $F_i$ by construction. Because $S_i$ contains all the vertices in $A_i$ among $T_{i,j,C}$, the claim follows for $T = T_{i,j,C}$.

Otherwise, we have $i < K$, and that for every $w \in A_{i} \cap B_G(v, D_i/(2\gamma_{decrSnc}\gamma_{vertexSparsifier}))$, there is a cluster $C_w \in \mathcal{P}_{i,j}$ for some $0 \leq j \leq k$, such that $C_w \cap A_{i+1} \neq \emptyset$. But note that each such cluster $C_w$ is thus represented by vertex $v_w = \pi(i,j, C_w) \in A_{i+1}$. Since $\diam(\hat{H}_i[C_w]) \leq D_i$ by \Cref{thm:dec_snc}, $\hat{H}_i \subseteq H_i$, and distances in $H_i$ dominate distances in $G$ by \Cref{thm:mainTheoremVSExt}, we have $\dist_G(w, v_w) \leq D_i$. Thus, all such vertices $v_w$ over all $w \in A_{i} \cap B_G(v, D_i/(2\gamma_{decrSnc}\gamma_{vertexSparsifier}))$ have distance at most $D_i + D_i/(2\gamma_{decrSnc}\gamma_{vertexSparsifier}) \leq 2 \cdot D_i = D_{i+1} / (2\gamma_{decrSnc}\gamma_{vertexSparsifier})$ from $v$ in graph $G$. Thus by the inductive hypothesis, there is at least one tree $T$ in forest $F_{i+1}$ such that all such vertices $v_w$ are contained in $\Pi_{V(F_{i+1}) \mapsto V(G)}(S_{i+1} \cap V(T))$.

Thus, in $F_i$, we have the tree $T$ where for each $w \in A_{i} \cap B_G(v, D_i/(2\gamma_{decrSnc}\gamma_{vertexSparsifier}))$ (and possibly some other vertices), we attach to every node $x \in V(T)$ that is identified with $v_w$, the tree $T_{i,j, C_w}'$ for the $0 \leq j\leq k$ where $C_w \in \mathcal{P}_{i,j}$. Thus, the resulting tree $T' \supseteq T$ in $F_i$ has its node set $V(T')$ mapped to a set of vertices in $V(G)$ that contains vertices in $\cup_w C_w \supseteq B_G(v, D_i/ (2\gamma_{decrSnc}\gamma_{vertexSparsifier}))$, as desired.
\end{proof}

The following Corollary now follows from \Cref{clm:coverClaim} and the fact that $A_0 = V$ at all times.

\begin{corollary}\label{cor:covered}
For every vertex $v \in V(G)$, there is a tree $T \in F$ such that $B_G(v, D/(4 \cdot \gamma_{decrSnc} \cdot \gamma_{vertexSparsifier})^{K + 2}) \subseteq \Pi_{V(F) \mapsto V(G)}(S \cap V(T))$.
\end{corollary}

We establish the Properties listed in \Cref{thm:dyn_snc}.

\begin{claim}
The data structure maintains a forest $F$ along with graph embedding $\Pi_{V(F) \mapsto V}$ and $S \subseteq V(F)$ such that for some $\gamma_{snc}= e^{O(\log^{41/42} m \log\log m)}$, at any time:
\begin{enumerate}
\item $F$ with map $\Pi_{V(F)\mapsto V}$ is a flat embedding of $G$.
\item For any vertex $v \in V$ there is a tree $T \in F$ such that
\[ B_G(v, D/\gammasnc) \subseteq \Pi_{V(F)\mapsto V}(S \cap V(T)). \]
\item For any tree $T \in F$ we have that $\diam_F(S \cap V(T)) \le \gammasnc \cdot D$.
\item The congestion satisfies $\vcong(\Pi_{V(F)\mapsto V}) \le \gammasnc$.
\end{enumerate}
\end{claim}
\begin{proof}
Let us establish the properties one by one:
\begin{enumerate}
    \item This follows by induction. Indeed, assume by induction that $F_{i+1}$ flatly embeds into $G$. Now, each $T_{i,j,C}$ and $T_{i,j,C}'$ flatly embeds into $G$ by \Cref{thm:mainTheoremLowDiamTree} and \Cref{thm:mainTheoremVSExt} (we got $T_{i,j,C}'$ by mapping from $\hat{H}_i$ to $G$). Finally, we attach $T_{i,j,C}'$ to $F_{i+1}$ at matching vertices by construction.
    \item The first property follows immediately from \Cref{cor:covered} and by picking $\gamma_{snc} \geq (4 \cdot \gamma_{decrSnc} \cdot \gamma_{vertexSparsifier})^{K + 2}$.
    \item By \Cref{thm:dec_snc}, \Cref{thm:mainTheoremLowDiamTree}, and the mapping back procedure from \Cref{thm:mainTheoremVSExt}, we know:
    \begin{align*}
        \diam_{T_{i,j,C}'}(\Pi_{V(\hat{H}_i) \mapsto V(T_{i,j,C}')}(C)) &\stackrel{(a)}{\le} \diam_{T_{i,j,C}}(\Pi_{C \to V(T_{i,j,C})}(C)) \\ &\stackrel{(b)}{\le} \gamma_{lowDiamTree} \diam(\hat{H}_i[C]) \le \gamma_{lowDiamTree} D_i,
    \end{align*} where $(a)$ follows by the construction of $T_{i,j,C}'$ and \Cref{thm:mainTheoremVSExt} \eqref{eq:distupper}, and $(b)$ follows from \Cref{thm:mainTheoremLowDiamTree}. Because each $F_i$ was gotten from $F_{i+1}$ by attaching trees to vertices of $F_i$, we get that for any $T_i \in F_i$ that contains $T_{i+1} \in F_{i+1}$,
    \[ \diam_{F_i}(S_i \cap V(T_i)) \le 2 \gamma_{lowDiamTree} D_i + \diam_{F_{i+1}}(S_{i+1} \cap V(T_{i+1})), \]
    by the fact that we can go from $u \in T_{i,j,C}$ to $v \in T_{i',j',C'}$ via $u \to \pi(i, j, C) \to \pi(i', j', C') \to v$, and that $\pi(i,j,C), \pi(i',j',C')$ are attached to $F_{i+1}$ at vertices in $S_{i+1}$. Iterating this shows that $\diam_F(S \cap V(T)) \le \sum_{i=0}^K 2\gamma_{lowDiamTree} D_i \le 4\gamma_{lowDiamTree} D_i$.
    \item By \Cref{thm:mainTheoremLowDiamTree,thm:mainTheoremVSExt} and the fact that there are $k+1$ partitions in the decremental SNC, the vertex congestion increases by a factor of at most $O((k+1)\gamma_{lowDiamTree}\gamma_{vertexSparsifier})$ from $F_{i+1}$ to $F_i$, because we connect a copy of each $T_{i,j,C}'$ to everything in $S_i$. Thus, the total vertex congestion is at most $O((k+1)\gamma_{lowDiamTree}\gamma_{vertexSparsifier})^K$.
\end{enumerate}
\end{proof}

Finally, we bound the total runtime to maintain the forest $F$ and the embedding $\Pi_{F \mapsto G}$.

\begin{claim}
For some reasonably chosen $\gamma_{snc}= e^{O(\log^{41/42} m \log\log m})$, the algorithm can be initialized in time $m \cdot \gamma_{snc}$, and thereafter processes each edge insertion/deletion in amortized time $\gamma_{snc}$.
\end{claim}
\begin{proof}
We start by bounding the number of changes to each set $A_i$ by $m \cdot (5\gamma_{vertexSparsifier})^i$. We prove by induction on $i$. For the base case, we have that $A_0$ initially has all vertices added and is thereafter not updated at any point of the algorithm. 

For the inductive step, observe that the set $A_{i+1}$ can undergo insertions and deletions. We upper bound the number of vertex insertions to $A_{i+1}$ by $\frac{1}{2} m \cdot (5\gamma_{vertexSparsifier})^{i+1}$ which suffices since the number of deletions from $A_{i+1}$ is at most the number of insertions. We first note that $A_{i+1}$ undergoes at most $2$ vertex insertions for every edge update to sparsifier $H_j$ for $j < i+1$ (see \Cref{lne:moveToHigherLevel}). But we have from \Cref{thm:mainTheoremVSExt}, that each sparsifier $H_j$ undergoes at most $\gamma_{vertexSparsifier}$ changes for every time the underlying graph $G$ is changed (which occurs at most $m$ times), or when a terminal is added or removed from $A_{j}$ (which occurs at most $m \cdot (5\gamma_{vertexSparsifier})^j$ times by the induction hypothesis. Thus, the total number of insertions to $A_{i+1}$ can be upper bounded by 
\begin{align*}
\sum_{j < i+1} \gamma_{vertexSparsifier} \cdot &(m + m \cdot (5\gamma_{vertexSparsifier})^j) \\&\leq (i+1) \gamma_{vertexSparsifier} + 2 \cdot 5^i m \cdot (\gamma_{vertexSparsifier})^{i+1} \\&< \frac{1}{2} \cdot (5 \cdot \gamma_{vertexSparsifier})^{i+1} 
\end{align*}
where we use that we have a geometric sum and that $\gamma_{vertexSparsifier}$ is sufficiently large. 

Next, we observe that whenever we enter the if-statement in \Cref{lne:ifRebuilt}, where $i$ is selected as the smallest index, we have that the if-statement sets $A_{i+1}$ to the empty set which means that for each such if-statement, we have that at least $m^{1-(i+1)/K}$ updates to $A_{i+1}$ occur which bounds the number of times that the if statement is executed with smallest index $i$ by \[ m \cdot (5\gamma_{vertexSparsifier})^{i+1}/ m^{1-(i+1)/K} = m^{(i+1)/K} \cdot (5\gamma_{vertexSparsifier})^{i+1}. \]

Finally, we have from the if-condition that whenever we enter the if-statement and restart a level $j \geq i$ where $i$ is the smallest index for which the if-condition holds, for all such levels the set $A_j$ is of size at most $m^{1-j/K}$ by minimality of $i$ (and the fact that $A_0 = V$). 

From the number of restarts of each level and the upper bound on the size of $A_i$ at any time when any data structure is rebuilt, the number of amortized recourse to the sets $A_i$ and the graphs $H_i$, we can then bound the total update time of all data structures $\mathcal{H}_i, \mathcal{D}_i$ and all data structures $\mathcal{T}_{i,j,C}$ from \Cref{thm:dec_snc}, \Cref{thm:mainTheoremLowDiamTree} and \Cref{thm:mainTheoremVSExt} by 
\begin{align*}
m^{1+1/K} \cdot (5\gamma_{vertexSparsifier})^{O(K)} \cdot \gamma_{lowDiamTree} \cdot \gamma_{decrSnc} &= m^{1+1/K} (e^{O(\log^{20/21}(m)\log\log m)})^{O(K)} \\ &= m \cdot e^{O(\log^{41/42}(m)\log\log m)}
\end{align*}
where we implicitly use a halving trick for the data structures $\mathcal{T}_{i,j,C}$ which is necessary because sometimes data structures $\mathcal{T}_{i,j,C}$ are newly initialized in \Cref{lne:initNewTree} but each edge in $\hat{H}_i$ participates in at most $O(\log m)$ such rebuilds as its new cluster $C$ then is incident to at most half the number of edges as its old cluster was.

Finally, in order to maintain the flat forest $F$, embedding $\Pi_{V(F) \mapsto V(G)}$, and set $S$, the runtime increases by a factor of \[ O((k+1) \cdot \gamma_{lowDiamTree} \cdot \gamma_{vertexSparsifier})^{O(K)} = e^{O(\log^{41/42}(m)\log\log m)} \] due to vertex congestion. Thus the total time is $m \cdot e^{O(\log^{41/42}(m)\log\log m)}$ as desired.
\end{proof}

\section{Fully-Dynamic Low-Recourse Spanner}
\label{sec:newSpanner}

In this section, we prove that given a graph $G=(V,E,l)$ undergoing edge deletions and insertions, vertex splits and vertex insertions (possibly non-isolated), we can maintain a spanner $H$ of $G$ with low recourse per update to $G$. Our result is summarized by the theorem below.

\NewSpanner*

We develop an incremental spanner that we use in conjunction with the batching scheme obtaining a spanner in a decremental graph developed by \cite{chen2022maximum, detMaxFlow}.

\subsection{Maintaining a Spanner under Edge Insertions}

In this section, we prove the following theorem by giving an algorithm that maintains a spanner $H$  with very low recourse for a graph $G$ undergoing edge insertions. It additionally maintains an embedding $\Pi$ of the edges in $G$ to paths in $H$ between the same endpoints. We rely on $G$ having small maximum degree at all times to obtain an efficient algorithm.

In the next section, we crucially rely on this result to build our fully-dynamic spanner. To be of use, we need our algorithm to also take vertex splits and deletions under consideration. That is, while we do not require the current spanner and embedding to be adjusted to a vertex split/deletion, we require all embedding paths added in the future to take the vertex split/deletion into account.

Even in this setting, vertex splits are quite challenging to deal with. If split vertices are allowed to reach the maximum degree threshold again, this causes problems for our algorithm. We thus introduce the following definition that formalizes how we further constrain the update sequence. 

\begin{definition}[Master node]
    We initially associate each vertex $v$ in the initial vertex set $V$ of $G$ with a unique master node $v = \master(v)$. Then, whenever $v$ is split into $v_1$ and $v_2$, we let $\master(v_1) = \master(v)$ and $\master(v_2) = \master(v)$. 
    For every $v \in V$, we define $\deg^{\master}(v) = \sum_{u \in V_G: \master(u) = v} \deg(u)$ and call this the master node degree of $v$.
\end{definition}

Before we state our result, we also define the length of an embedding path. 

\begin{definition}(Embedding length)
    \label{def:spanner_length}
    Given an embedding $\Pi_{G \mapsto H}$ for some graphs $G$ and $H$, we let $\length(\Pi_{G \mapsto H}) = \max_{e}|\Pi_{G \mapsto H}(e)|$ denote the maximum number of edges on a (possibly broken) embedding path. 
\end{definition}

Finally, we summarize the main technical result of this section. To build intuition, we recommend the reader ignore vertex splits when reading the section for the first time.

\begin{restatable}{theorem}{incrSpanner}\label{thm:incremental_spanner}
Given an integer $1 \leq \Delta \leq n$ and an $n$-vertex, unweighted, undirected graph $G = (V,E)$ that is initially empty and undergoes up to $n\Delta$ edge insertions and $n$ vertex splits/edge deletions such that the master node degree $\deg^{\master}(v)$ is bounded by $\Delta$ for every vertex $v$ throughout.

For some fixed $\gamma_{\mathrm{incSpanner}} = e^{O(\log^{20/21} m \log \log m)}$, there is an algorithm $\textsc{IncrementalSpanner}(G)$ that maintains a sparse subgraph $H \subseteq G$ such that whenever an edge $e$ is added to $H$, it remains in $H$ until the end of the algorithm, and an embedding $\Pi$ that maps each edge $e = (u,v)$ to a $uv$-path $\Pi(e)$ in the graph $G$ at the time of the edge insertion. Afterwards, the embedding path remains fixed, except that the endpoints of edges that were split are updated to reflect the new endpoints of the edge.

The algorithm ensures that at any time, $\Pi$ has vertex congestion at most $\gamma_{\mathrm{incSpanner}} \Delta$ and embedding length at most $\gamma_{\mathrm{incSpanner}}$. The algorithm takes total time $\tilde{O}(n\Delta \gamma_{\mathrm{incSpanner}})$ to process the entire sequence of edge insertions.
\end{restatable}

\paragraph{Additional Preliminaries: Dynamic APSP. } In this section, we use the following result from \cite{KMP23} in our algorithm.

\begin{restatable}{theorem}{mainTheoremAPSP}\label{thm:mainTheoremAPSP}
Given an $m$-edge input graph $G = (V,E,l)$ with lengths in $[1,L]$, there is a data structure $\textsc{DynamicAPSP}$ that can process a polynomial\footnote{In this paper, the term polynomial always refers to a polynomial in $m$.} number of edge insertions and deletions to $G$ and at any point in time answers queries where inputted $u,v \in V$, it returns a distance estimate $\widehat{\dist}(u,v)$ such that $\dist_G(u,v) \leq \widehat{\dist}(u,v) \leq \gamma_{approxAPSP} \cdot \dist_G(u,v)$,  for some $\gamma_{approxAPSP} =  e^{O(\log^{6/7} m \log\log m)}$. It can further report a $uv$-path $P$ with $l(P) \leq \gamma_{approxAPSP} \cdot \dist_G(u,v)$.

For some $\gamma_{timeAPSP} =  e^{O(\log^{20/21} m \log\log m)}$, the data structure can be initialized in time $m \cdot \gamma_{timeAPSP} \cdot \log L$, processes each edge update in worst-case time $\gamma_{timeAPSP} \cdot \log L$ and each query in worst-case time $O(\log m \log L)$, and reports a path $P$ in worst-case time $O(|P|\log m\log L)$.
\end{restatable}

\paragraph{Algorithm.} We give the pseudocode for our algorithm in \Cref{alg:incremental_spanner}. Upon initialization, our algorithm is given the number of vertices and a bound on the maximum degree $\Delta$ a node can attain throughout the insertion sequence. We assume that the graph is initially empty. If there is a non-empty initial graph, we first insert the edges one by one as if they were edge insertions. Our construction consists of $K = 10\log n$ layers. The goal of each layer $i$ is to embed a large constant fraction of the edges that layer $i - 1$ was not able to embed. To do so, each layer $i$ maintains a spanner $H_i$, and a sub-graph $\widehat{H}_i$ of $H_i$. It further maintains an APSP data structure $\mathcal{D}_i$ on the graph $\widehat{H}_i$. All these graphs are initialized to empty graphs on $n$ vertices. We will maintain that the total number of edges seen by layer $i$ is at most $n \Delta/2^i$.  

Whenever an edge is inserted we pass it to layer $0$. The only time an edge $e$ ever gets passed to a layer is when it is inserted. It then gets passed down to layers of increasing index until a layer $j$ is handling it. Edge $e$ will then forever be handled by layer $j$ and layers $0, \ldots, j$ are said to have seen it. 

Whenever an edge $e = (u,v)$ is passed to layer $i$, the algorithm distinguishes the following three cases:
\begin{enumerate}
    \item \underline{if $\mathcal{D}_i$ returns a $uv$-path $P$ of length at most $2 \gamma_{\mathrm{approxAPSP}} \log n$ in $\widehat{H}_i$:} then set the new embedding path $\Pi(e)$ of $e$ equal to the found path $P$. Further, remove edges $e'$ on $P$ with congestion larger $2\gamma_{\mathrm{approxAPSP}} \cdot \Delta \cdot \log 2n/2^i$ from $\widehat{H}_i$ and $\mathcal{D}_i$. 
    \item \underline{if the returned path $P$ is of greater length \textbf{and} $\deg^{\master}_{H_i}(u), \deg^{\master}_{H_i}(v) \leq 32 \cdot 2^i$:} then, add $e$ to $H_i$, $\widehat{H}_i$ and $\mathcal{D}_i$. 
    \item \underline{otherwise:} pass $e$ to layer $i + 1$. 
\end{enumerate}

Whenever an edge is deleted, it gets deleted from the spanner. 

Whenever a vertex is split, we simulate it by removing the smaller side and adding it to some isolated vertex for every layer $i$. We update $H_i$, $\widehat{H}_i$, and $\mathcal{D}_i$ accordingly. Notice that we maintain the embedding $\Pi$ for the re-inserted edges and that we refer to them as the same edges as before. In particular, we do not invoke the procedure $\textsc{Insert}$ on edges that are re-inserted.

\begin{algorithm}
\SetKwProg{myalg}{Procedure}{}{}
\myalg{$\textsc{Init}(n, \Delta)$}{
    \tcp{We assume we are given an empty graph. If there is a non-empty initial graph, we simply insert its edges via $\textsc{Insert}$.}
    $K \gets 10 \log n$ \\
    $H_0, H_1, \ldots, H_{K} \gets (V, \emptyset)$ \\
    $\widehat{H}_0, \widehat{H}_1, \ldots, \widehat{H}_{K} \gets (V, \emptyset)$ \\
    Let $H := \bigcup_{i = 0}^{K} H_i$ throughout.\\
    Initialize APSP datastructure $\mathcal{D}_{i}$ on $\widehat{H}_i$ for all $i = 0, \ldots, K$.\\
    $\Pi \gets \emptyset$
}
\myalg{$\textsc{Insert}(e = (u,v))$}{
    \For{$i = 0, \ldots, K$}{
        \If{$\mathcal{D}_{i}.\textsc{Dist}(u, v) \leq 2 \gamma_{\mathrm{approxAPSP}} \cdot \log(2n)$}{ \label{ln:girth}
        $P \gets \mathcal{D}_i.\textsc{Path}(u, v)$. \\ 
        $\Pi(e) \gets P$. \\
        \ForEach{$e' \in P$ where $\econg(\Pi, e') \geq 2\gamma_{\mathrm{approxAPSP}} \cdot \Delta \cdot \log n/2^i$}{
            Delete $e'$ from $\widehat{H_i}$ and therefore also from the data structure $\mathcal{D}_{i}$. \\
        }
        \Return $H, \Pi$\\
        }\ElseIf{$\max(\deg^{\master}_{H_i}(u), \deg^{\master}_{H_i}(v)) \leq 32 \cdot 2^i$}{
            Add $e$ to $H_i$ and $\widehat{H}_i$ and update $\mathcal{D}_{i}$ accordingly. \\
            \Return $H, \Pi$\\
        }
    }
}
\myalg{$\textsc{DeleteEdge}(e)$}{
    \For{$i = 0, \ldots, K$}{
        Delete edge $e$ from $H_i$, $\widehat{H}_i$. 
    }
}
\myalg{$\textsc{Split}(v, E_{\mathrm{move}}, E_{\mathrm{cross}})$}{\tcp{We may assume wlog that there are no self loops, and thus no crossing edges by delaying their insertion.}
    \For{$i = 0, \ldots, K$}{
        Assume wlog that $|E_{\mathrm{move}} \cap E_{H_i} | \leq |\deg_{H_i}(v)|/2$. \\
        Delete all edges in $E_{H_i} \cap (E_{\mathrm{move}} \cup  E_{\mathrm{cross}})$ from $H_i$, $\widehat{H}_i$ and $\mathcal{D}_i$ and re-insert them adjacent to an isolated vertex. Preserve embeddings. \\
    }
    \Return $H, \Pi$\\
}
\caption{$\textsc{IncrementalSpanner}()$ } 
\label{alg:incremental_spanner}
\end{algorithm}

\paragraph{Proof of \Cref{thm:incremental_spanner}.}

We first bound the total number of edges in $H_i$. We do this by bounding the number of edges in $H_i \setminus \hat{H}_i$ by using congestion bounds. The number of edges in $\hat{H}_i$ is $O(n)$ because its girth is at least $2\log n$.

\begin{claim}
    \label{clm:inc_spanner_size}
    Assume the total number of edges ever seen by layer $i$ is $\Delta n/2^i$. Then $|E_{H_i}| \leq 9 n$ for all $i = 0, \dots, K$. 
\end{claim}
\begin{proof}
    We first analyse the number of edges in $H_i \setminus \widehat{H}_i$. Every edge that is in $H_i \setminus \widehat{H}_i$ has at least $2\gamma_{\mathrm{approxAPSP}} \cdot \Delta \cdot \log n/2^i$ edges embedded into it. Since the total number of edges passed to layer $i$ is at most $\Delta n/2^i$ and each edge has an embedding path of length at most $2\gamma_{\mathrm{approxAPSP}} \log n$, the total number of edges in $H_i \setminus \widehat{H}_i$ is at most $n$. 

    By standard techniques for analysing spanners \cite{althofer1993sparse} we have that the number of edges in $\widehat{H}_i$ is at most $8n$ because it has girth $> 2 \log(2n)$ and at most $2n$ vertices (since there are initially $n$ vertices and the at most $n$ vertex splits can increase the number of vertices by at most $n$). 
\end{proof}

We then analyse the total number of edges passed down by layer $i$. 

\begin{claim}
    \label{clm:spanner_reduction}
    Assume the total number ever seen by layer $i$ is $\Delta n/2^i$. Then the total number of edges ever seen by layer $i + 1$ is at most $\Delta n/2^{i + 1}$. 
\end{claim}
\begin{proof}
    Consider the final spanner $H_i$ computed by layer $i$. Let $S_i = \{v \in V: \deg^{\master}_{H_i}(v) \geq 32 \cdot  2^i\}$. By \Cref{clm:inc_spanner_size} the total number of edges in $H_i$ is at most $16n$. Therefore, the number of master nodes with degree larger than $32 \cdot  2^i$ is at most $16n / (32 \cdot  2^i) = n/2^{i + 1}$. For each of these master nodes, we pass down at most $\Delta$ edges. Therefore, the number of edges passed down is at most $n\Delta/2^{i + 1}$.
\end{proof}

We perform a simple induction on the previous two claims.

\begin{lemma}
    \label{lem:size_and_reduction_spanner}
    For all $i = 0, \dots, K$ we have
    \begin{enumerate}
        \item $|E_{H_i}| \leq 9n$ and
        \item the number of edges seen by layer $i$ is at most $\Delta n/2^{i}$.
    \end{enumerate}
\end{lemma}
\begin{proof}
    We show the claim by induction. Consider the base case $i = 0$. The second point is true because the total number of edges inserted and this passed to layer $0$ is at most $\Delta n$. The first point is true because of \Cref{clm:inc_spanner_size}. We then assume that the claim hold for layer $i$. Then, we have that at most $n\Delta/2^{i + 1}$ edges get passed to $i + 1$ by \Cref{clm:spanner_reduction} which shows the second point. The first point then follows from \Cref{clm:inc_spanner_size}.
\end{proof}

Next, we show that the vertex congestion of $\Pi$ is low. 

\begin{lemma}
    \label{lem:cong_spanner}
    For every vertex $v$, we have that $\sum_{e = (v, u) \in H} \econg(\Pi, e) \leq   320 \gamma_{approxAPSP} \Delta \log^2 n$.
\end{lemma}
\begin{proof}
    We show the bound $\sum_{e = (v, u) \in H_i} \econg(\Pi, e) \leq 16 \gamma_{approxAPSP} \Delta \log n$. The lemma then follows by summing them up. 

    The total number of edges adjacent to $v$ at layer $i$ is at most $32 \cdot 2^i$, since that is the maximum number of edges adjacent to all vertices with the same master vertex $\master(v)$. They can be congested up to $\gamma_{approxAPSP} \Delta \log n/2^i$. Therefore the total congestion is at most $32\gamma_{approxAPSP} \Delta \log n$. 
\end{proof}

Next, we show that the total recourse of our incremental spanner algorithm is low. 

\begin{lemma}
    \label{lem:recourse_spanner}
    The total recourse of $\textsc{IncrementalSpanner}()$ is bounded by $\tilde{O}(n)$.
\end{lemma}
\begin{proof}
    We show the recourse of each layer separately. Excluding vertex splits, the total number of insertions to $H_i$ is at most $10n$ by the bound $9n$ on the final size of the spanner and since vertex splits don't reduce the number of edges in $H_i$ and there are at most $n$ deletions. Therefore, the only thing we have to analyse are the simulation of the vertex splits. To do so, consider the following potential $\Phi_i = \sum_{v \in V_{H_i}} \deg_{H_i}(v) \log \deg_{H_i}(v)$. Whenever a new edge is inserted to $H_i$ the potential increases by at most $\log n$. Therefore the total potential increase is at most $9n \log n$. Whenever a vertex $v$ of current degree $D$ is split into the old vertex $v$ and a new vertex $v'$ such that $v'$ is thereafter incident to $R$ edges, we have that
    \begin{align*}
        (D - R) \log (D - R) + R \log R \leq (D - R) \log D + R (\log D - 1)) \leq D \log D - R. 
    \end{align*}
    Therefore, splitting off $R$ edges decreases the potential by $R$, and the total number of edges split off during $n$ vertex splits can be at most $\tilde{O}(n)$. 
    We obtain our result by summing up the changes over all layers.
\end{proof}

\begin{lemma}
    \label{lem:runtime_spanner}
    The total runtime of \textsc{IncrementalSpanner}() is $O(n \Delta \gamma_{rtSpanner})$ for $\gamma_{rtSpanner} = e^{O(\log^{20/21} m \log \log m)}$.
\end{lemma}
\begin{proof}
    We analyze the runtime of each layer separately. Given the upper bound $\Delta$ on the vertex degree, each vertex split can be implemented via at most $\Delta$ insertions and deletions. Since the number of insertions (and deletions to $\widehat{H}_i$) for each layer is bounded by $O(n \Delta)$ as well, the total runtime per layer is at most $O(n \Delta (\gamma_{\mathrm{approxAPSP}} + \gamma_{\mathrm{timeAPSP}}))$. We choose $\gamma_{\mathrm{rtSpanner}} = 10 \log n(\gamma_{\mathrm{approxAPSP}} + \gamma_{timeAPSP}) = e^{O(\log^{20/21} m \log \log m)}$. 
\end{proof}

We finally prove our main result by assembling the lemmas. 

\begin{proof}[Proof of \Cref{thm:incremental_spanner}]
    The theorem follows from \Cref{lem:size_and_reduction_spanner,lem:cong_spanner,lem:runtime_spanner,lem:recourse_spanner}, since $H_{10 \log n}$ is an empty graph throughout because it contains at most $\Delta n/2^{10 \log n} < 1$ edges by \Cref{lem:size_and_reduction_spanner} where we chose $\gamma_{\mathrm{incSpanner}} = \gamma_{\mathrm{rtSpanner}} + 2 \gamma_{\mathrm{approxAPSP}} \log n = e^{O(\log^{20/21} m \log \log m)}$. Notice that the length of the embedding paths is bounded by construction.
\end{proof}

\subsection{Fully Dynamic Spanner}

Our fully dynamic spanner uses a batching scheme based on the incremental spanner presented in the previous section and the decremental spanner of \cite{chen2022maximum, detMaxFlow}.

\paragraph{Preliminaries: Deterministic Vertex Congestion Spanner.} 

\begin{theorem}[See Theorem 8.1 in \cite{detMaxFlow}]
    \label{thm:batch_spanner}
    Given undirected, unweighted graphs $H$ and $J$ with $V_J \subseteq V_H$ and an embedding $\Pi_{J \mapsto H}$ from $J$ into $H$. Then there is a deterministic algorithm $\textsc{Sparsify}(H, J, \Pi_{J \mapsto H})$ that returns a sparsifier $\widetilde{J}$ with an embedding $\Pi_{J \mapsto \widetilde{J}}$ such that for $\gamma_c, \gamma_l = e^{O(\log^{2/3} m \log \log m)}$ we have
    \begin{enumerate}
        \item $\vcong(\Pi_{J \mapsto H} \circ \Pi_{J \mapsto \widetilde{J}}) \leq \gamma_c \cdot (\vcong(\Pi_{J \mapsto H}) + \Delta_{J})$ where $\Delta_{J}$ is the maximum degree of graph $J$, and 
        \item $\length(\Pi_{J \mapsto H} \circ \Pi_{J \mapsto \widetilde{J}}) \leq \gamma_{\ell} \cdot \length(\Pi_{J \mapsto H})$, and
        \item $|E(\widetilde{J})| = \tilde{O}(|V(J)| \gamma_{\ell})$.
    \end{enumerate}
    The algorithm runs in time $\tilde{O}(|E(J)| \gamma_{\ell}^2 \length(\Pi_{J \mapsto H}))$.
\end{theorem}

\paragraph{Overview.} Our fully-dynamic spanner algorithm (See \Cref{alg:fully_dyn_spanner}) is obtained via a twist on the batching scheme used in \cite{chen2022maximum, detMaxFlow}. Notice that the spanner algorithms of \cite{chen2022maximum, detMaxFlow} are decremental in nature, and therefore only directly extend to the fully dynamic setting when the number of insertions is small enough to allow all of them to be added to the spanner. Since we have to be able to deal with a much larger number of insertions we use the incremental spanner developed in the previous section to reduce to a setting comparable to \cite{chen2022maximum, detMaxFlow}. Both the description of our algorithm and the analysis closely follows section 5 in \cite{chen2022maximum}. Without loss of generality, we assume that the graph $G = (V,E)$ is initially empty. 

\paragraph{Data structures.} Our algorithm maintains graphs $H_{-1}, \ldots H_{K}$ for $K = O(\log^{1/21} m)$ and the maintained spanner is given by $H = \bigcup_{i = {-1}}^{K} H_i$ throughout. It further maintains embeddings $\Pi_{-1}, \ldots, \Pi_K$ so that $\Pi_j$ maps a subset of $E$ into the graph $H_{\leq j} := \bigcup_{i = -1}^{j} H_i$. Since the domains of the embeddings $\Pi_j$ are not disjoint, we let $\Pi_{\leq j}$ denote the embedding that maps every edge $e$ via the embedding $\Pi_j$ with highest index $j$ that has $e$ in its pre-image. Throughout, we will ensure that $\Pi_{\leq j}$ embeds into $H_{\leq j}$ albeit possibly making use of broken paths, and $\Pi_{\leq K}$ embeds into $H$ using only proper paths. 

Further, we maintain sets $S_{-1}, \ldots, S_K$ of vertices touched by deletions and vertex splits, which we formally define next.   

\begin{definition}[Definition 5.2 in \cite{chen2022maximum}]
    \label{def:touched}
    We say that the $t$-th update touches a vertex $v$ if it is an edge deletion of an edge incident to $v$ or it is a vertex split and $v$ is one of the vertices resulting from the split. 
\end{definition}

The graph $H_{-1}$ is a spanner maintained via the data structure $\mathcal{D}_{-1} \gets \textsc{IncrementalSpanner}()$, and the graphs $H_0, \ldots, H_K$ are periodically recomputed using $\textsc{Sparsify}()$ (\Cref{thm:batch_spanner}). The layers with higher indices are recomputed more often than the layers with lower indices, but we ensure that their progressively smaller size makes these rebuilds computationally cheaper to perform. We next describe how our algorithm reacts to updates. 

\paragraph{Insertions.} See \Cref{ln:spanner_insert} in \Cref{alg:fully_dyn_spanner} for pseudocode. Given the insertion of an edge $e$, we simply update $\mathcal{D}_{-1}.\textsc{InsertEdge}
(e)$ and thus $H_{-1}$. This changes the embedding $\Pi_{-1}$ maintained by $\mathcal{D}_{-1}$, but does not cause any changes to the layers $0, \ldots, K$ and the embeddings $\Pi_0, \ldots \Pi_K$. 

\paragraph{Deletions and vertex splits. } See \Cref{ln:spanner_delete} in \Cref{alg:fully_dyn_spanner} for pseudocode. Deletions and vertex splits are passed to all layers $-1, \ldots, K$. Layer $-1$ is updated via $\textsc{DeleteEdge}()/\textsc{Split}()$, whereas the graphs $H_0, \ldots, H_{K}$ are explicityly updated. Vertex splits are implemented by copying the smaller side to a new vertex. Then, some layer gets re-built to repair the broken embedding paths. To describe the repairing, we first define edge embedding projections as in \cite{chen2022maximum}. 

\begin{definition}[Edge-embedding projection, see definition 5.4 in \cite{chen2022maximum}]
    For $i = 0, \ldots, K$ and an edge $(u,v) \in E$ such that $\Pi_{\leq i - 1}(e) \cap S_{i - 1} \neq \emptyset$ we let $\eproj_{i - 1}(e) = \hat{e} \defeq (\hat{u}, \hat{v})$ be a new edge associated with $e$. Its endpoints are obtained by taking the closest vertices in $S_{i - 1}$ to $u$ and $v$ respectively in the graph $G_{\Pi_{\leq i - 1}(e)}$, where $G_{\Pi_{\leq i - 1}(e)}$ is obtained by taking the embedding path of edge $e$ and performing all the splits and edge deletions on it that happened since it was created. 
\end{definition}

After the $t$-th split/deletion, we search for the least index $j$ such that $t$ is divisible by $n^{(1 - j)/K}$. Then, we re-set all data-structures at layers $j, \ldots K$ to be empty, and let $S_{j - 1}$ denote the set of vertices that have been touched (by a vertex split or deletion) since the last time layer $j - 1$ got re-built. We then use $\eproj_{j - 1}(\cdot)$ to project all edges with broken paths to edges on $S_j$ where $\eproj_{j - 1}(\cdot)$ picks the first touched vertex on both sides of the broken embedding path. We denote this projected graph as $J$. 

Using \Cref{thm:batch_spanner}, we then compute a low vertex congestion sparsifier \[ \widetilde{J} \gets \textsc{Sparsify}(H_{< j} \cup E_{\mathrm{affected}}, J, \Pi_{J \mapsto H_{< j} \cup E_{\mathrm{affected}}}) \] and
repair the broken paths with $\Pi_{j}(e) \gets \Pi_{< j}(e)[v, \hat{v}] \concat [\Pi_{J \mapsto H_{< j} \cup E_{\mathrm{affected}}} \circ \Pi_{J \mapsto \widetilde{J}}](\hat{e}) \concat \Pi_{< j}(e)$
where $\Pi_{J \mapsto H_{< j} \cup E_{\mathrm{affected}}}$ denotes the mapping of edges in $J$ to $G$ obtained via $\Pi_{J \mapsto H_{< j} \cup E_{\mathrm{affected}}}(e) \gets \Pi_{<j}(e)[\hat{u}, u] \concat e \concat \Pi_{<j}(e)[v, \hat{v}]$. We then add the pre-image of all edges in $\widetilde{J}$ to $H_j$. 

\begin{algorithm}
\SetKwProg{myalg}{Procedure}{}{}
\myalg{$\textsc{Init}(n, \Delta)$}{
    $\mathcal{D}_{-1} \gets \textsc{IncrementalSpanner}()$; $\mathcal{D}_{-1}.\textsc{Init}(n, \Delta)$ \\
    $t \gets 0$ 
}
\myalg{$\textsc{InsertEdge}(e = (u,v))$}{ \label{ln:spanner_insert}
    $\mathcal{D}_{-1}.\textsc{InsertEdge}(e)$
}
\myalg{$\textsc{DeleteEdge}(e')$/$\textsc{Split}(v', E_1, E_2)$}{ \label{ln:spanner_delete}
    $t \gets t + 1$  \\
    \tcp{timestep $t'$}
    Update $H_{-1}$ with the update via $\mathcal{D}_{-1}.\textsc{DeleteEdge}()$ or $\mathcal{D}_{-1}.\textsc{Split}()$ depending on the type of the update. \\ 
    Update $H_0, \ldots, H_K$ and $S_{-1}, \ldots, S_K$ with the update. \\
    $j \gets \min\{j' \in \Z_{\geq 0}|t \mathrm{ is divisible by } n^{(1 - j')/K}\}$ \label{ln:min_j_spanner} \\
    $t_{j - 1} \gets \floor{t/n^{1 - (j - 1)/K}} \cdot n^{1 - (j - 1)/K}$ \\
    \tcp{Recompute layers $j, \ldots K$} 
    \For{$i = j, \ldots K$}{
        $H_i \gets \emptyset$; $\Pi_{i} \gets \emptyset$; $S_i \gets \emptyset$
    }
    $J \gets (S_{j - 1}, \emptyset)$ \\
    $E_{\mathrm{affected}} \gets \{e \in E| \Pi_{< j}(e) \cap S_{j - 1} \neq \emptyset \} $ \\
    $\Pi_{J \mapsto H_{< j} \cup E_{\mathrm{affected}}} \gets \emptyset$ \\
    \ForEach{$e = (u,v) \in E_{\mathrm{affected}}$}{
        $\hat{e} = (\hat{u}, \hat{v}) \gets \eproj_{i - 1}(e)$ \label{ln:hat_e}\\
        Add $\hat{e}$ to $J$. \\
        $\Pi_{J \mapsto H_{< j} \cup E_{\mathrm{affected}}}(\hat{e}) \gets \Pi_{<j}(e)[\hat{u}, u] \concat e \concat \Pi_{<j}(e)[v, \hat{v}]$. 
    }
    $(\widetilde{J}, \Pi_{J \mapsto \widetilde{J}}) \gets \textsc{Sparsify}(H_{< j} \cup E_{\mathrm{affected}}, J, \Pi_{J \mapsto H_{< j} \cup E_{\mathrm{affected}}})$ \\
    \ForEach{$\hat{e} \in \widetilde{J}$}{Add $e$ to $H_j$ \label{ln:add_tilde_J}}
    \ForEach{$e = (u,v) \in E_{\mathrm{affected}}$}{
        $\hat{e} = (\hat{u}, \hat{v}) \gets \eproj_{j - 1}(e)$ \\
        $\Pi_{j}(e) \gets \Pi_{< j}(e)[v, \hat{v}] \concat [\Pi_{J \mapsto H_{< j} \cup E_{\mathrm{affected}}} \circ \Pi_{J \mapsto \widetilde{J}}](\hat{e}) \concat \Pi_{< j}(e)[\hat{u}, u]$
    }
    \tcp{timestep $t$}
}
\caption{$\textsc{DynamicSpanner}()$}
\label{alg:fully_dyn_spanner}
\end{algorithm}

\paragraph{Proof of \Cref{thm:newSpanner}. } Our proof follows section 5.1 in \cite{chen2022maximum}, replacing the randomized sparsification procedure with its deterministic version developed in \cite{detMaxFlow}.

For the purpose of analysis, we let $X^{(t)}$ denote the variable $X$ after the $t$-th deletion/split. Notice that insertions between $t$ and $t + 1$ are only handled by $\mathcal{D}_{-1}$ and do not cause any changes to the layers $0, \ldots, K$. We first establish that $\Pi_{\leq K}$ embeds every edge to a path in $H$ throughout.  

\begin{lemma}[See Lemma 5.6 in \cite{chen2022maximum}]
    \label{lem:proper_spanner}
    For $i = 0, \ldots, K$ and $t$ divisible by $n^{1 - i/K}$, $\Pi^{(t)}_{\leq i}$ embeds $G^{(t)}$ into $H^{(t)}_{\leq i}$. In particular, at any stage $t$, $\Pi_{G \mapsto H}^{(t)} = \Pi_{\leq K}$ embeds $G^{(t)}$ into $H^{(t)}$. The property additionally holds for the intermediate graphs between deletion/split $t$ and $t + 1$. 
\end{lemma}
\begin{proof}
    The proof is by induction on $t$. If there have not been any deletions/splits $(t = 0)$, then $\Pi_{-1}$ embeds $E$ into $G^{(0)}$ by \Cref{thm:incremental_spanner}. All other $\Pi_i$ are initially empty and therefore the lemma holds for $t = 0$. 
    Now assume the lemma holds for all $t' < t$. We let $j \gets j^{(t)}$, $t_{j - 1} = t_{j - 1}^{(t)}$ and $t'$ be the moment in time right before the $t$-th deletion/split (as annotated in \Cref{alg:fully_dyn_spanner}). We first show that edge $e$ had a valid embedding path in $\Pi_{< j}$ at some point between timesteps $t_{j - 1}$ and $t'$.
    We consider two cases that together establish this claim. 
    \begin{enumerate}
        \item \underline{$e \in E^{(t_{j - 1})}$:} In this case $\Pi^{(t_{j - 1})}_{< j}(e)$ contains a valid embedding to $H^{(t_{j - 1})}_{< j}$ by the induction hypothesis. 
        \item \underline{$e \notin E^{(t_{j - 1})}$:} In this case the edge was inserted between time $t_{j - 1}$ and $t'$. Then, $\Pi_{-1}$ contained a valid embedding of $e$ after it was inserted. 
    \end{enumerate}
    By definition of $S_{j - 1}$, we have that if $\Pi_{<j} \cap S_{j - 1} = \emptyset$ then the path $\Pi_{< j}(e)$ still exists in $H^{(t)}$, i.e. if no deletion/vertex split touches the embedding path it is still valid. 

    Therefore, we consider the case where $\Pi_{<j}(e) \cap S_{j - 1} \neq \emptyset$, i.e. a deletion/vertex split touches the path. Then $\hat{e} = \eproj_{j - 1}(e)$ in \Cref{ln:hat_e} and thus $\hat{e}$ is added to $J$. We claim that the edge $\hat{e}$ maps to a valid path in $H^{(t)}_{\leq j}$ via $[\Pi_{J \mapsto H_{< j} \cup E_{\mathrm{affected}}} \circ \Pi_{J \mapsto \widetilde{J}}](\hat{e})$. Firstly, the map $\Pi_{J \mapsto \widetilde{J}}$ maps each edge $\widehat{e}$ in $J$ to a valid embedding path in $\tilde{J}$. Then, the map $\Pi_{J \mapsto H_{< j} \cup E_{\mathrm{affected}}}$ maps these paths to a valid path in $H_{< j}$ since all the edges it uses are pre-images of edges in $\tilde{J}$ and thus added in Line \Cref{ln:add_tilde_J}.
    
    This implies that the whole path $\Pi_{< j}(e)[v, \hat{v}] \concat [\Pi_{J \mapsto H_{< j} \cup E_{\mathrm{affected}}} \circ \Pi_{J \mapsto \widetilde{J}}](\hat{e}) \concat \Pi_{< j}(e)$ is in $H^{(t)}_{\leq j}$ by the definition of $S_{j - 1}$ and $\eproj_{j - 1}(e)$.
    The claim follows since $\Pi^{(t)}_{j + 1}, \ldots \Pi^{(t)}_{K} = \emptyset$, and insertions between $t$ and $(t + 1)'$ (i.e. until the start of the processing of the next deletion/split $t + 1$) have a valid embedding path in $\Pi_{-1}^{(t + 1)'}$ by \Cref{thm:incremental_spanner}. 
\end{proof}

Now that we established that the embedding $\Pi_{\leq K}$ is proper, we bound the vertex congestion and embedding path length. 

\begin{claim}
    \label{clm:spanner_len_cong}
    For $i = 0, \ldots, K$ we have 
    \begin{enumerate}
        \item $\length(\Pi_{\leq i}) \leq 2^{i + 1} \gamma_l^{i + 1} \gamma_{\mathrm{incSpanner}}$ and
        \item $\vcong(\Pi_{\leq i}) \leq 8^{i + 1} \gamma_{c}^{i + 1} \gamma_{\mathrm{incSpanner}} \Delta$
    \end{enumerate}
    throughout. 
\end{claim}
\begin{proof}
    We prove the items separately. 
    \begin{enumerate}
        \item Between times that $\Pi_i$ is recomputed the embedding path $\Pi_i(e)$ remains fixed. We therefore focus on bounding the length when a re-computation of layer $i$ happens.
        
        The proof is by induction on $t$. We have that $\length(\Pi_{-1}) \leq \gamma_{\mathrm{incSpanner}}$ throughout by \Cref{thm:incremental_spanner} and the description of our algorithm. Since $\Pi_0, \ldots, \Pi_K = \emptyset$ for $t = 0$ the base case follows.
        
        We let $j \gets j^{(t)}$, $t_{j - 1} = t_{j - 1}^{(t)}$ and $t'$ be the moment right before the $t$-th deletion/split. By the induction hypothesis, we have that $\length(\Pi_{< j}^{t_{j - 1}}) \leq 2^{j} \gamma_l^{j} \gamma_{\mathrm{incSpanner}}$ since $t_{j - 1} < t_j$ by the minimality of $j$ in \Cref{ln:min_j_spanner}. Therefore, we have length $\length(\Pi_{< j}^{t'}) \leq 2^{j} \gamma_l^{j} \gamma_{\mathrm{incSpanner}}$ by \Cref{thm:incremental_spanner} and the description of our algorithm since all newly embedded edges since $t_{j - 1}$ have embedding length at most $\gamma_{\mathrm{incSpanner}}$ by \Cref{thm:incremental_spanner}. 
        
        The segments $\Pi_{<j}(e)[\hat{u}, u]$ and $\Pi_{<j}(e)[v, \hat{v}]$ contain at most $2^{j} \gamma_l^{j} \gamma_{\mathrm{incSpanner}}$ edges by the induction hypothesis. It remains to bound the length of the segments $[\Pi_{J \mapsto H_{< j} \cup E_{\mathrm{affected}}} \circ \Pi_{J \mapsto \widetilde{J}}](\hat{e})$. We observe that its length is at most $\gamma_l \cdot \length(\Pi_{J \mapsto H_{< j} \cup E_{\mathrm{affected}}}) \leq 2^j \gamma_l^{j + 1} \gamma_{\mathrm{incSpanner}}$ by \Cref{thm:batch_spanner}. The claim follows by adding up the segments. 

        \item We next show the second item, again by induction on the number of deletions/vertex splits $t$. We have that $\vcong(\Pi_{\leq -1}) \leq \gamma_{\mathrm{incSpanner}} \Delta$ throughout by \Cref{thm:incremental_spanner}. Therefore the claim follows for $t = 0$ up until right before the first edge deletion/split since then $\Pi_0, \ldots, \Pi_K = \emptyset$. 

        We let $j \gets j^{(t)}$, $t_{j - 1} = t_{j - 1}^{(t)}$ and $t'$ be the moment right before the $t$-th deletion/split. Then, by the induction hypothesis we have that $\vcong(\Pi_{< j}^{(t_{j - 1})}) \leq 8^j \gamma_c^j \gamma_{\mathrm{incSpanner}} \Delta$. Therefore, the vertex congestion at moment $t'$ is bounded by $\vcong(\Pi_{< j}^{(t')}) \leq 8^j \gamma_c^j \gamma_{\mathrm{incSpanner}} \Delta + \gamma_{\mathrm{incSpanner}} \Delta \leq 2 \cdot 8^j \gamma_c^j \gamma_{\mathrm{incSpanner}} \Delta$ by \Cref{thm:incremental_spanner} since only the embedding $\Pi_{-1}$ changed in $\Pi_{< j}$ between $t_{j - 1}$ and $t'$.

        By the description of our algorithm, the vertex congestion   $\vcong(\Pi_{< j}^{(t')})$ upper bounds the degree of $J$. By \Cref{thm:batch_spanner}, we have
        \begin{equation*}
            \vcong(\Pi_{J \mapsto H_{< j} \cup E_{\mathrm{affected}}} \circ \Pi_{J \mapsto \widetilde{J}}) \leq \gamma_c(\vcong(\Pi_{< j}) + \Delta_J) \leq 4 \cdot 8^j \gamma_c^j \gamma_{\mathrm{incSpanner}} \Delta. 
        \end{equation*}
        Adding the extra congestion caused by segments $\Pi_{<j}(e)[\hat{u}, u]$ and $\Pi_{<j}(e)[v, \hat{v}]$ we obtain that 
        \begin{equation*}
             \vcong(\Pi^{(t)}) \leq 5 \cdot  8^j \gamma_c^j \gamma_{\mathrm{incSpanner}} \Delta. 
        \end{equation*}
    \end{enumerate}
    Since the extra congestion added in-between timesteps is at most an additive $\gamma_{\mathrm{incSpanner}} \Delta$ the claim follows. 
\end{proof}

We finally establish the runtime and recourse bounds. 

\begin{lemma}
    \label{lem:runtime_recourse_spanner}
    At any stage $H$ consist of $\tilde{O}(\gamma_l \cdot n) = O(n \cdot e^{O(\log^{1/2} m \log \log m)})$ edges and the amortized number of changes to $H$ is at most $\tilde{O}(n^{1/K}) = e^{O(\log^{20/21} m)}$ given at least $n$ insertions. The total runtime of the algorithm after $k$ insertions and $t$ deletions/splits is $\tilde{O}(n^{1/K}\gamma^2_{incSpanner}(\gamma_c\gamma_l)^{O(K)} \Delta t + \gamma_{\mathrm{incSpanner}} \Delta k) \leq (t + k) \cdot e^{O(\log^{20/21} m \log \log m)}$. Further, the total number of embedding changes is $(t + k) \cdot e^{O(\log^{20/21} m \log \log m)}$. 
\end{lemma}
\begin{proof}
    Firstly $H_{-1}$ contains at most $\tilde{O}(n)$ edges throughout, and the total recourse of $H_{-1}$ is $\tilde{O}(n)$. The total runtime of layer $-1$ is at most $\tilde{O}(n\gamma_{\mathrm{incSpanner}}\Delta)$. 

    We therefore focus on the layers $0, \ldots, K$. 
    
    We first show sparsity. Whenever the layer $i$ gets rebuilt, we have that $S_{i - 1}$ is of size at most $O(n^{1 - (i - 1)/K})$. Since the spanner $\widetilde{J}$ built over $S_{i - 1}$ is sparse, the number of edges is $\tilde{O}(\gamma_l |S_{j - 1}|)$ by \Cref{thm:batch_spanner}. The bounds on the sparsity follow, and the recourse caused by rebuilds is at most $\tilde{O}(n^{1/K})$ via a simple argument over the frequency of the rebuilds. 

    We additionally have to bound the recourse caused by vertex splits in-between rebuilds. Whenever a vertex is split, we copy the smaller side over. We do so by layer, and notice that every edge can be on the smaller side at most $\log n$ times before the next rebuild. Therefore, the total recourse is $\tilde{O}(K \cdot n^{1/K}) = \tilde{O}(n^{1/K})$.

    We then bound the running time. Every time a layer $i$ gets recomputed, the size of $J$ is at most $O(8^i \gamma_c^i \gamma_{\mathrm{incSpanner}} \Delta)$ by \Cref{clm:spanner_len_cong} and the length of the embedding $\Pi_{J \mapsto H_{< j} \cup E_{\mathrm{affected}}}$ is at most $2^{i} \gamma_l^{i} \gamma_{\mathrm{incSpanner}}$ by \Cref{clm:spanner_len_cong}.  Therefore, the total runtime spent for a rebuild of layer $i$ is bounded by 
    \begin{equation*}
        \tilde{O}(8^i \gamma_c^i \gamma_{\mathrm{incSpanner}}|S_{i - 1}| \gamma_l^2 2^{i} \gamma_l^{i} \gamma_{\mathrm{incSpanner}})
    \end{equation*}
    and the runtime bound follows since the cost of splits in-between re-computations is proportional to the recourse. Each embedding change can be directly attributed to runtime, and thus analysed in the same way. 
\end{proof} 

\begin{proof}[Proof of \Cref{thm:newSpanner}]
It follows from \Cref{lem:proper_spanner}, \Cref{clm:spanner_len_cong} and \Cref{lem:runtime_recourse_spanner}. 
\end{proof}
\section{Minimum Cost Flow IPM}
\label{sec:ipm}
In this section, we describe the min-ratio cycle framework of \cite{chen2022maximum} for solving min-cost flow and its extension to incremental graphs by \cite{vdBrand23incr}. Finally, this yields a proof of \Cref{thm:static_ipm} and \Cref{thm:inc_to_solver}. 

In the following, we consider the problem of finding a min-cost circulation. This is equivalent to min-cost flow via adding high cost edges between sources and sinks, which is a standard reduction. Further, we focus on the proof of \Cref{thm:inc_to_solver}, since \Cref{thm:static_ipm} directly follows from \Cref{thm:inc_to_solver} via a binary search over the threshold $F$.

\subsection{Minimum Cost Flow via Min-Ratio Cycles}

We recall the static min-ratio cycle approach to min-cost flow of \cite{chen2022maximum} that was adapted to incremental graphs by \cite{vdBrand23incr}.

For an incremental graph $G = (V, E, \uu, \cc)$ with polynomially bounded capacities $\uu(e) \in [1, U]$ and costs  $\cc(e) \in [-C, C]$ we recall the potential 
\begin{equation}
    \label{def:pot}
    \Phi(\ff) := 20m \log (\cc^\top  \ff - F) + \sum_{e \in E}(\ff(e) + \delta)^{-\alpha} + (\uu(e) - \ff(e))^{-\alpha}
\end{equation}
where $m$ is the total number of edge insertions $G$, $\delta := 1/(20m^2C)$, $\alpha := 1/(5000 \log (m C U))$ and $-\delta \leq \ff(e) \leq \uu(e)$ from \cite{vdBrand23incr}. This is an modification of the potential introduced by \cite{chen2022maximum} in the context of static min-cost flow. The additional additive factor $\delta$ allows the addition of new edges with flow $0$ without increasing the potential by more than a constant. 

\begin{lemma}[See Lemma 4.7 in \cite{vdBrand23incr}]
    \label{lem:phi_inc}
    Adding an edge $e$ with some capacity $\uu(e) \leq U$ and cost $\cc(e) \leq C$ to the graph and setting $\ff(e) = 0$ increases the potential by $O(1)$.
\end{lemma}
\begin{proof}
    Adding an edge increases the potential by $\delta^{-\alpha} + \uu(e)^{-\alpha} \leq 3$.
\end{proof}

Next, we define the lengths and gradients as in \cite{vdBrand23incr}.  

\begin{definition}[Lengths \& gradients]
    Given the potential $\Phi(\ff)$ for a graph $G = (V, E, \uu, \cc)$ as in \Cref{def:pot}, we define the \emph{length} vector $\ll \in \R^{|E|}$ with 
    \begin{equation*}
        \ll(e) = (\ff(e) + \delta)^{-1} + (\uu(e) - \ff(e))^{-1}
    \end{equation*}
    for every $e \in E$ and the \emph{gradient} vector $\gg \in \R^{|E|}$ as
    \begin{equation*}
        \gg(e) := 20m(\cc^\top \ff - F)^{-1} \cc(e) + \alpha (\uu(e) - \ff(e))^{-1-\alpha} - \alpha(\ff(e) + \delta)^{-1 - \alpha}.
    \end{equation*}
    We further let $\LL = \diag(\ll)$.
\end{definition}

Next, we state a lemma that shows that a good quality min-ratio cycle always exists if there exists a flow of cost at most $F$. As in \cite{chen2022maximum, vdBrand23incr}, we will crucially instantiate our solver data structure \Cref{def:solver_ds} with approximate gradients and lengths $\tilde{\ll}$ and $\tilde{\gg}$. This is necessary, since the exact gradients and lengths can change very frequently. The following lemma bounds the quality of the min-ratio cycle for a sufficiently close approximation of the gradients and lengths.  

\begin{lemma}[Lemma 4.5 in \cite{vdBrand23incr}, See Lemma 4.7 in \cite{chen2022maximum}]
    \label{lem:quality}
    Let $G$ be a graph such that there is a feasible circulation $\ff^\star$ with $\cc^\top \ff^\star \leq F$. Let $\widetilde{\gg} \in \R^{|E|}$ satisfy $\norm{\LL^{-1}(\widetilde{\gg} - \gg)}_{\infty} \leq \epsilon$ for some $\epsilon \leq \alpha/2$ and $\widetilde{\ll} \approx_2 \ll$. If $\Phi(\ff) \leq 200m \log mCU$ and $\log(\cc^\top \ff - F) \geq -10 m \log m C U$, then 
    \begin{equation*}
        \frac{\widetilde{\gg}^\top(\ff^\star - \ff)}{\norm{\widetilde{\LL}(\ff^\star - \ff)}_1} \leq -\alpha/4.
    \end{equation*}
\end{lemma}

The next lemma shows that a suitably scaled min-ratio cycle reduces the potential by an amount that is comparable to its quality. 

\begin{lemma}[Lemma 4.4 in \cite{chen2022maximum}]
    \label{lem:progress}
    Let $\widetilde{\gg} \in \R^{|E|}$ satisfy $\norm{\LL^{-1} (\widetilde{\gg} - \gg)}_{\infty} \leq \kappa/8$ for some $\kappa \in (0, 1)$, and $\widetilde{\ll} \in \R^{|E|}_{\geq 0}$ satisfying $\widetilde{\ll} \approx_2 \ll$. Let $\DDelta$ satisfy $\BB^\top \DDelta = 0$ and $\widetilde{\gg}^\top \DDelta / \norm{\LL \DDelta}_1 \leq - \kappa$. Let $\eta$ be such that $\eta \widetilde{\gg}^\top\DDelta = -\kappa^2/50$. Then, $\ff + \eta \DDelta$ is feasible and 
    \begin{equation*}
        \Phi(\ff + \eta \DDelta) \leq \Phi(\ff) - \kappa^2/500
    \end{equation*}
\end{lemma}

To provide a bound on the number of iterations, we state an upper bound the initial potential. 

\begin{lemma}[See Section 4 of \cite{vdBrand23incr}]
    \label{lem:phi_upper}
    $\Phi(\veczero) \leq 100 m \cdot \log(mCU)$
\end{lemma}
\begin{proof}
    We have 
    \begin{equation*}
        \Phi(\veczero) = 20m \log(-F) + \sum_{e \in E} (\uu(e)^{-\alpha} + \delta^{-\alpha}) \leq 100 m \cdot \log(mCU)
    \end{equation*}
    since $|F| \leq mCU$ as otherwise the flow would never be feasible. 
\end{proof}

The next lemma then shows that reducing the potential to $-\tilde{O}(m)$ suffices to obtain a flow sufficiently close to $F$. 

\begin{lemma}[Lemma 4.3 in \cite{vdBrand23incr}]
    \label{lem:phi_lower}
    If $\Phi_G(\ff) \leq -200m \log(mCU)$, then $\cc^\top \ff \leq F + (mCU)^{-10}$.
\end{lemma}

Therefore, given an update $\DDelta$ as in \Cref{lem:progress} we need at most $\tilde{O}(m/\kappa^2)$ updates.

\subsection{Algorithm}

Next, we show that the min-ratio cycle solver developed in \Cref{sec:min_ratio_ds} and summarized in \Cref{def:solver_ds} can be used to extract min-ratio cycles as in \Cref{lem:progress} efficiently to prove \Cref{thm:inc_to_solver}.

Our algorithm uses the framework of \cite{vdBrand23incr}, but replaces the routine for detecting and augmenting along min-ratio cycles. See \Cref{alg:thresh_min_cost_flow} for pseudocode. It is based on the data structure $\mathcal{D} \gets \textsc{Solver}()$ from \Cref{def:solver_ds}. 
We let $\ff$ denote the flow maintained by $\mathcal{D}$, while the approximation that our IPM algorithm maintains is denoted $\bar{\ff}$. We further denote the approximate lengths and gradients passed to $\textsc{Solver}$ as $\widetilde{\ll}$ and $\widetilde{\gg}$.

\paragraph{Initialization.} Given an incremental graph $G$ with edge capacities $\uu$ in $[1, U]$ and costs $\cc$ in $[-C, C]$ and a target cost $F$, it first initializes the flow $\bar{\ff} \gets 0$ and $r \gets F$. Then, it computes the lengths and gradients of all edges in $G$ via 
\begin{equation*}
        \widetilde{\ll}(e) \gets (\bar{\ff}(e) + \delta)^{-1} + (\uu(e) - \bar{\ff}(e))^{-1}
\end{equation*}
and 
\begin{equation*}
    \widetilde{\gg}(e) \gets r/(\cc^\top \bar{\ff} - F)(20m \cc(e)/r + \alpha(\uu(e) - \bar{\ff}(e))^{-1-\alpha} - \alpha(\bar{\ff}(e) + \delta)^{-1-\alpha})
\end{equation*}
with $\alpha = 1/(5000 \log mCU)$. The algorithm further uses the parameters $q = \alpha/4\gamma_{\mathrm{approx}}$, $\Gamma = (\alpha/4\gamma_{\mathrm{approx}})/800$, and $\epsilon = 1/40\gamma_{\mathrm{approx}}$. 

Finally, it initializes the data structure $\mathcal{D} \gets \textsc{Solver}(G, \widetilde{\ll}, \widetilde{\gg}, \cc, \bar{\ff}, q, \Gamma, \epsilon)$. Notice that $\ff = \bar{\ff}$ at this moment. 

\paragraph{Cost minimization and edge insertions.} The data structure repeatedly calls $E', \ff'(E') \gets \mathcal{D}.\textsc{ApplyCycle}()$, and updates the maintained flow $\bar{\ff}(E) \gets \ff'(E')$, as well as the lengths and gradients of the returned edges in $E'$ using $\mathcal{D}.\textsc{UpdateEdge}()$.

Whenever the total flow cost as changed significantly, it fully re-initializes the data structure $\mathcal{D} \gets \textsc{Solver}(G, \widetilde{\ll}, \widetilde{\gg}, \cc, \bar{\ff}, q, \Gamma, \epsilon)$ after reading of the full flow $\ff$ from $\mathcal{D}$ and recomputing all lengths and gradients. 

Whenever $\mathcal{D}.\textsc{ApplyCycle}()$ reports that it didn't find a quality $q$ cycle, the algorithm processes the next insertion. 

When the maintained flow value $\cc^\top  \ff$ in $\mathcal{D}$ is less than $F + (mCU)^{-10}$ the algorithm terminates and returns the flow $\mathcal{D}.\textsc{ReportFlow}()$ after rounding it to an exact solution (See e.g. Section 4 of \cite{kang2015flow} and Lemma 4.1 in \cite{chen2022maximum}).  

See \Cref{alg:thresh_min_cost_flow} for detailed pseudocode of our algorithm. 

\begin{algorithm}
\caption{$\textsc{MinCostFlow}(G = (V, E, \cc, \uu), F)$}
\label{alg:thresh_min_cost_flow}
Let $\alpha = 1/(5000 \log mCU)$, $\gamma_{\mathrm{approx}}$ as in \Cref{def:solver_ds} and $q = \alpha/4\gamma_{\mathrm{approx}}$, $\Gamma = (\alpha/4\gamma_{\mathrm{approx}})/800$, and $\epsilon = 1/40\gamma_{\mathrm{approx}}$. \\
$r \gets -F$ \\
$\bar{\ff}, \widetilde{\gg}, \widetilde{\ll} \gets \veczero$ \\
\For{$e \in E$}{
    $\widetilde{\ll}(e) \gets (\bar{\ff}(e) + \delta)^{-1} + (\uu(e) - \bar{\ff}(e))^{-1}$ \\
    $\widetilde{\gg}(e) \gets r/(\cc^\top \bar{\ff} - F)(20m \cc(e)/r + \alpha(\uu(e) - \bar{\ff}(e))^{-1-\alpha} - \alpha(\bar{\ff}(e) + \delta)^{-1-\alpha})$
}
$\mathcal{D} \gets \textsc{Solver}(G, \widetilde{\ll}, \widetilde{\gg}, \cc, \bar{\ff}, q, \Gamma, \epsilon)$ \\ 
\While{true}{
    \While(\tcp*[f]{Sufficient quality circulation found}){$E', \ff'(E') \gets \mathcal{D}.\textsc{ApplyCycle}()$ finds circulation}{
        Update $\bar{\ff}(E') \gets \ff'(E')$. \\
        \uIf{$\mathcal{D}.\textsc{FlowCost}() \leq r/(1 + \epsilon)$ or $\mathcal{D}.\textsc{FlowCost}() \geq r(1 + \epsilon)$}{
            \tcp{Full recompute of lengths and gradients. }
            $r \gets \mathcal{D}.\textsc{FlowCost}() -F$; $\bar{\ff} \gets \mathcal{D}.\textsc{ReturnFlow}()$ \\
            \ForEach{$e \in E$}{
                $\widetilde{\ll}(e) \gets (\bar{\ff}(e) + \delta)^{-1} + (\uu(e) - \bar{\ff}(e))^{-1}$ \\
                $\widetilde{\gg}(e) \gets r/(\cc^\top \bar{\ff} - F)(20m \cc(e)/r + \alpha(\uu(e) - \bar{\ff}(e))^{-1-\alpha} - \alpha(\bar{\ff}(e) + \delta)^{-1-\alpha})$
            }
            $\mathcal{D} \gets \textsc{Solver}(G, \widetilde{\ll}, \widetilde{\gg}, \cc, \bar{\ff}, q, \Gamma, \epsilon)$
        }\Else{
            \ForEach{$e \in E'$}{
                $\widetilde{\ll}(e) \gets (\bar{\ff}(e) + \delta)^{-1} + (\uu(e) - \bar{\ff}(e))^{-1}$\\ 
                $\widetilde{\gg}(e) \gets r/(\cc^\top \bar{\ff} - F)(20m \cc(e)/r + \alpha(\uu(e) - \bar{\ff}(e))^{-1-\alpha} - \alpha(\bar{\ff}(e) + \delta)^{-1-\alpha})$ \\ 
                $\mathcal{D}.\textsc{UpdateEdge}(e, \widetilde{\ll}(e), \widetilde{\gg}(e))$.
            } 
        }
        \If{$\mathcal{D}.\textsc{FlowCost}() \leq F + (mCU)^{-10}$}{
            \Return $\mathcal{D}.\textsc{ReturnFlow}()$ \tcp*{Round flow to exact solution}
        }
    }
    Process the next insertion of edge $e$ with capacity $\uu(e)$ and cost $\cc(e)$ to $G$. \\
    $\bar{\ff}(e) \gets 0$ \\
    $\widetilde{\ll}(e) \gets (\bar{\ff}(e) + \delta)^{-1} + (\uu(e) - \bar{\ff}(e))^{-1}$ \\
    $\widetilde{\gg}(e) \gets r/(\cc^\top \bar{\ff} - F)(20m \cc(e)/r + \alpha(\uu(e) - \bar{\ff}(e))^{-1-\alpha} - \alpha(\bar{\ff}(e) + \delta)^{-1-\alpha})$ \\
    $\mathcal{D}.\textsc{InsertEdge}(e, \widetilde{\ll}(e), \widetilde{\gg}(e))$
}
\end{algorithm}

\paragraph{Stability of gradients and lengths.} Notice that the data structure $\mathcal{D}$ in $\textsc{MinCostFlow}(G, F)$ does not maintain the exact gradients $\gg$ and $\ll$, but maintains a synchronization with some slack to guarantee efficiency. We will refer to the internal lengths and gradients in $\mathcal{D}$ as $\widetilde{\ll}$ and $\widetilde{\gg}$ respectively. We then show that we update $\widetilde{\ll}$ and $\widetilde{\gg}$ sufficiently regularly to achieve the preconditions of \Cref{lem:quality} and \Cref{lem:progress}. 

We state three lemmas from \cite{vdBrand23incr}. They show that the residuals, lengths and gradients are sufficiently stable. 

\begin{lemma}[Lemma 5.24 in \cite{vdBrand23incr}]
    \label{lem:residual_stab}
    Let $\widetilde{\gg} \in \R^{|E|}$ satisfy $\norm{\LL^{-1}(\widetilde{\gg} - \gg)}_{\infty} \leq \epsilon$ for some $\epsilon \in (0, 1/2]$, and let $\widetilde{\ll} \in \R^{|E|}_{\geq 0}$ satisfy $\tilde{\ll} \approx_{2} \ll$. Let $\DDelta$ satisfy $\BB^\top \DDelta = \veczero$ and $\widetilde{\gg}^\top\DDelta/\norm{\widetilde{\LL} \DDelta}_1 \leq - \kappa$ for $\kappa \in (0, 1)$. Then 
    \begin{equation*}
        \frac{|\cc^\top  \DDelta|}{\cc^\top  \ff - F} \leq |\widetilde{\gg}^\top \DDelta|/(\kappa m)
    \end{equation*}
\end{lemma}

\begin{lemma}[Lemma 5.25 in \cite{vdBrand23incr}]
    \label{lem:length_stab}
    If $\norm{\LL(\ff - \bar{\ff})}_{\infty} \leq \epsilon$ for some $\epsilon \leq 1/100$ then the lengths $\ll$ and $\bar{\ll}$ for flow $\ff$ and $\bar{\ff}$ respectively satisfy $\ll(e) \approx_{1 + 3 \epsilon} \bar{\ll}(e)$
\end{lemma}

\begin{lemma}[Lemma 5.26 in \cite{vdBrand23incr}]
    \label{lem:grad_stab}
    If $\norm{\LL(\ff - \bar{\ff})}_{\infty} \leq \epsilon$ and $r \approx_{1 + \epsilon} \cc^\top  \bar{\ff} - F$ and  
    \begin{equation*}
        \widetilde{\gg}(e) = r/(\cc^\top \bar{\ff} - F)(20m \cc(e)/r + \alpha(\uu(e) - \ff(e))^{-1-\alpha} - \alpha(\ff(e) + \delta)^{-1-\alpha})
    \end{equation*}
    for all $e \in E$ satisfies
    \begin{equation*}
        \norm{\LL^{-1}(\widetilde{\gg} - \gg)}_1 \leq 10\alpha\epsilon. 
    \end{equation*}
\end{lemma}

\paragraph{Proof of \Cref{thm:inc_to_solver}}

We fix $\kappa = q = \alpha/4\gamma_{\mathrm{approx}}$ and $\epsilon = 1/40\gamma_{\mathrm{approx}}$ for this proof. We then show a claim about the approximate lengths and gradients we maintain. 

\begin{claim}
    \label{clm:stab}
    The approximate lengths $\widetilde{\ll}$ and gradients $\widetilde{\gg}$ as maintained by $\mathcal{D}$ satisfy the preconditions of \Cref{lem:quality} and \Cref{lem:progress}, i.e., $\norm{\LL^{-1}(\widetilde{\gg} - \gg)}_{\infty} \leq \min\{\epsilon,\kappa/8\}$ and $\widetilde{\ll} \approx_2 \ll$. 
\end{claim}
\begin{proof}
    Note that throughout, $\bar{\ff}(e)$ is the flow value $\ff(e)$ that was on edge $e$ the last time $e \in E'$ (or after the last re-initialization of $\mathcal{D}$). By the guarantees of $\mathcal{D}.\textsc{ApplyCycle}()$ in \Cref{def:solver_ds}, we have $\norm{\LL(\ff - \bar{\ff})}_1 \leq \epsilon$. Thus, the lemma follows from \Cref{lem:length_stab} and \Cref{lem:grad_stab}.
\end{proof}

\begin{claim}[Calls to a solver]
    \label{clm:calls}
    Our algorithm makes the following number of calls to the solver data structure (\Cref{def:solver_ds}).  
    \begin{enumerate}
        \item The number of calls to $\textsc{ApplyCycle}()$/$\textsc{UpdateEdge}()$/$\textsc{InsertEdge}()/\textsc{ReturnCost}()$ is bounded by $\tilde{O}(m \gamma_{\mathrm{approx}}^{O(1)})$ in total.
        \item The number of calls to $\textsc{ReturnFlow}()$ is bounded by $\tilde{O}(\gamma_{\mathrm{approx}}^{O(1)})$ and the number of solver data structures $\mathcal{D}$ initialized is $\tilde{O}(\gamma_{\mathrm{approx}}^{O(1)})$. 
    \end{enumerate}
\end{claim}
\begin{proof}
    We show the two items separately.
    \begin{enumerate}
        \item Since there are at most $m$ unsuccessful calls to $\mathcal{D}.\textsc{ApplyCycle}()$ because there are $m$ insertions total, and each successful one reduces the potential $\Phi$ by $\Omega(\kappa^2)$ by \Cref{clm:stab} and \Cref{lem:progress}, the bound follows from the upper bound on the potential (\Cref{lem:phi_upper}), its increase due to insertions (\Cref{lem:phi_inc}), and the lower bound on the potential (\Cref{lem:phi_lower}). From these, we conclude that the algorithm terminates after $\tilde{O}(m/\kappa^2) = \tilde{O}(m \gamma_{\mathrm{approx}}^{O(1)})$ calls to $\textsc{ApplyCycle}()$ and thus also $\textsc{ReturnCost}()$.

        After each call to $\textsc{ApplyCycle}()$ some edges $E'$ are returned. The amortized number of returned edges is at most 
        \begin{align*}
            \Gamma/(\epsilon q) = \tilde{O}(\gamma_{\mathrm{approx}}^{O(1)}).
        \end{align*}
        by \Cref{def:solver_ds} and the description of our algorithm. These lead to a call to $\textsc{UpdateEdge}()$ each, and thus to at most $\tilde{O}(m \gamma_{\mathrm{approx}}^{O(1)})$ calls to $\textsc{UpdateEdge}()$ total. Finally, there are at most $m$ insertions, leading to at most $m$ calls $\textsc{InsertEdge}()$.
        \item By \Cref{clm:stab} and \Cref{lem:residual_stab} we have
        \begin{align*}
        \frac{|\cc^\top  \DDelta|}{\cc^\top  \ff - F} \leq |\widetilde{\gg}^\top \DDelta|/(q m) =  \Gamma/(qm) \leq  1/800 m
        \end{align*}
        for each update to the maintained flow $\ff$ in $\mathcal{D}$. From $(1 + 1/(800m))^t \geq (1 + \epsilon)$ we get that at least $t \geq m/\epsilon$ steps happened since the last adjustment. The bound on the number of full rebuilds and calls to $\textsc{ReturnFlow}()$ follows from the total number of iterations as bounded in the previous item, since our algorithm only calls $\textsc{ReturnFlow}()$ whenever a full rebuild happens. 
    \end{enumerate}    
\end{proof}

\begin{claim}[Correctness]
    \label{clm:correct}
    The algorithm returns a feasible flow $\ff$ of cost at most $F$ the first time it exists.
\end{claim}
\begin{proof}
    Firstly, whenever a flow of cost at most $F$ exists, $\mathcal{D}.\textsc{ApplyCycle}()$ finds a cycle of quality less than $-q$ by \Cref{thm:solver_ds} since $\mathrm{OPT} \cdot \gamma_{\mathrm{approx}} \geq q$ by \Cref{lem:quality} and \Cref{clm:stab}. Secondly, the flow remains feasible by \Cref{lem:progress} and \Cref{clm:stab}. The claim follows. 
\end{proof}

\Cref{thm:inc_to_solver} follows from \Cref{clm:calls} and \Cref{clm:correct}. Given a sufficiently accurate flow, we can round it to an exact solution in near linear time by toggling cycles on a dynamic tree (see e.g. Section 4 of \cite{kang2015flow} and Lemma 4.1 in \cite{chen2022maximum}).  

\subsection{Strongly Connected Components}
\label{subsec:scc}

In this section we explain how to use the incremental IPM framework to maintain the strongly connected components (SCCs) in an incremental directed graph, thus showing \Cref{thm:scc}.

The idea is to run the IPM discussed in this section, and then contract edges once their flow values are nontrivial. The correctness of this hinges on the following combinatorial lemma.
\begin{lemma}
\label{lemma:contract}
Let $G$ be a directed graph, $\delta \le \frac{1}{20m^2}$, and let $\ff \in \R^{E(G)}$ be a circulation in $G$ satisfying $-\delta \le \ff(e) \le 1$ for all $e \in E(G)$. If $\ff(e) \ge 1/(10m)$, then $e$ is in a SCC.
\end{lemma}
\begin{proof}
Let $H$ be the condensation of $G$, i.e., $G$ with all SCCs contracted. Let $\ff^H$ be the flow on $H$ which has $\ff^H(e) = \ff(e)$ for all uncontracted edges $e$. It is easy to see that $\ff^H$ is a circulation. We will argue that because $H$ is a DAG, that $\ff^H(e) < 1/(10m)$ on all edges $e$. Indeed, perform a cycle decomposition of $\ff^H$ into at most $m$ cycles. Each cycle has flow at most $\delta$, hence $\ff^H(e) \le \delta m < 1/(10m)$ for all $e$. Hence, any edge with $\ff(e) \ge 1/(10m)$ is in an SCC.
\end{proof}

\paragraph{Algorithm.} Recall that incremental cycle detection is reduced to thresholded mincost flow by setting $\cc(e) = -1$ and $\uu(e) = 1$ for every edge $e$ that arrives, and the desired threshold $F = -1$.
Now run \Cref{alg:thresh_min_cost_flow}, except whenever an edge $e' \in E'$ has $\ff(e') \ge 1/(10m)$ we contract the endpoints of $e'$ to a single vertex. Now, we remove any edges contracted to a self-loop, and otherwise we keep the flow values the same.

\begin{proof}[Proof of \Cref{thm:scc}]
We first discuss the correctness of the algorithm. By \Cref{lemma:contract} we never contract vertices that are not in a real SCC.
By the guarantees of $\mathcal{D}.\textsc{ApplyCycle}()$ we know that every uncontracted edge has $\ff(e) < 1/(5m)$, because the last time $e \in E'$ it satisfied $\ff(e) < 1/(10m)$. Thus, $\cc^\top \ff \ge -1/5$. However, if $H$ has a directed cycle then there exists a circulation $\ff^*$ with $\cc^\top \ff^* \le -1$. Thus \Cref{lem:quality} implies that the algorithm can find a min ratio cycle to make progress.

Now we bound the number of iterations of the algorithm. It suffices to bound the potential increase from contracting edges. Note that $-\log(\ff(e)+\delta)-\log(1-\ff(e)) \ge 0$, so removing $e$ does only decreases this piece of the potential. For the $20m \log(\cc^\top \ff - F) = 20m \log(\cc^\top \ff + 1)$ part, recall that $\cc^\top \ff \ge -1/5$ at all times. Thus, removing an edge $e$ with $\ff(e) \le 1/(5m)$ can only increase that piece of the potential by
\[ 20m\log\left(\frac{\cc^\top \ff + 1 + 1/(5m)}{\cc^\top \ff + 1} \right) \le 20m \cdot \frac{1/(5m)}{\cc^\top \ff + 1} \le 5, \] as $\cc^\top\ff + 1 \ge 4/5$. So the total potential increase from edge contractions is bounded by $5m$.

Finally, we discuss how to implement contractions as operations to the min-ratio cycle data structure of \Cref{thm:solver_ds}. In particular, we never pass edge contractions to the data structure of \Cref{thm:solver_ds}. Instead, if an edge $e$ is contracted, we just treat it as having $\gg(e) = 0$ and $\ll(e) = 1/m^{10}$ from then on. This simulates contracting $e$.
\end{proof}

\pagebreak

\bibliographystyle{alpha}
\bibliography{refs}

\end{document}